%% file: main.tex
\documentclass[11pt]{article}
\usepackage{amsmath,amsfonts,amssymb,amsthm}
\usepackage{fullpage}
\usepackage[ruled]{algorithm2e}

\usepackage{subfigure} 
\usepackage{epstopdf}
\usepackage{enumerate}
\usepackage{enumitem}
\usepackage{graphicx}
\usepackage{natbib}
\usepackage{color}
\usepackage[sans]{dsfont}
\usepackage{mathtools}
\usepackage{hyperref}
\usepackage{cleveref}
\usepackage{tcolorbox}
\usepackage{multirow}
\usepackage{xspace}
\tcbuselibrary{skins,breakable}
\tcbset{enhanced jigsaw}

\newcommand{\george}[1]{{\color{green} George: #1}}
\newcommand{\zihan}[1]{{\color{blue} Zihan: #1}}
\newcommand{\tianyi}[1]{{\color{purple} Tianyi: #1}}

\newtheorem{theorem}{Theorem}
\newtheorem{observation}[theorem]{Observation}
\newtheorem{lemma}[theorem]{Lemma}
\newtheorem{claim}[theorem]{Claim}
\newtheorem{corollary}[theorem]{Corollary}

\newtheorem{definition}[theorem]{Definition}

\newcommand{\poly}{\mathsf{poly}}
\newcommand{\aset}{{\mathcal{A}}}
\newcommand{\bset}{{\mathcal{B}}}

\newcommand{\dset}{{\mathcal{D}}}
\newcommand{\sset}{{\mathcal{S}}}
\newcommand{\dist}{\textnormal{\textsf{dist}}}
\newcommand{\eps}{\varepsilon}

\newcommand{\set}[1]{\left\{ #1 \right\}}

\newenvironment{properties}[2][0]
{
	\begin{enumerate} \setcounter{enumi}{#1}}{\end{enumerate}}

\begin{document}

\begin{titlepage}
	
\title{Paths and Intersections: Exact Emulators for Planar Graphs}
\date{}
\author{George Z. Li\thanks{Carnegie Mellon University, PA, USA. Email: {\tt gzli929@gmail.com}. Part of the work was done at the 2023 DIMACS REU program, supported by NSF grants CCF-1836666 and CNS-2150186.}  
\and 
Zihan Tan\thanks{University of Minnesota Twin Cities, MN, USA. Email: {\tt zihantan1993@gmail.com}.}
\and 
Tianyi Zhang \thanks{ETH Zürich, \href{}{tianyi.zhang@inf.ethz.ch. Supported by the starting grant ``A New Paradigm for Flow and Cut Algorithms'' (no. TMSGI2\_218022) of the Swiss National Science Foundation. Work done at Tel Aviv University, supported by European Research Council (ERC) under the European Union’s Horizon 2020 research and innovation programme (grant agreement No 803118 UncertainENV).}} } 

\maketitle
	
\thispagestyle{empty}

\input{00_abstract}

\end{titlepage}

\tableofcontents

\newpage

\input{01_intro}
\input{02_two-face}

\input{03_general}

\appendix
\input{04_appendix}

\bibliographystyle{alpha}
\bibliography{refs.bib}

\end{document}

%% file: 00_abstract.tex
\begin{abstract}
    We study vertex sparsification for preserving distances in planar graphs. Given an edge-weighted planar graph with $k$ terminals, the goal is to construct an emulator, which is a smaller edge-weighted planar graph that contains the terminals and exactly preserves the pairwise distances between them. We construct exact planar emulators of size $O(f^2k^2)$ in the setting where terminals lie on $f$ faces in the planar embedding of the input graph. Our result generalizes and interpolates between the previous results of Chang and Ophelders and Goranci, Henzinger, and Peng which is an $O(k^2)$ bound in the setting where all terminals lie on a single face (i.e., $f=1$), and the result of Krauthgamer, Nguyen, and Zondiner, which is an $O(k^4)$ bound for the general case (i.e., $f=k$).

    Our construction follows a recent new way of analyzing graph structures, by viewing graphs as paths and their intersections, which we believe is of independent interest.
\end{abstract}

%% file: 01_intro.tex
\section{Introduction}

Graph compression is an approach that reduces the size of large graphs while preserving key properties, such as flow/cut/distance information. By shrinking the graph before performing computations, it saves computational resources. This approach has played a crucial role in designing faster and more efficient graph algorithms.

We study a specific type of graph compression, called \emph{vertex sparsifiers for preserving distances}, on planar graphs. The input is an edge-weighted planar graph $G$ and a subset $T\subseteq V(G)$ of its vertices called \emph{terminals}. The goal is to compute a smaller planar graph $H$ that contains all terminals in $T$ and preserves their pairwise distances in $G$.
Such a graph $H$ is called an \emph{emulator} of $G$. Formally, $H$ is a \emph{quality-$q$} emulator of $G$ with respect to $T$, for some real number $q\ge 1$, iff
\[
\forall t,t'\in T, \quad \dist_G(t,t')\le \dist_H(t,t')\le q\cdot \dist_G(t,t').
\]
The research focus in vertex sparsification lies in the trade-off between the quality and the emulator size, measured by $|V(H)|$, the number of vertices in the emulator. Ideally, we would like to preserve all distances accurately with small-sized emulators.

If there were no constraints on the structure of the emulator, then for any graph $G$ (not necessarily planar) with $k$ terminals, we could simply set $H$ as the clique on the terminals and give each edge $(t,t')$ length $\dist_G(t,t')$, trivially achieving optimal quality $q=1$ and optimal size $|V(H)|=k$.
Therefore, for distance-preserving vertex sparsification to be interesting and useful, the emulator needs to inherit some structural properties from the input graph $G$. A natural requirement is that $H$ be a \emph{minor} of $G$. Such emulators are called \emph{distance-preserving minors} and have been studied extensively.
In the setting where $q=1$, it was shown by Krauthgamer, Nguyen, and Zondiner \cite{krauthgamer2014preserving} via a simple analysis that every graph $G$ admits a quality-$1$ distance-preserving minor with size $O(k^4)$, and they also proved a lower bound of $\Omega(k^2)$, leaving a gap to be closed. Another interesting setting is when $|V(H)|=k$, and is called the \emph{Steiner Point Removal} problem (first studied by Gupta \cite{gupta2001steiner}), requiring that the emulator only contain terminals.
For general graphs, the best achievable quality is shown to be $O(\log |T|)$ \cite{kamma2015cutting,cheung2018steiner,filtser2018steiner} and $\tilde \Omega(\sqrt{\log |T|})$ \cite{chan2006tight,chen2024lower}.
For minor-free graphs, after an exciting line of work
\cite{basu2008steiner,filtser2020scattering,hershkowitz20211,chang2023covering,chang2023shortcut}, Chang, Conroy, Le, Milenkovic, Solomon and Than recently showed that quality $O(1)$ can be achieved.

Another regime that recently received much attention is planar emulators for planar graphs. The input is a planar graph $G$, and goal is to construct an emulator $H$ that is also a planar graph but not necessarily a minor of $G$.
A recent work by Chang, Krauthgamer, and Tan \cite{chang2022almost} showed that quality $(1+\eps)$ can be achieved by emulators of near-linear size $O(k\cdot\poly(\log k/\eps))$. For exact (quality-$1$) planar emulators, the best upper and lower bounds remain $O(k^4)$ and $\Omega(k^2)$ \cite{krauthgamer2014preserving}.
Recently, Chang and Ophelders~\cite{ChangO20} and Goranci, Henzinger, and Peng~\cite{goranci2020improved} studied the case where all terminals lie on a single face in the planar embedding of $G$, and they showed that such graphs have exact planar emulators of size $O(k^2)$, which is also shown to be tight.

Therefore, the natural next question is:
\[\emph{Does every planar graph with k terminals have an exact planar emulator of size $O(k^2)$?}
\]

\subsection{Our Result}

In this paper, we make progress towards this question. Our result is summarized in the following theorem.

\begin{theorem}
\label{thm: main}
Let $G$ be an edge-weighted planar graph. Let $T\subseteq V(G)$ be a set of $k$ vertices, such that there exist $f$ faces in the planar embedding of $G$ that contain all vertices of $T$. Then there exists another edge-weighted planar graph $H$ with $T\subseteq V(H)$ and $|V(H)|=O(f^2k^2)$, such that for all  $t,t'\in T$, $\dist_H(t,t')=\dist_G(t,t')$.
\end{theorem}

When $f$ takes its minimum value $1$, our result recovers the $O(k^2)$ bound given by \cite{ChangO20,goranci2020improved}.
When $f$ takes its maximum value $k$, our result recovers the $O(k^4)$ bound given by \cite{krauthgamer2014preserving}.
Therefore, our result can be viewed as interpolating and generalizing these previous results.

Unlike the previous work \cite{ChangO20} which showed a universal construction and gave clean formula for directly computing edge weights from terminal distances, or the previous work \cite{goranci2020improved,krauthgamer2014preserving} which constructed the emulators by performing a series of operations on the input graph, we first construct a skeleton graph that captures the shortest path structures in the input graph, and then \emph{use a linear program} to find the edge weights that preserve the shortest-path distances between terminals.
For the first part, our construction of the skeleton graph relies on a new understanding recently adopted in \cite{chen2025path}: instead of viewing a graph as consisting of vertices and edges, we view it as being formed by paths and their intersections.
For the second part, the crux is establishing the feasibility of the edge-weight linear program. We take the dual of this LP, reducing the feasibility to proving the flow-equivalence properties of certain path-morphing operations, which is then done by Wye-Delta transformations\footnote{The Wye-Delta transformation is a celebrated technique for analyzing electrical networks. This technique has also been used in \cite{goranci2020improved} for simplifying half-grid graphs with terminals on the diagonal. We employ this technique in a different way, to show non-existence of certain flows.}. We believe that our approaches in both parts are of independent interest and will prove useful for more distance-based graph problems.

\paragraph{Related Work.}
Chang, Gawrychowski, Mozes, and Weimann \cite{chang2018near} also studied emulators for unweighted planar graphs, and proved that every $n$-vertex undirected unweighted planar graph admits an exact emulator (not necessarily planar) of size $\tilde O(\min\{k^2,\sqrt{kn}\})$, which is near-optimal.


\paragraph{Organization.} 
We provide a technical overview \Cref{sec: warmup}, featuring our skeleton graph construction, using the $2$-face case ($f=2$) as an example, to provide intuitions for the general $f$-face case. The detailed construction for the $f$-face case is then provided in \Cref{sec: main}. In \Cref{sec: edge weight}, we show how to compute the edge weights for the skeleton graph to preserves the shortest-path distances, where in \Cref{subsec: weight overview} contains a high-level overview of the approach in this part.


%% file: 02_two-face.tex
\section{Technical Overview: the 2-Face Case}
\label{sec: warmup}

Our approach consists of two steps: constructing the emulator graph and setting its edge weight.
In this section, we provide an overview of the emulator construction step, illustrating the ideas in the special case where all terminals lie on $2$ faces. An overview of the edge weight setting step can be found in  \Cref{subsec: weight overview}.

The construction for the $O(k^4)$-size emulators \cite{krauthgamer2014preserving} works as follows.
\begin{itemize}
    \item For every pair of terminals $t,t'$, take its shortest path $\pi_{t,t'}$ in $G$ (assuming it is unique, which can be achieved by standard techniques), so we have in total $\binom{k}{2}$ terminal shortest paths.
    \item For every pair of terminal shortest paths $\pi,\pi'$, their intersection must be either empty or a subpath of both $\pi$ and $\pi'$. Call the endpoints of this subpath \emph{special vertices}, so in total we have $2\cdot\binom{\binom{k}{2}}{2}=O(k^4)$ special vertices.
    \item Let $H$ be the union of all terminal shortest paths $\set{\pi_{t,t'}}_{t,t'\in T}$ (so $H\subseteq G$, and clearly $H$ is an exact emulator). The only vertices in $H$ with degree $\ne 2$ are special vertices (which include all terminals). We replace each induced path between a pair of special vertices with an edge connecting them whose length is the total length of the induced path, and the resulting graph $H$ is a minor of $G$ with $O(k^4)$ vertices. 
\end{itemize}

Their approach is simple and elegant, but also conveys an important message: the key information to be stored in an emulator is the \emph{intersections between terminal shortest paths} in $G$.

But how can we reduce the number of such intersections, which can be $O(k^4)$? The previous work \cite{ChangO20}
showed us a way: let different pairs of shortest paths share their intersections. Specifically, they constructed a compact quarter-grid structure that make  single-terminal-source shortest paths overlap significantly, and thereby ensuring that the intersections between them and other shortest paths are reused heavily. Alternatively, their strategy can be also interpreted as: let shortest paths ``support'' each other, in that some shortest paths simply come from the union of two other shortest paths. See \Cref{fig:quartergrid} for an illustration.

\begin{figure}[h]
	\centering
	\includegraphics[scale=0.1]{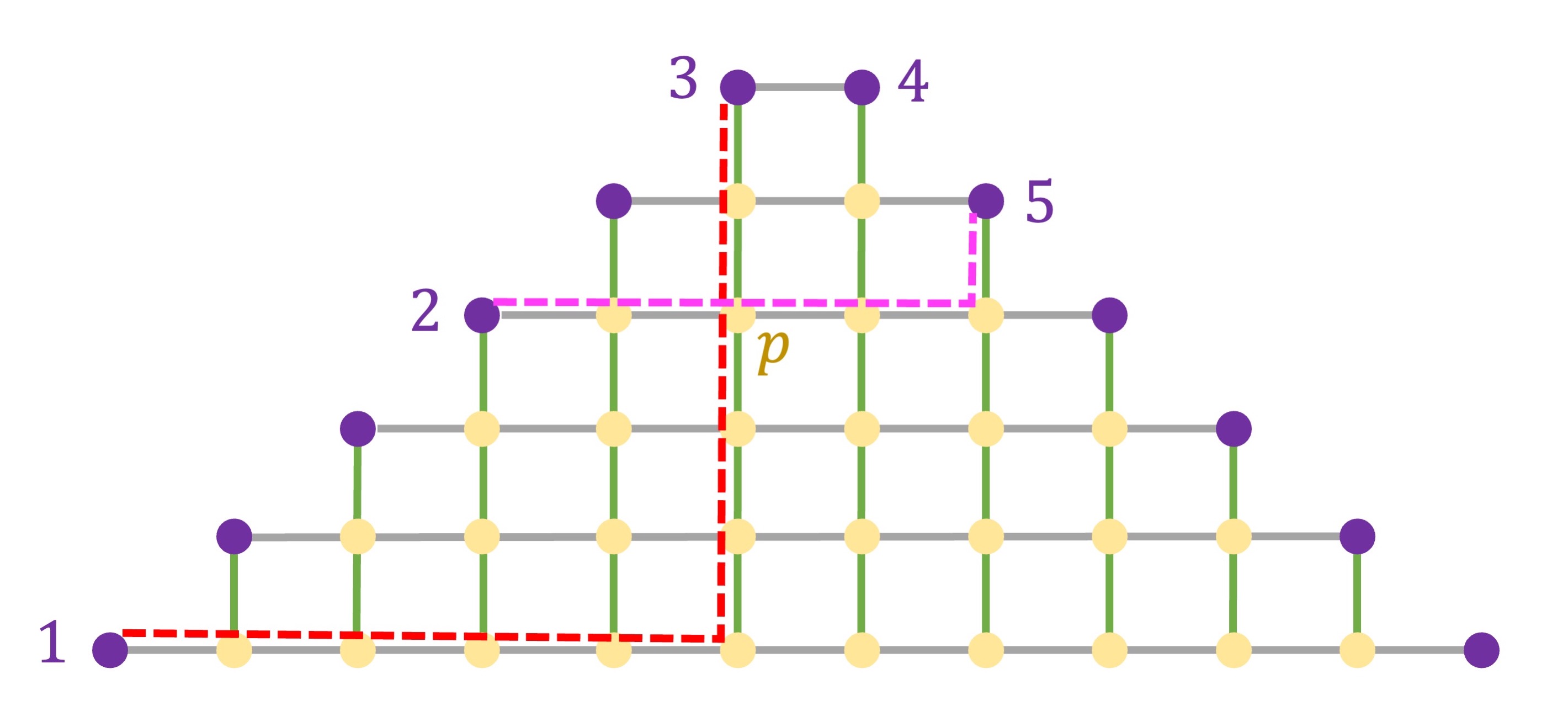}
	\caption{An illustration of the quarter-grid construction in \cite{ChangO20}. Terminals (purple) lie on the boundary. Terminal shortest paths are L-shaped. The $1$-$3$ path (red) and the $2$-$5$ path (pink) intersect at $p$. The $1$-$3$ path and the $2$-$4$ path also intersect at $p$, so $p$ as an intersection is shared/reused. Moreover, the $2$-$3$ paths is the $2$-$p$-$3$ L-shaped path, ``supported'' by the pink and the red paths.\label{fig:quartergrid}}
\end{figure}

Our approach also follows this general strategy of decreasing the number of intersections: \emph{extract important shortest paths and let them support other shortest paths}.
But our way of carrying out this strategy is fundamentally different from the previous work. We now discuss them in more detail, using the $2$-face instances as an example.

\subsection*{Critical paths}

Assume the input graph $G$ is a $2$-face instance. That is, all terminals lie on the boundaries of two faces, that we assume without loss of generality to be the outer face and the inner face.

Denote by $t_1,\ldots,t_{I}$ the terminals lying on the outer face and denote by $t'_1,\ldots,t'_{J}$ the terminals lying on the inner face, both in the clockwise order. For each pair $i,j$, we denote by $P_{i,j}$ the shortest path in $G$ connecting $t_i$ to $t'_j$. Consider now a terminal $t_i$ and all shortest paths $P_{i,1},\ldots,P_{i,J}$. 
We use the following observation, whose proof is straightforward and is deferred to \Cref{apd: Proof of obs: split}.

\begin{observation}
\label{obs: split}
For each $i$, there is a unique $j$, such that the inner face is contained in the region enclosed by path $P_{i,j}$, path $P_{i,j+1}$, and the segment of the inner boundary  from $t'_j$ clockwise to $t'_{j+1}$. 
\end{observation}

In other words, there exists a unique $j$, such that when we ``morph'' from $P_{i,j}$ to $P_{i,j+1}$, we suddenly ``switch to the other side''.
In this case, we say that terminal $t_i$ \emph{splits} at the terminal pair $(t'_j,t'_{j+1})$.
Intuitively, the paths $P_{i,j}$ and $P_{i,j+1}$ govern the behavior of all other shortest paths from $t_i$, so we call them \emph{critical paths} from $t_i$
(see \Cref{fig:split}).
Critical paths are the important shortest paths we extract from the input graph $G$.

\begin{figure}[h]
	\centering
	\includegraphics[scale=0.1]{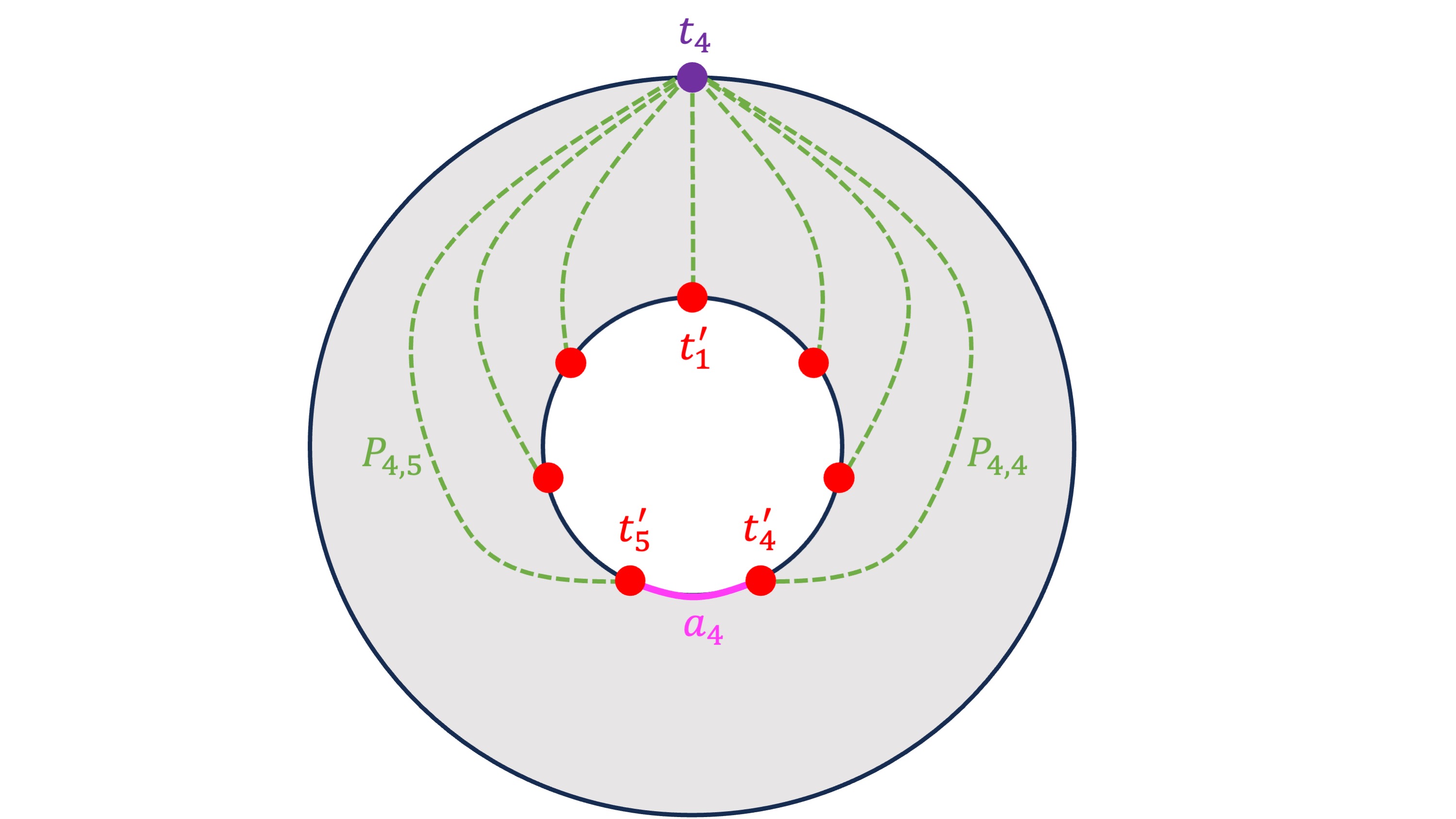}
	\caption{An illustration of $t_4$ splitting at $(t'_4,t'_5)$. Paths $P_{4,4}$ and $P_{4,5}$ are critical paths from $t_4$.\label{fig:split}}
\end{figure}

Furthermore, we can show that, as we move from terminal to terminal clockwise on the outer face, their ``split location'' also moves clockwise on the inner face. The proof is deferred to \Cref{apd: Proof of obs: split move}.

\begin{observation}
\label{obs: split move}
Let $t_{i_1},t_{i_2},t_{i_3}$ be terminals on the outer face clockwise in this order. If $t_{i_1},t_{i_2},t_{i_3}$ split at $(t'_{j_1},t'_{j_1+1}),(t'_{j_2},t'_{j_2+1}),(t'_{j_3},t'_{j_3+1})$, respectively, then $t'_{j_1},t'_{j_2},t'_{j_3}$ lie on the inner face clockwise in this order.
\end{observation}

\underline{\emph{Simplifying Assumptions:}} In order to better illustrate the ideas in the construction,
we assume that  $k$ is even and $I=J=k/2$. That is, $k/2$ terminals lie on the outer face, and $k/2$ terminals lie on the inner face. Furthermore, we assume that for each $i\in [k/2]$, terminal $t_i$ splits at $(t'_i,t'_{i+1})$, and so paths $P_{i,i}$ and $P_{i,i+1}$ are critical paths from terminal $t_i$.

\subsection*{Emulator structure for preserving inter-face distances\footnote{We remark that, for the purpose of technical overview, we only discuss the emulator structure needed for preserving distances between a terminal on the inner face and a terminal on the outer face, which is the core part of $2$-face emulators. A complete description of the emulator structure is provided in \Cref{apd: 2-face emulator}.}}

Now that we have found the important shortest paths (i.e., the critical paths), it is time for them to form a structure and support other shortest paths. Specifically, the plan is to construct an emulator consisting of only critical paths, and then design the route of other shortest paths accordingly. 

The first question is: \emph{in our emulator $H$, how should these critical paths intersect each other?}

Consider a pair $P_{1,1},P_{2,2}$ of critical paths, where $P_{1,1}$ connects $t_1$ to $t'_1$ and $P_{2,2}$ connects $t_2$ to $t'_2$. Note that, if we let $P_{1,1}$ and $P_{2,2}$ intersect in $H$, then we can extract a $t_1$-to-$t'_2$ path and a $t_2$-to-$t'_1$ path by non-repetitively using the edges of $E(P_{1,1})\sqcup E(P_{2,2})$ (if an edge belongs to both paths then we are allowed to use it twice). Since in graph $H$, $P_{1,1}$ is meant to be the shortest path connecting $t_1$ and $t'_1$ and $P_{2,2}$ is meant to be the shortest path connecting $t_2$ and $t'_2$, by allowing $P_{1,1}$ and $P_{2,2}$ to intersect in $H$, we are essentially enforcing that
\[\dist_H(t_1,t'_2)+\dist_H(t_2,t'_1)\le \dist_H(t_1,t'_1)+\dist_H(t_2,t'_2).\]
However, this is not necessarily true for the corresponding distances in $G$. Therefore, since we want $H$ to be an emulator for $G$, to be on the safe side, we cannot allow paths $P_{1,1}$ and $P_{2,2}$ to intersect.


By similar reasons, for each $1\le i\le k/2$, we cannot allow path $P_{i,i}$ and path $P_{i+1,i+1}$ to intersect. Therefore, the paths $P_{1,1},P_{2,2},\ldots,P_{k/2,k/2}$ must be ``parallel'' to each other. Similarly, the paths $P_{1,2},P_{2,3},\ldots,P_{k/2,1}$ must be parallel, too.
This leads us to the structure of $H$:
it is simply obtained by (i) adding, for each $i$, two curves from $t_i$ to $t'_{i}$ and $t'_{i+1}$ respectively, that share only the endpoint $t_i$ and are in the same shape as their corresponding critical paths in $G$; (ii) ensuring that curves corresponding to parallel paths do not intersect; (iii) replacing each intersection between two curves with a new vertex. See \Cref{fig: H3} for an illustration.
By reorganizing the terminal locations, we obtain a grid-shape structure shown on the right of \Cref{fig: H3}.
It is easy to verify that $|V(H)|=O(k^2)$.

\begin{figure}[h]
\centering
\subfigure
{\scalebox{0.1}{\includegraphics{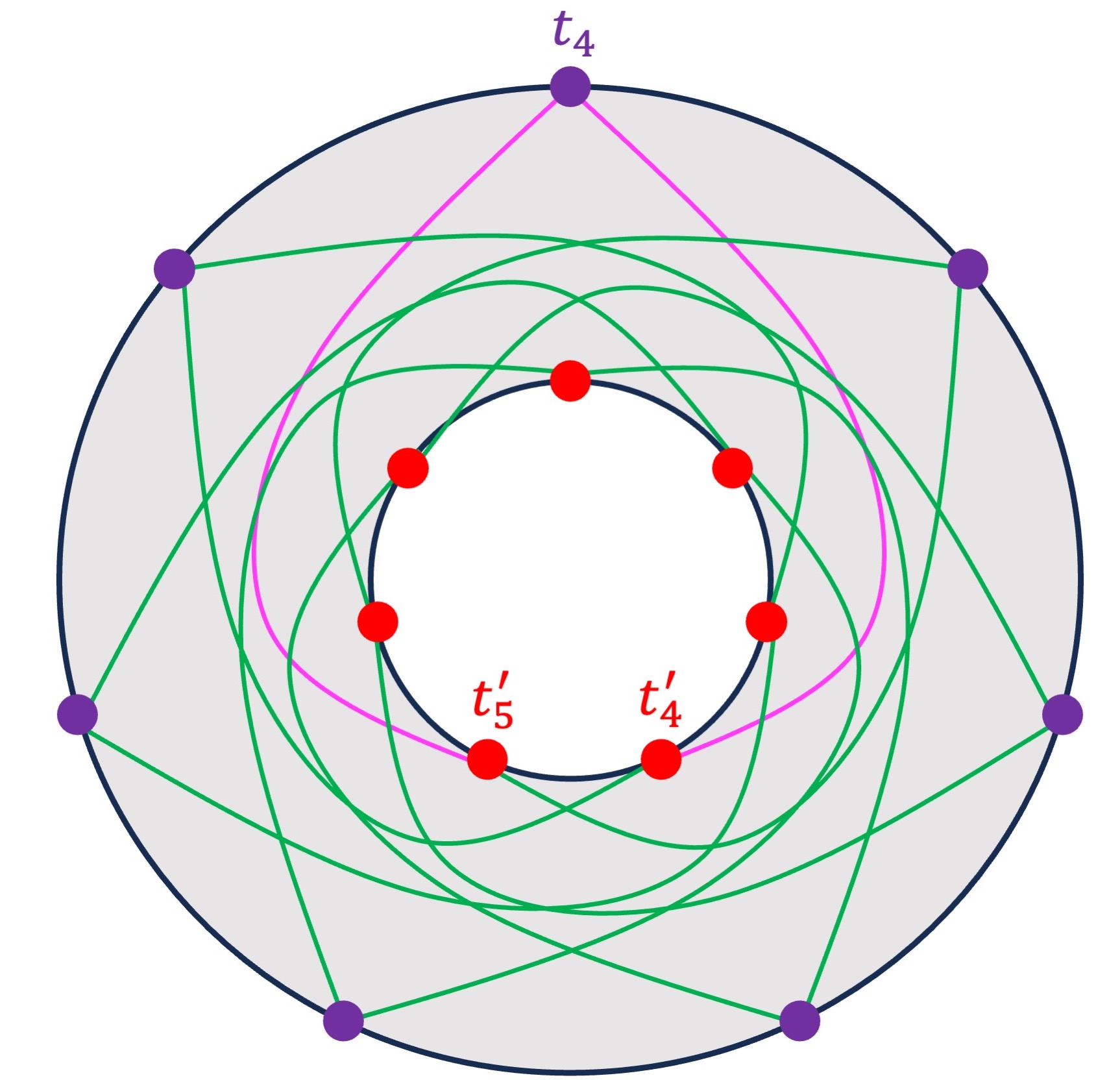}}}
\hspace{0.5cm}
\subfigure
{\scalebox{0.105}{\includegraphics{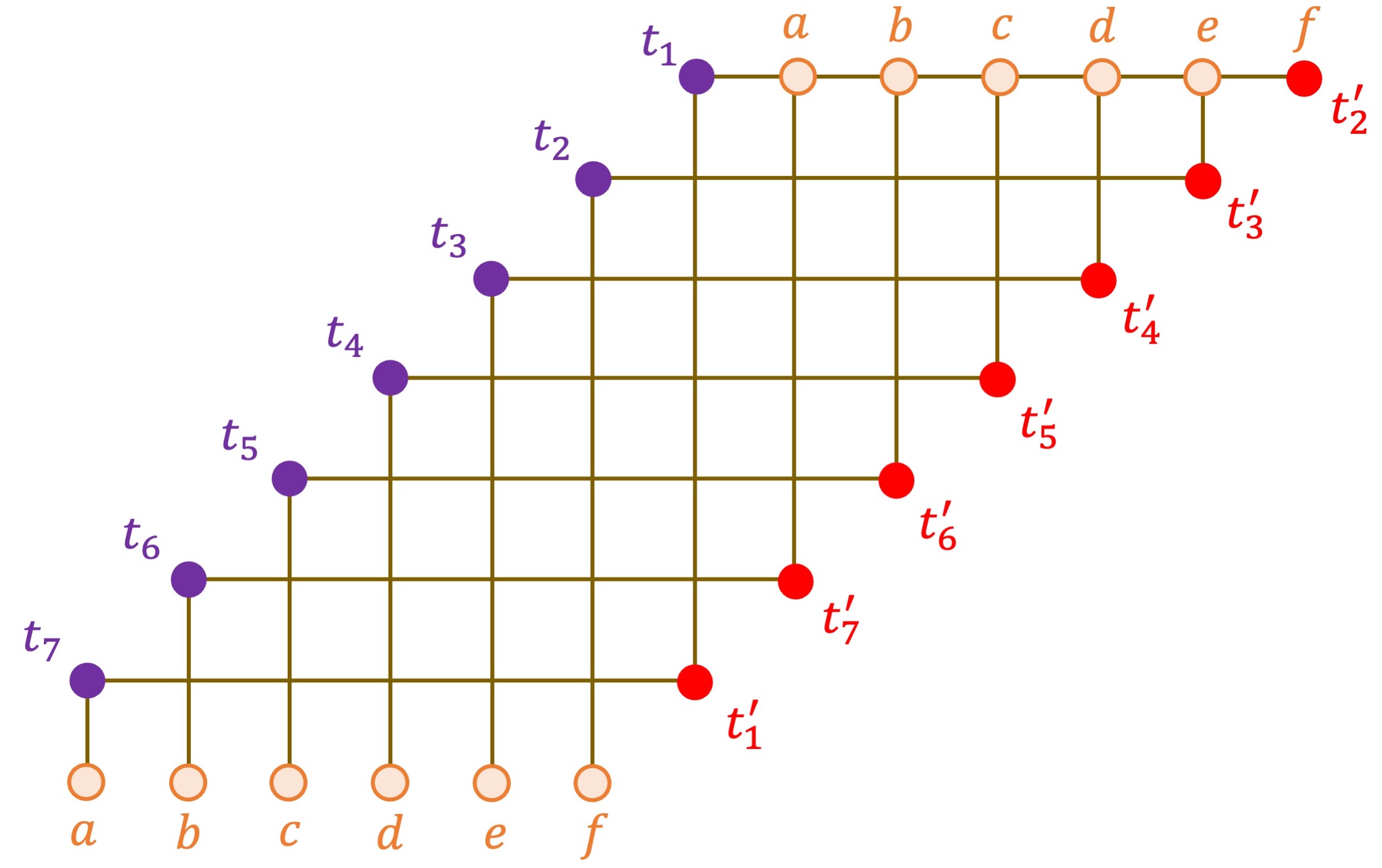}}}
\caption{An illustration of the structure of $H$. Left: the picture after adding all curves. The curves corresponding to the critical paths of $t_4$ are shown in pink. Right: the grid-shape structure after reorganizing the location of terminals. Each intersection between a horizontal line and a vertical line represents a new vertex. The down edge from $t_7$ connects to the second vertex on the first row (both marked $a$), and similarly for $b,c,d,e,f$. \label{fig: H3}}
\end{figure}

The next question is: \emph{how should we design the route of other shortest paths on this structure?}

Consider the $t_1$-$t'_3$ shortest path. Since its shape is supposed to be governed by critical paths $P_{1,1}$ and $P_{1,2}$, it may not escape the area enclosed by $P_{1,1}, P_{1,2}$ and the inner face boundary. 
Moreover, its intersections with $P_{1,2}$ and $P_{2,3}$ should be subpaths.
So there are two natural ($1$-bend) choices: 
\begin{enumerate}
\item $t_1$ first goes horizontally to $e$ and then vertically down to $t'_3$; and 
\item $t_1$ first goes vertically down to path $P_{2,3}$ and then follows $P_{2,3}$ horizontally to $t'_3$. 
\end{enumerate}
They are in fact symmetric, so let's take option (2).

Surprisingly, this single choice we made for the $t_1$-$t'_3$ shortest path essentially determines the routes for all other shortest paths. To see why it is true, note that
\begin{itemize}
    \item by similar reasons, the shortest paths between $t_1$-$t'_3$, $t_2$-$t'_4$, $t_3$-$t'_5,\ldots,$ $t_{k/2}$-$t'_2$ should be parallel (mutually vertex-disjoint) as well, and the only way to achieve this property in the structure of $H$ is to route, for each $i$, the $t_i$-$t'_{i+2}$ shortest path as: starting at $t_i$, going one step vertically down to the critical path $P_{i+1,i+2}$, and then following $P_{i+1,i+2}$ horizontally to $t'_{i+2}$; 
    \item when we next consider the $t_1$-$t'_4$ shortest path, since its intersections with $t_1$-$t'_2$, $t_1$-$t'_3$ should both be its subpaths, the only natural ($1$-bend) option is to first go $2$ steps vertically down to the critical path $P_{3,4}$ and then follow $P_{3,4}$ horizontally to $t'_4$ (see \Cref{fig:direction}); this in turn forces the similar routes for all $t_2$-$t'_5$, $t_3$-$t'_6,\ldots, t_{k/2}$-$t'_3$ shortest paths: all first $2$ steps down and then all the way to the right; and
    \item finally, via similar arguments, we can argue that all the $t_i$-$t'_j$ shortest paths should take the ``first vertically down, then horizontally right'' $1$-bend route.
\end{itemize}

\begin{figure}[h]
	\centering
	\includegraphics[scale=0.1]{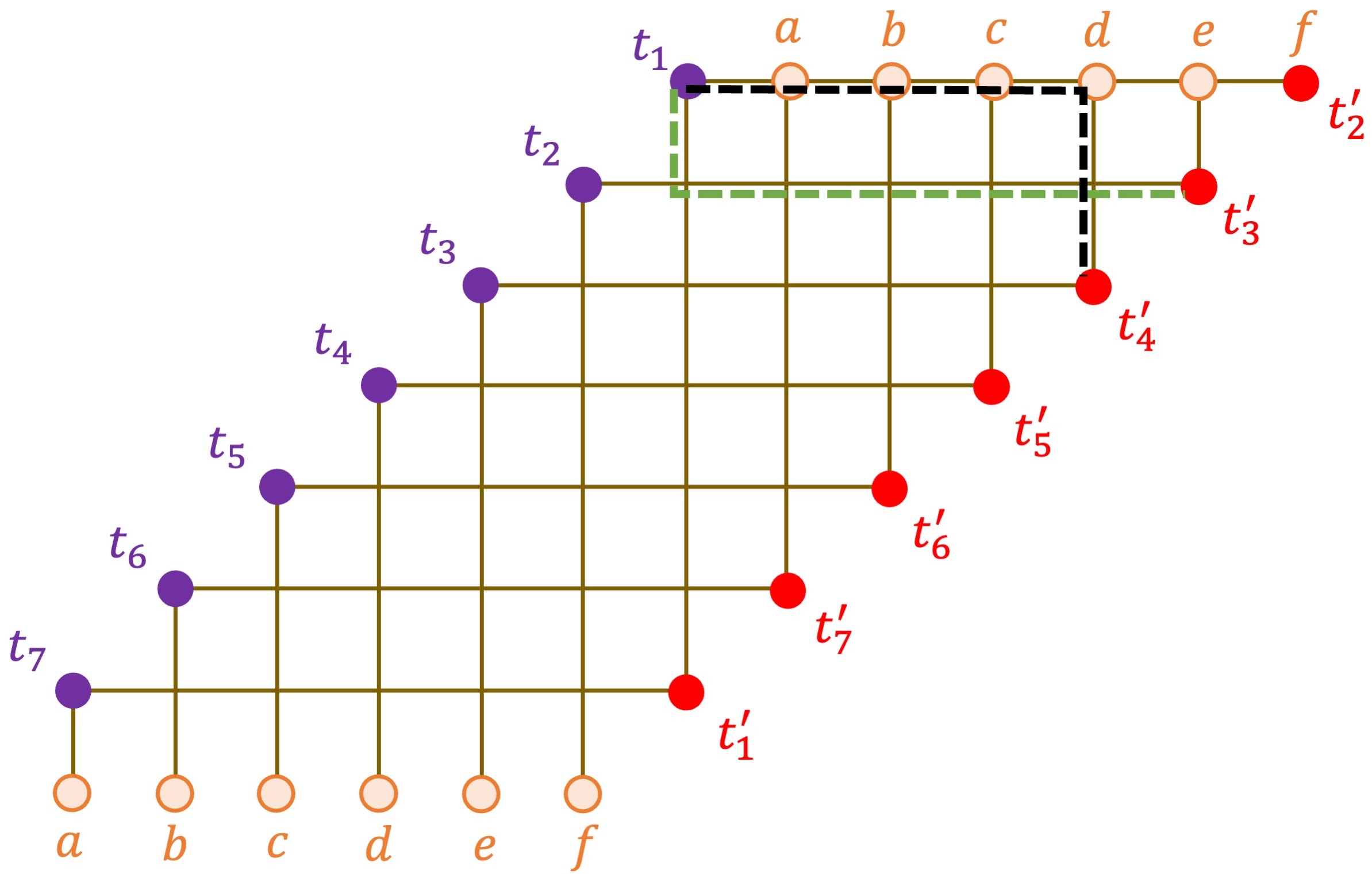}
	\caption{If the $t_1$-$t'_3$ shortest path follows the ``down-right'' route (shown in green), while the $t_1$-$t'_4$ shortest path follows the ``right-down'' route (shown in black), then this will create a pair of vertices with multiple shortest paths between them, which is undesired.\label{fig:direction}}
\end{figure}

To summarize, by taking option (2) for the $t_1$-$t'_3$ shortest path, we essentially prioritize the vertical direction over the horizontal direction in the structure of $H$ shown in \Cref{fig: H3}, thereby guiding all $t_i$-$t'_j$ shortest paths to go first vertically and then horizontally. Equivalently, we are prioritizing the vertical critical paths over the horizontal ones. 
Looking back, if we had taken option (1) for routing the $t_1$-$t'_3$ shortest path, then we would have eventually prioritized the horizontal critical paths over the vertical ones. In either way, we have to break the ties between the two directions to establish a correct shortest path structure in $H$.
To formalize this, we introduce some new notions.


\paragraph{Primary/secondary paths, and canonical paths.}
From a terminal $t_i$, there are two critical paths connecting it to $t'_i$ and $t'_{i+1}$, respectively. We call $P_{i,i}$ the \emph{primary path from $t_i$} and $P_{i,i+1}$ the \emph{secondary path from $t_i$}. Similarly, from each terminal $t'_j$ there are two critical paths connecting it to $t_{j-1}$ and $t_{j}$. We call $P_{j-1,j}$ the \emph{primary path from $t'_j$} and $P_{j,j}$ the \emph{secondary path from $t'_j$}. 

Between a pair $t_i, t'_j$ of terminals, we define their \emph{canonical path} $\gamma_{i,j}$ as follows. Let $p$ denote the node where the primary paths from $t_i$ and $t'_j$ intersect. The canonical path $\gamma_{i,j}$ is the union of (i) the subpath of the $t_i$-primary path between $t_i$ and $p$; and (ii) the subpath of the $t'_j$-primary path between $p$ and $t'_j$. 

For every pair $t_i, t'_j$, we will eventually ensure that \emph{canonical path $\gamma_{i,j}$ is the shortest $t_i$-$t'_j$ path in $H$}, whose length is exactly $\dist_G(t_i,t'_j)$, their distance in $G$.

\subsection*{Summary for the $2$-face emulators and generalization to $f$-face emulators}

At the end of this section, we briefly summarize the ideas for constructing the emulator graph for $2$-face instances. The overall idea is to extract important shortest paths from the input graph and let them support other shortest paths. Concretely, in a $2$-face instance
\begin{itemize}
    \item from each terminal, there are two inter-face terminal shortest paths that govern the shape of others from left and right sides, and we call them \emph{critical paths};
    \item to form a structure using only critical paths, we have to draw all left critical paths in parallel (and they become horizontal lines) and similarly draw all right critical paths in parallel (and they become vertical lines), thereby obtaining a grid-shape structure;
    \item to design routes for other shortest paths in this structure, we have to prioritize one direction (vertical/horizontal) in the grid over the other, thereby classifying the critical paths as primary and secondary ones, and eventually set the routes of other shortest paths as the concatenation of primary subpaths.
\end{itemize}

Since the whole structure consists of $O(k)$ critical paths, and every pair of critical paths intersect at most once, the total number of vertices in the graph is $O(k^2)$.

Our construction for $f$-face instances follows the same strategy as the $2$-face case: extract critical paths, draw them in parallel, let them intersect and form a structure, prioritize directions and guide the routes of other shortest paths. The main difference is that, from a terminal, we will need more than $2$ critical paths to govern the shapes of others. 
Specifically, we will have $O(f)$ critical paths from each terminal (see \Cref{fig: critical_example}), therefore having a total of $O(fk)$ critical paths. Their pairwise intersections form a structure on $O(f^2k^2)$ nodes, which will be our emulator structure.

\begin{figure}[h]
\centering
\includegraphics[scale=0.12]{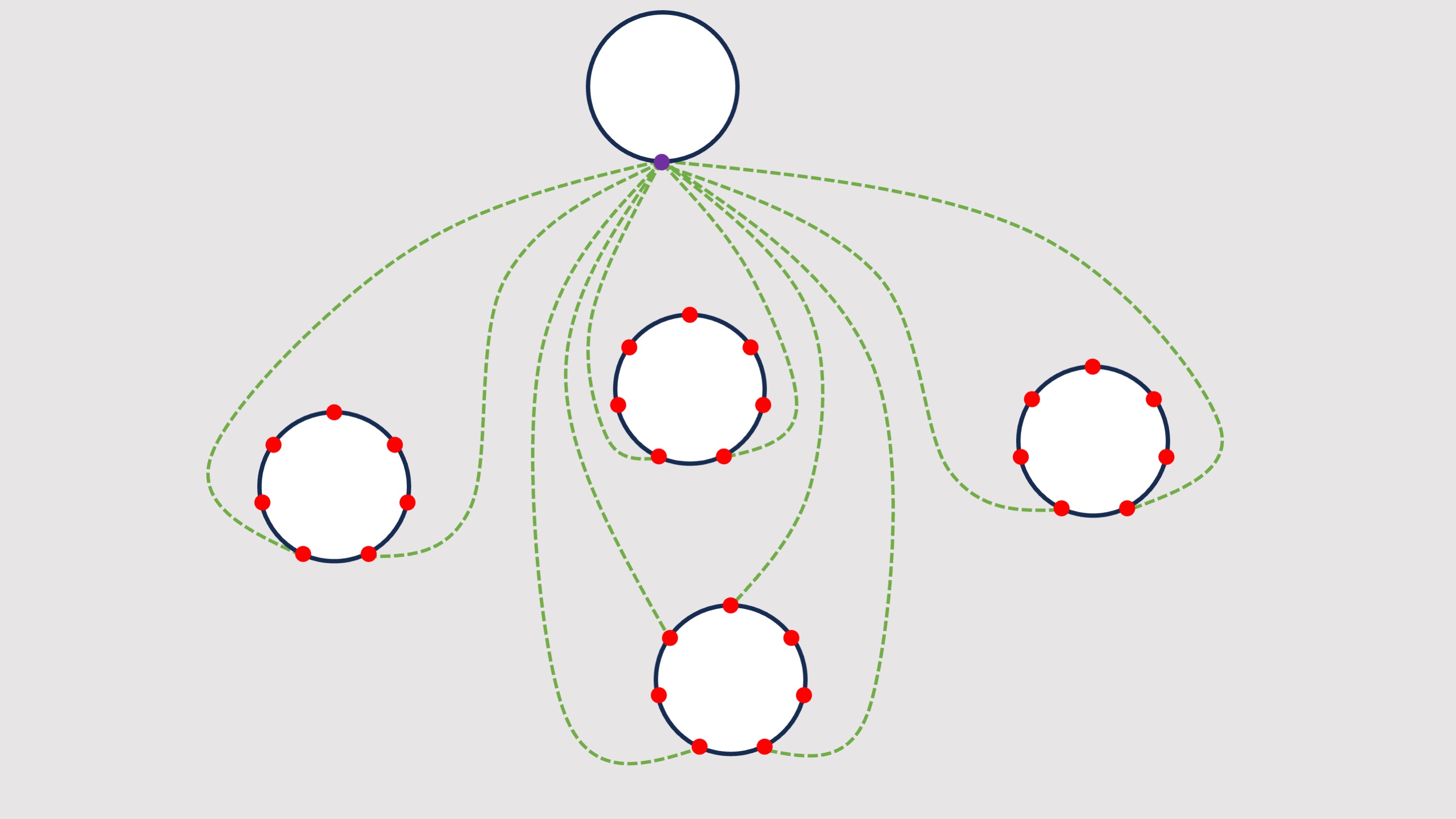}
\caption{An illustration of $O(f)$ critical paths from a terminal in a $f$-face instance.\label{fig: critical_example}}
\end{figure}

%% file: 03_general.tex
\section{The $f$-face Case: Constructing the Graph}
\label{sec: main}

In this section we provide the proof of \Cref{thm: main}. Recall that we are given a plane graph $G$ and a set $T$ of its vertices called terminals, such that all terminals lie on the boundaries of $f$ faces in its planar embedding, and the goal is to construct a planar emulator $H$ with $|V(H)|=O(f^2k^2)$. 
Similar to the way we handle the $2$-face case in \Cref{sec: warmup}, we will first show the graph structure and then assign weights to its edges.

\subsection{Simplifying the input graph}

In this subsection, we massage the input graph $G$ so that it has good properties that simplifies our construction of the emulator. Let $T$ be the set of terminals, and let $F_1,\ldots,F_f$ be the faces in the planar embedding of $G$ that contain all terminals, then we would like to ensure that

\begin{properties}{G}
\item the boundary of each face is a simple cycle that only contains terminals;\label{prop: G1}
\item the face cycles are vertex-disjoint, so each terminal lies on only one face among $F_1,\ldots,F_f$; and \label{prop: G2}
\item for distinct terminals $t_1,t_2,t_3,t_4$, if we denote by $P$ the $t_1$-$t_2$ shortest path, and denote by $Q$ the $t_3$-$t_4$ shortest path, then either
\begin{itemize}
    \item $P,Q$ are vertex-disjoint; or
    \item $P,Q$ intersect at a single internal vertex, and they cross over this intersection (that is, if we denote by $e_1,e_2$ the edges of $P$ and by $e_3,e_4$ the edges of $Q$ incident to the intersection, then the circular order of the edges $e_1,e_2,e_3,e_4$ entering this vertex is either $e_1,e_3,e_2,e_4$ or $e_1,e_4,e_2,e_3$).
\end{itemize}\label{prop: G3}
\end{properties}

For a face $F$, label the terminals on $F$ in a clockwise order as $t_1,\ldots,t_{r}$. To guarantee Property~\ref{prop: G1}, we simply add edges $(t_i,t_{i+1})$ with weight $\dist_G(t_i,t_{i+1})$ for each $1\le i\le r$ (using the convention $r+1=1$) inside face $F$, so now terminals $t_1,\ldots,t_{r}$ lie on a new face surrounded by the added edges. Clearly, the new face satisfies Property~\ref{prop: G1}, the boundary of the new face contains the same terminals as $F$, and the distance between any pair of terminals remains unchanged. Perform similar operations on other faces and eventually $G$ satisfies Property~\ref{prop: G1}.

Let $F_1,F_2$ be two faces and let $T_{1,2}$ denote the set of terminals lying on both $F_1$ and $F_2$. In order to guarantee Property~\ref{prop: G2}, we simply think of $T_{1,2}$ as terminals on $F_1$ but not on $F_2$. That is, when we add edges to form a new face (in the previous paragraph) for $F_2$, we simply ignore terminals in $T_{1,2}$. It is easy to verify that the new face to replace $F_2$ does not contain terminals in $T_{1,2}$. Perform similar operations on other pairs of faces and eventually $G$ satisfies Property~\ref{prop: G2}.

For achieving Property~\ref{prop: G3}, first we can ensure by standard techniques that between every pair of its vertices there is a unique shortest path connecting them in $G$ (for example, the lexicographic perturbation scheme in \cite{erickson2018holiest}).
As a corollary, the intersection between every pair $P,P'$ of shortest paths in $G$ is either empty or a subpath of both $P$ and $P'$.

In order to achieve Property~\ref{prop: G3}, we ``redraw'' the graph as follows. First, we remove from $G$ all edges from that does not participate in any shortest path connecting a pair of terminals. Then,
\begin{enumerate}
\item for each vertex $u$, create a tiny disc $\dset_u$ around $u$ (so all discs $\set{\dset_u}_{u\in V(G)}$ are disjoint);
\item for each edge $e=(u,v)$, create a thin strip $\sset_e$ along the image of $e$, with one end touching disc $\dset_u$ and the other end touching disc $\dset_v$, such that all strips $\set{\sset_e}_{e\in E(G)}$ are disjoint;
\item for each terminal shortest path $P=(u_0,u_1,\ldots,u_r)$, we define its \emph{channel} as the union of:\\ $\dset_{u_0},\sset_{(u_0,u_1)}, \dset_{u_1},\sset_{(u_1,u_2)},\ldots,\dset_{u_r}$;
\item for each terminal shortest path $P$ connecting $t$ to $t'$, we redraw it as a simple curve $\gamma_P$ from $t$ to $t'$ that lies entirely in its channel, such that there are no crossings in any strips, and all crossings happen within some disc;
\item repeatedly, for a pair $P,Q$ whose images $\gamma_P,\gamma_Q$ cross at least twice, find a pair $c,c'$ of crossings between them that appear consecutively, and uncross $\gamma_P,\gamma_Q$ at their segment between $c,c'$ (decreasing the total number of curve crossings by $2$), until there are no such pairs; and
\item view the curves as the drawing of the new graph, by placing a new vertex on each crossing, and viewing each segment between a pair of consecutive crossings as an edge;
\item for each new edge $(u',v')$, if its image passes through strips $e_1,e_2,\ldots,e_r$, then its new weight is defined to be $w'(u',v')=\sum_{1\le i\le r}w(e_i)$; if its image lies entirely in some disc, then its new weight $w'(u',v')=0$.
\end{enumerate}
 
\begin{figure}[h]
\centering
\subfigure
{\scalebox{0.07}{\includegraphics{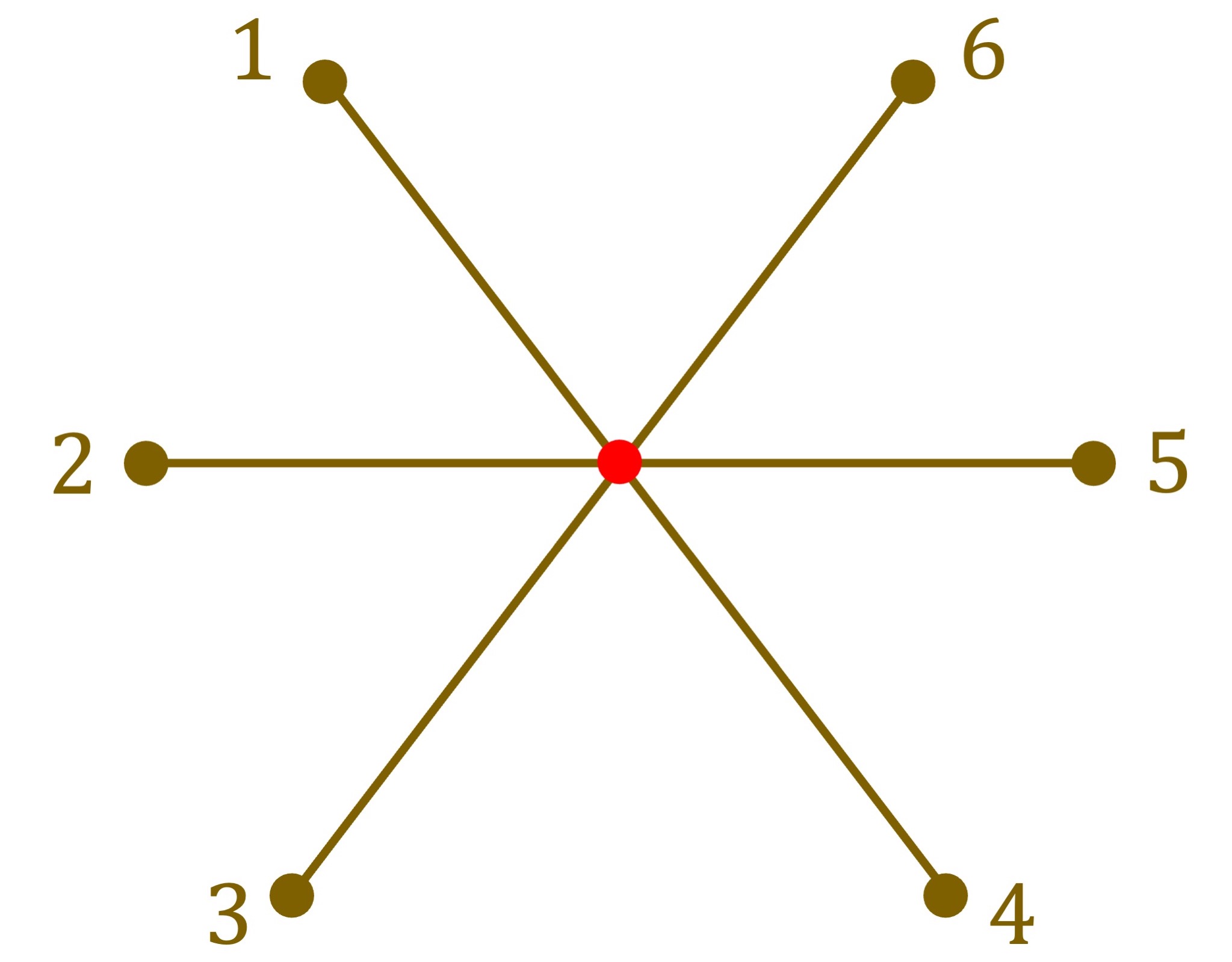}}}
\hspace{1.0cm}
\subfigure
{\scalebox{0.07}{\includegraphics{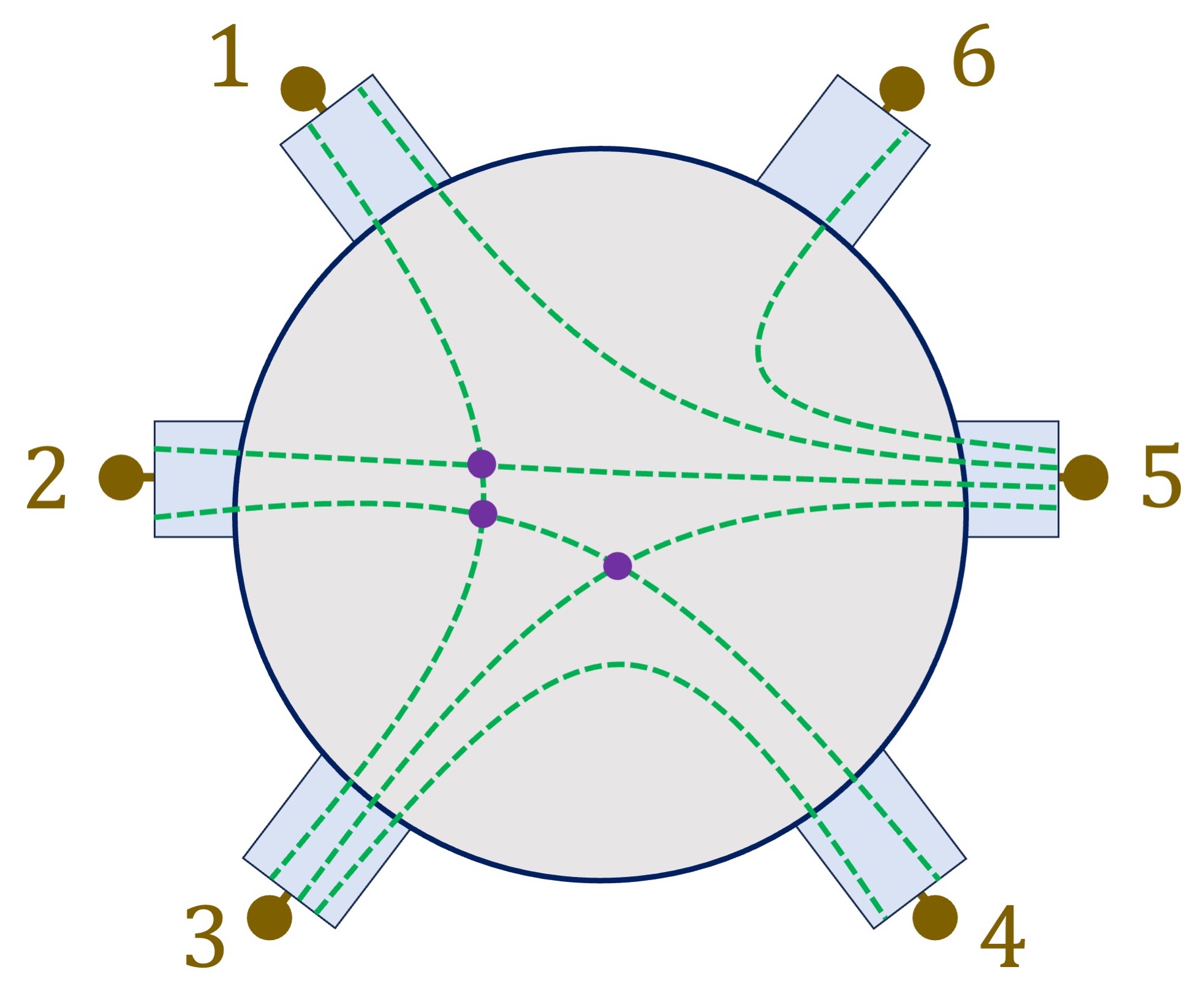}}}
\caption{An illustration of replacing shortest paths with curves. The red vertex in the original graph (left) supports shortest paths $1$-$3$, $1$-$5$, $2$-$4$, $2$-$5$, $3$-$4$, $3$-$5$, $5$-$6$. The curves in its disc and its incident thin strips are shown on the right, where new vertices (purple) are placed on the created intersections. \label{fig: YD}}
\end{figure}    

See \Cref{fig: YD} for an illustration. It is easy to verify that the above process produces a new graph satisfying Property~\ref{prop: G3}, does not change the faces and the distance between any pair of terminals.

\subsection{The graph structure of the emulator}

Denote by $F_1,\ldots,F_f$ the faces in the plane embedding of $G$ where all terminals lie on. For each $1\le r\le f$, let $T_r$ be the set of terminals lying on face $F_r$.
We will construct several graphs and then stick them together to preserve terminal distances in $G$. Specifically, we construct 
\begin{itemize}
\item for each $1\le r\le f$, a one-face emulator $H_r$ with terminals in $T_r$ lying on a single face $F'_r$ in the same order as they lie on $F_r$, such that for every pair $t,t'\in T_r$, $\dist_{H_r}(t,t')=\dist_{G}(t,t')$; 
and in fact we will simply use the construction of \cite{ChangO20} and \cite{goranci2020improved}, which gives a construction of $H_r$ on $O(|T_r|^2)$ vertices;
and
\item a central graph $H^*$ that contains all terminals, where for each $1\le r\le f$, terminals in $T_r$ lie on a single face $F^*_r$ in the same order in which they lie on $F_r$ in $G$, and we will ensure that 
\begin{itemize}
    \item for each pair $t,t'$ of terminals lying on different faces, $\dist_{H^*}(t,t')=\dist_G(t,t')$;
    \item for each pair $t,t'$ of terminals lying on the same face, $\dist_{H^*}(t,t')\ge \dist_G(t,t')$.
\end{itemize}
\end{itemize}
The graphs $H^*,H_1,\ldots,H_f$ are then glued together in a similar way as in \Cref{apd: 2-face emulator}, by placing the drawings of $H_1,\ldots,H_f$ into the faces $F^*_1,\ldots,F^*_r$ of $H^*$, respectively, and identifying the copies of the each terminal. The correctness of the emulator follows immediately from the above properties.
Therefore, the only missing part is the construction of the central graph $H^*$.

\subsection{Constructing the graph $H^*$: preparation}
\label{sec: prepare}

The construction of graph $H^*$ is a generalization of  the $2$-face emulator construction in \Cref{sec: warmup}. 
The high-level idea for the $2$-face case is to keep, for each terminal $t$ on the outer face, the two critical paths, one named primary and the other named secondary, that govern the shapes of all shortest paths connecting $t$ to some terminal on the inner face, and let them cross in the same way as they cross in $G$. 
Moreover, we define canonical paths connecting every pair of terminals lying on different faces as the concatenation of subpaths of their primary paths,
and designate them as the shortest paths in the emulator. By doing this, we are able to ensure that two canonical paths cross iff their corresponding real shortest paths in $G$ cross.

We will construct $H^*$ using a similar approach. We first discuss on a high level (i) how to get rid of the simplifying assumptions made in \Cref{sec: warmup}; and (ii) what additional problems we will encounter for general $f$-face instances and how we plan to handle them.

The two assumptions we made in \Cref{sec: warmup} were: the numbers of terminals lying on the inner and the outer faces are the same, and the split location moves ``smoothly''. 
The advantage for these assumptions is that we only need to consider critical paths from terminals lying on the outer face, as they ``automatically include'' the critical paths from terminals on the inner face. Specifically, we could have defined critical paths from terminals on the inner face in the same way as the terminals on the outer face and added them into the structure of $H$, but in fact they were already there. Specifically, in \Cref{fig:direction}, for terminal $t'_1$: the critical path from $t_7$ to $t'_1$ and the critical path from $t_1$ to $t'_1$ are exactly the critical paths we wanted to add from $t'_1$; in other words, we could also say that $t'_1$ splits at $(t_7,t_1)$. Similarly, the critical paths from other terminal $t'_i$ were also there already.
However, this is no longer true if we remove the two simplifying assumptions. For example, many outer face terminals may split at the same location, leaving some inner face terminals not serving as the endpoint of any critical path, while others serving as many. Therefore, we do need to include critical paths from inner face terminals as well.

For $f$-face instances, when we look at all shortest paths from a terminal $t$ to other terminals, we will find that their shapes change more frequently, will need more than two of them to govern the rest. Luckily, we are able to show that there always exist $O(f)$ paths among them to govern the rest. This increases the number of total critical paths we need to add to $O(fk)$, which then leads to the $O(f^2k^2)$ bound of the size of $H^*$ (see \Cref{lem: critical} and \Cref{clm: size}). 

Another problem we encounter in handling $f$-face instances is that it becomes unclear what ``cross in the same way as they cross in $G$'' means. In \Cref{sec: warmup}, this simply meant that a pair of critical paths cross in $H$ iff they cross in $G$. However, for $f$-face instances, there can be more than three critical paths crossing each other, so the orders in which the crossings appear on these paths are crucial in the structure of $H^*$, and need to be handled with care. In fact, we cannot simply let the crossings appear in the same order as in $G$, but need to delicately move them, such that when we define canonical paths as the concatenation of certain critical subpaths, the ``canonical paths cross iff their corresponding shortest paths cross in $G$'' property is achieved.
All these need to be formalized, which we do now.

\paragraph{Critical Paths.}
For each $1\le r\le f$, denote $T_r=\set{t^r_1,\ldots,t^r_{|T_r|}}$, where the terminals are indexed according to the clockwise order in which they appear on the boundary of $F_r$. 
Let $t$ be a terminal not in $T_r$. 
We say that the $t$-$t^r_j$ shortest path in $G$ is \emph{equivalent to} the $t$-$t^r_{j+1}$ shortest path in $G$, iff the region enclosed by the union of (i) the $t$-$t^r_j$ shortest path; (ii) the $t$-$t^r_{j+1}$ shortest path; and (iii) boundary segment of face $F_r$ clockwise from $t^r_j$ to $t^r_{j+1}$, contains none of the faces in $\set{F_1,\ldots,F_f}$.
We say that the $t$-$t^r_j$ shortest path in $G$ is \emph{critical}, iff the $t$-$t^r_j$ shortest path in $G$ is not equivalent to
either the $t$-$t^r_{j+1}$ shortest path or the $t$-$t^r_{j-1}$ shortest path.

Intuitively, the $t$-$t^r_j$ shortest path is critical, iff the curve morphing from the $t$-$t^r_{j-1}$ shortest path to the $t$-$t^r_{j}$ shortest path and then to the $t$-$t^r_{j+1}$ shortest path has to ``scan over'' some other face and therefore incur a fundamental shape change. See \Cref{fig:critical} for an illustration.

\begin{figure}[h]
\centering
\includegraphics[scale=0.12]{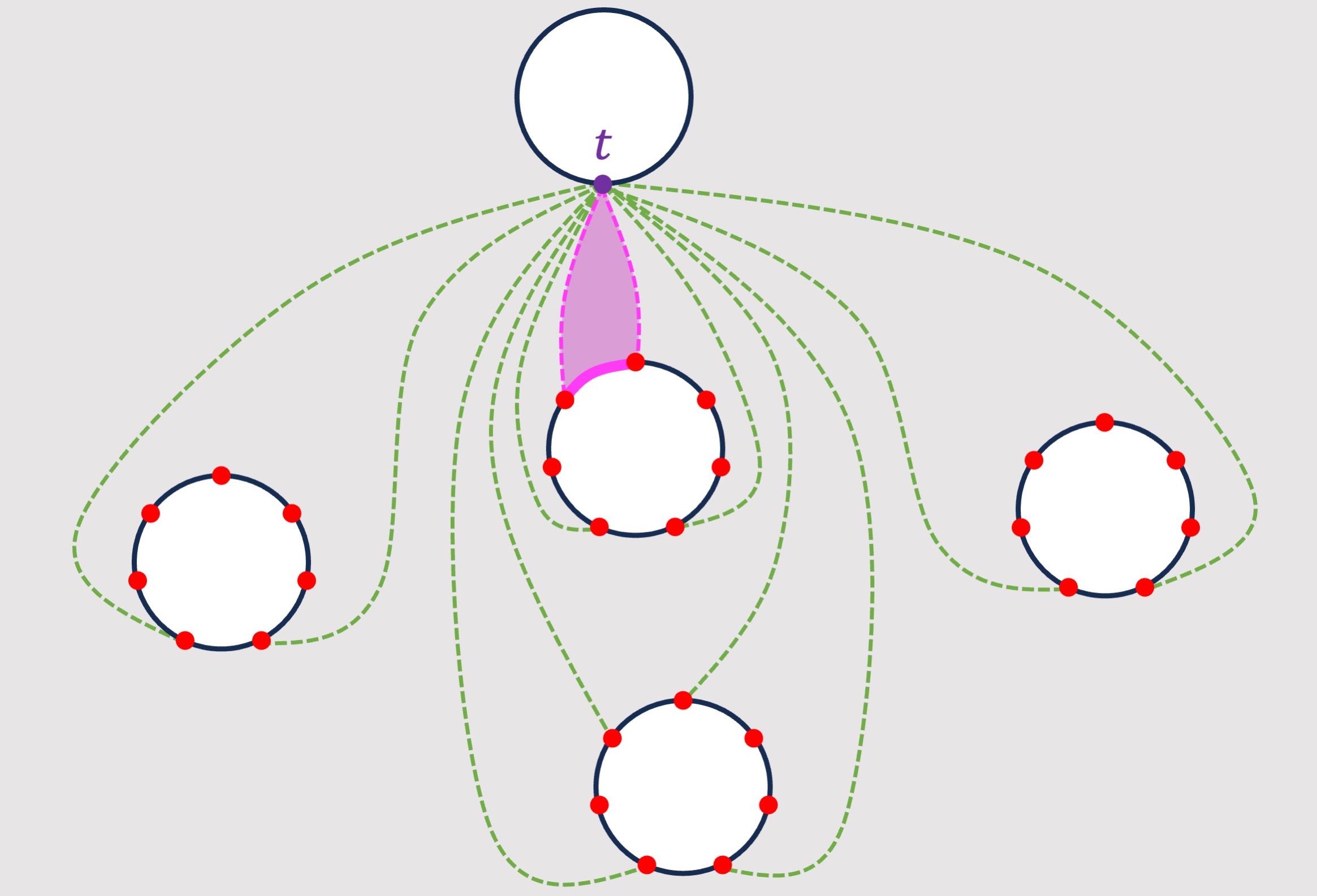}
\caption{All critical paths (green) from terminal $t$ and two non-critical shortest paths (pink) from $t$ with their enclosed region with the boundary segment (which contains no faces from $\set{F_1,\ldots,F_f}$).\label{fig:critical}}
\end{figure}

We have the following lemma.

\begin{lemma}
\label{lem: critical}
For each terminal $t$, the number of critical paths from $t$ is $O(f)$.
\end{lemma}
\begin{proof}
Let terminal $t$ lie on $F_1$ without loss of generality. Consider the pairs $(P_1,P_2)$ of critical paths with an endpoint at $t$ such that $P_1$, $P_2$, and an edge enclose a region containing some face.
Define the set system $\mathcal{A}$ to be the collection of regions induced by such pairs $(P_1,P_2)$ with one endpoint at terminal $t$. We claim that $\mathcal{A}$ is a laminar family. Indeed, the boundary of two regions for different critical paths can never cross since two critical paths starting at $t$ can only have one intersection (at $t$), and no other intersections.
Equivalently, the set system $\mathcal{B}$ consisting of the set of faces in the regions enclosed induced by a pair of critical paths $(P_1,P_2)$ from $t$ must also be a laminar family. Since the universe of $\mathcal{B}$ is of size $f$, we have $|\mathcal{B}|\le 2f-1$. 
\begin{observation}
    Each set of faces $B\in\mathcal{B}$ corresponds to at most two different regions $A_1,A_2\in\mathcal{A}$.    
\end{observation}
\begin{proof}
    Let us say $A_1$ is a Type 1 region if $A_1$ has terminal endpoints $(t,t_1,t_2)$ where $t_1,t_2\in B$ and $A_2$ is a Type 2 region if $A_2$ has terminal endpoints $(t,t_1',t_2')$ where $t_1',t_2'\not\in B$. We will show that there is at most one region of each type (for each $B\in\mathcal{B}$).
    
    Suppose there is another Type 1 region $A\neq A_1$ for $B$, defined by two critical paths $P_1,P_2$ from $t$ to consecutive terminals on some face in $B$. Since both endpoints of $P_1$ (resp. $P_2$) are in $A_1$ and $P_1$ (resp. $P_2$) can't cross the boundary of $A_1$ twice, we have $P_1$ (resp. $P_2$) lies entirely in $A_1$ except for the endpoint $t$. But then $A$ doesn't contain $t_1$ or $t_2$, since that would mean $P_1$ or $P_2$ would touch the boundary of $A_1$, a contradiction. 
    
    Suppose there is another Type 2 region $A\neq A_2$ for $B$, defined by two critical paths $P_1,P_2$ from $t$ to consecutive terminals on some face not in $B$. Since $A$ doesn't contain any other faces, we have that $t_1',t_2'\not\in A$. But by laminarity of $\mathcal{A}$, this implies that $A\subseteq A_2$ since $A\cap A_2\neq\emptyset$ and $t_1'\not\in A$ so $A_2\not\subseteq A$. But there are no faces in $A_2$ not in $B$ (by definition of $B$), so the endpoints of $P_1,P_2$ must be $t_1',t_2'$ and $A=A_2$.
\end{proof}
Combining $|\mathcal{B}|\le 2f-1$ with the above observation, we have $|\mathcal{A}|\le 4f-2$. Since each region in $\mathcal{A}$ corresponds to at most $2$ critical paths, there are at most $8f-4$ critical paths starting at $t$.
\end{proof}

\paragraph{Primary and secondary paths, and canonical paths.}
The critical paths can be naturally paired according to their shapes.
Let $t$ be a terminal and consider a face $F_r$ such that $t\notin F_r$. For terminals $t^r_1,\ldots,t^r_{|T_r|}$ that lie on $F_r$, according to the equivalence relation defined on the shortest paths from $t$ to them in $G$, the circular ordering $(1,\ldots,|T_r|,1)$ can be partitioned into segments: 
\[[a_0+1,\ldots,a_1],[a_1+1,\ldots,a_2],\ldots,[a_{s}+1,a_0],\]
where for each $0\le i\le s$, the $t$-$t^r_{a_i+1}$ path, the $t$-$t^r_{a_i+2}$ path, ..., and the $t$-$t^r_{a_{i+1}}$ path are equivalent to each other, but not to the $t$-$t^r_{a_i}$ path or $t$-$t^r_{a_{i+1}+1}$ path.
By definition, the $t$-$t^r_{a_i+1}$ path and the $t$-$t^r_{a_{i+1}}$ path are critical paths, and we say that they are \emph{paired}. 
We call the boundary segment of $F_r$ going clockwise from $t^r_{a_i+1}$ to $t^r_{a_{i+1}}$ the \emph{governed segment} of these paired shortest paths.

For paired critical paths, we define one of them as \emph{primary}, and the other as \emph{secondary}, as follows. Consider paired critical paths $t$-$t^r_{a}$ and $t$-$t^r_{b}$, where $t^r_a$ and $t^r_b$ lie on face $F_r$ and $t$ lies on face $F_{r^*}$, and the $F_r$ boundary segment between $t^r_{a}$ and $t^r_{b}$ is going clockwise from $t^r_{a}$ to $t^r_{b}$. Now
\begin{itemize}
\item if $r^*<r$, then the $t$-$t^r_{a}$ path is primary, and the $t$-$t^r_{b}$ path is secondary; and
\item if $r^*>r$, then the $t$-$t^r_{b}$ path is primary, and the $t$-$t^r_{a}$ path is secondary.
\end{itemize}

Consider now a pair $t,t'$ of terminals lying on distinct faces $F_r,F_{r'}$, respectively, with $r<r'$. The \emph{canonical path} between terminals $t,t'$ are defined as follows.
Let $t'$-$t^r_{a}$, $t'$-$t^r_b$ be the paired critical paths whose governed segment $[t^r_{a},t^r_b]$ on face $F_r$ contains $t$, and let $t$-$t^{r'}_{a}$, $t$-$t^{r'}_b$ be the paired critical paths whose governed segment $[t^{r'}_{a},t^{r'}_b]$ on face $F_{r'}$ contains $t'$. So the $t$-$t^{r'}_{a}$ shortest path and the $t'$-$t^{r}_{b}$ shortest path must cross, and we denote their crossing by $p$.
The $t$-$t'$ canonical path is defined to be the concatenation of 
\begin{itemize}
\item the subpath of the $t$-$t^{r'}_{a}$ shortest path between $t$ and $p$; and 
\item the subpath of the $t'$-$t^{r}_{b}$ shortest path between $t'$ and $p$.
\end{itemize}
Additionally, we call the intersection $p$ the \emph{bend} of the canonical path $t$-$t'$. 

\subsection{Constructing the graph $H^*$: the algorithm}

Similar to the construction of the structure of $H$ in \Cref{sec: warmup}, we can construct $H^*$ by (i) starting with the subgraph of $G$ induced by all critical paths from all terminals; (ii) suppressing all degree-$2$ vertices; and (iii) for every pair $t,t'$ of terminals, define its canonical path $\gamma_{t,t'}$ as in \Cref{sec: prepare}.
However, this construction does not achieve the crucial property that a pair of canonical shortest paths cross iff their corresponding shortest paths in $G$ cross (for example, see \Cref{fig: misplace_1}).
Therefore, we need to modify the structure by an iterative process to achieve this.

\begin{figure}[h]
\centering
\includegraphics[scale=0.12]{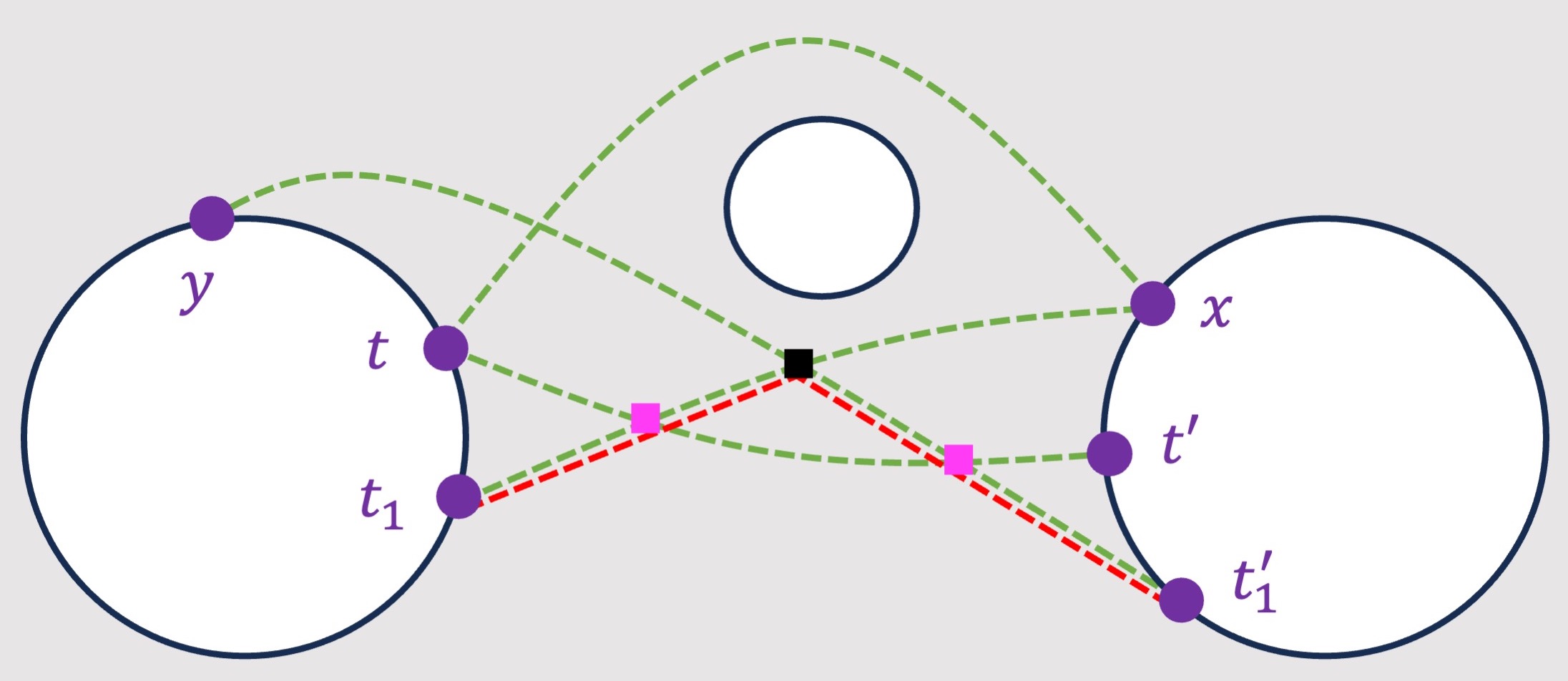}
\caption{An illustration of multiple crossings between the $t$-$t'$ critical path (the green dashed line between $t,t'$) and the $t_1$-$t'_1$ canonical path (red), and some other relevant critical paths.\label{fig: misplace_1}}
\end{figure}

Throughout, we maintain a graph $\hat H$, that is initialized as the structure obtained in the previous paragraph.
We also maintain a drawing of $\hat H$. Initially, the drawing of $\hat H$ is induced by the drawing of $G$ as $\hat H$ comes from $G$.
We will ensure that graph $\hat H$ always satisfies the following properties.

\begin{enumerate}
\item Graph $\hat H$ is an $f$-face instance that contains all terminals; for each $1\le r\le f$, the terminals in $T_r$ lie on $F_r$ in the same order in which they appear on $F_r$ in $G$.
\label{prop: f-face instance}
\item Graph $\hat H$ is the union of, for each critical path $P$ in $G$, a path $\rho_P$ connecting the same pair of terminals as $P$. For a pair $P,P'$ of critical paths whose endpoints are four distinct terminals, \label{prop: canonical intersection}
\begin{itemize}
    \item if $P,P'$ are vertex-disjoint, then $\rho_P,\rho_{P'}$ are also vertex-disjoint; and
    \item if $P,P'$ cross (at a single vertex), then $\rho_P,\rho_{P'}$ cross (also at a single vertex).
\end{itemize}
\item For every pair $t_1,t_2$ of terminals on different faces, the region enclosed by the $t_1$-$t_2$ canonical path in $\hat H$ and the $t_1$-$t_2$ shortest path in $G$ contains no face and no other terminals.
\label{prop: region}
\end{enumerate}

Clearly, Properties \ref{prop: f-face instance} and \ref{prop: canonical intersection} are satisfied by the initial $\hat H$. We show that Property \ref{prop: region} is also satisfied.

\begin{observation}
The initial graph $\hat H$ satisfies Property \ref{prop: region}.
\end{observation}
\begin{proof}
    Observe that by our construction of the critical and canonical paths, both the $t_1$-$t_2$ canonical path and the $t_1$-$t_2$ shortest path lie in the region enclosed by the critical paths outgoing from $t_1$ to the face $F_{t_2}$ where $t_2$ lies. Symmetrically, the two paths also lie in the region enclosed by the critical paths outgoing from $t_2$ to the face $F_{t_1}$ where $t_1$ lies. Call these regions $A_1$ and $A_2$, respectively. Observe that the only face which lies in region $A_1$ is $F_{t_2}$ and symmetrically, the only face which lies in region $A_2$ is $F_{t_1}$. Since the region enclosed by the $t_1$-$t_2$ canonical path and the $t_1$-$t_2$ shortest path lies in the intersection of $A_1$ and $A_2$, there must be no faces contained in it. Similarly, the only terminals which lie in region $A_1$ are $t_1$ and the terminals on $F_{t_2}$ and symmetrically, the only terminals which lie in $A_2$ are $t_2$ and the terminals on $F_{t_1}$. Since the region enclosed by the $t_1$-$t_2$ canonical path and the $t_1$-$t_2$ shortest path lies in the intersection of $A_1$ and $A_2$, the only terminals contained in the region are $t_1$ and $t_2$.
\end{proof}

Eventually, at the end of the iterative process, we will ensure that the resulting graph $\hat H$ satisfies the following additional property.

\begin{enumerate}[resume]
\item For a canonical path $\gamma_{t_1,t_2}$ between $t_1,t_2$ and another canonical path $\gamma_{t_3,t_4}$ between $t_3,t_4$ with $t_1,t_2,t_3,t_4$ being distinct terminals, $\gamma_{t_1,t_2}$ and $\gamma_{t_3,t_4}$ cross iff the $t_1$-$t_2$ shortest path in $G$ and the $t_3$-$t_4$ shortest path in $G$ cross.
\label{prop: intersect iff}
\end{enumerate}

\subsubsection{Definitions and notations}

\paragraph{Representation of critical paths and their subpaths.} The Property~\ref{prop: canonical intersection} ensures that, vaguely speaking, for any two non-terminal vertices $u,v$, the number of canonical paths that contain both of them is either $0$ or $1$ (since otherwise there will be two canonical paths crossing twice at $u,v$). 
Therefore, for notational convenience, when $u,v$ both lie on some canonical path $\rho$, we use $(u,v)$ to denote $\rho[u,v]$, the subpath of $\rho$ between $u$ and $v$. The notation is unambiguous even with $\rho$ being absent.
When we want to emphasize the direction of such a critical path/subpath, we will also use the notation $u\to v$ instead of $(u,v)$.
When $u,u'$ are terminals lying on the same face, we let $(u,u')$ be the face segment from $u$ to $u'$.
For a terminal $x$, we denote by $F_x$ the face that contains it.

\paragraph{Triangles.}
For three vertices $x,y,z$ where every pair of them appear on some canonical path, we denote by $(x,y,z)$ the triangle formed by $(x,y)$, $(y,z)$ and $(z,x)$.
We can also define triangles for $x,y,z$ where $x,y$ are terminals lying on the same face, and both pairs $x,z$ and $y,z$ lie on the same canonical path.

\paragraph{Areas.}\footnote{We will mostly use the notions areas to describe a type of ``morphing'' process that we will use later. Essentially, morphing corresponds to continuous deformation between a pair of homotopic simple curves (when terminal faces are viewed as holes). For curves $\alpha,\beta$ between terminals, we say that $\alpha$ is homotopic to $\beta$, or equivalently that $\alpha$ can be morphed into $\beta$, iff we can gradually change $\alpha$ to $\beta$, with the endpoints of $\alpha$ always staying on the same faces and the curve not scanning over any terminal faces. To avoid sophisticated formal definitions in topology, instead of saying that two curves are homotopic, we will just say that their enclosed area contains no terminal faces.} For a triangle $(x,y,z)$, we denote by $\textsc{Area}(x,y,z)$ its enclosed region. For critical paths $(x,y)$ and $(x',y')$ where $F_x=F_{x'}$ and $F_y=F_{y'}$, \emph{the area between paths $(x,y),(x',y')$}, denoted by $\textsc{Area}((x,y),(x',y'))$, is defined to be the area enclosed by (i) $(x,y)$; (ii) $(x',y')$; (iii) the segment of $F_x$ between $x,x'$; and (iv) the segment of $F_y$ between $y,y'$. That is, if $(x,y)$ and $(x,y')$ cross at $z$, then it is the union $\textsc{Area}(x,x',z)\cup\textsc{Area}(y,y',z)$; if $(x,y)$ and $(x',y')$ don't cross, then it is the quadrilateral $\textsc{Area}(x,y,y',x')$. See Figure \ref{fig: area_definition} for an illustration. 

\begin{figure}[h]
\centering
\subfigure
{\scalebox{0.25}{\includegraphics{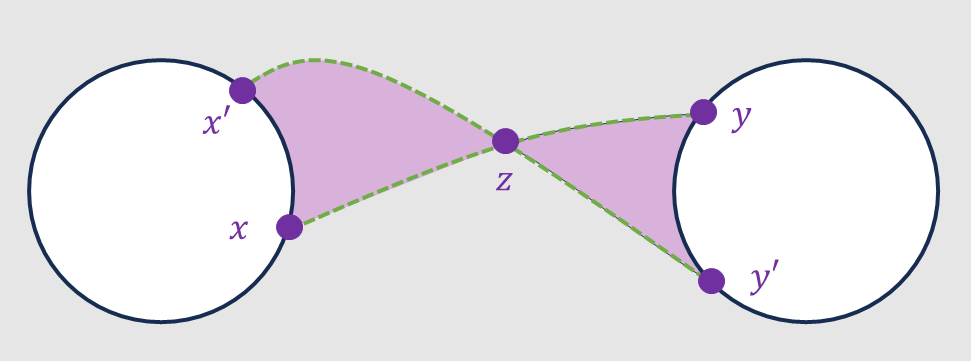}}}
\hspace{1.0cm}
\subfigure
{\scalebox{0.25}{\includegraphics{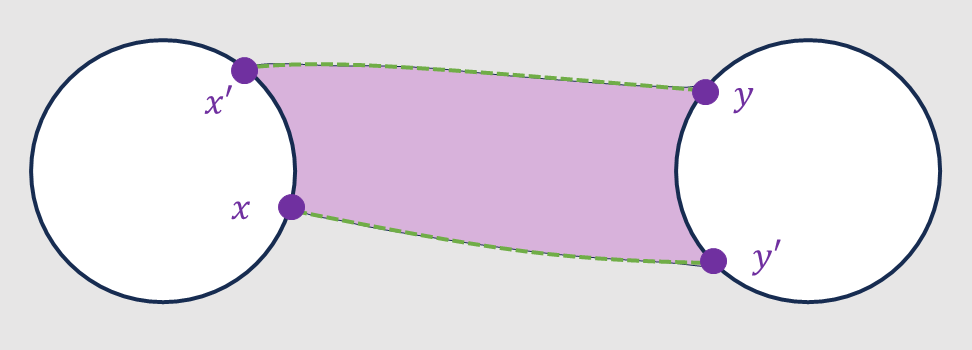}}}
\caption{An illustration of $\textsc{Area}((x,y),(x',y'))$ (pink). \label{fig: area_definition}}
\end{figure}  

\paragraph{Relative directions between critical paths.} 
Fix two faces $F_1$ and $F_2$, where $F_1$ has lower face index than $F_2$. Consider the canonical path $\gamma$ between a vertex $u\in F_1$ and a vertex $v\in F_2$.
Assume this canonical path is formed by primary paths $(u,v')$ and $(v,u')$, and their paired secondary paths are $(u,v'')$ and $(v,u'')$, respectively.
The claim is that when we move from $u$ to $v'$ along the path $(u,v')$ and reaches $w$, the subpath $(w,v)$ exits from its left side (see Figure \ref{fig:orientation}).

This can be proved as follows.
Since $F_1$ has lower face index, the path $u\to v'$ leaves from the right side of the path $u\to v''$ at $u$. Similarly, the path $v\to u'$ leaves from the left side of $v\to u''$ at $v$. From the figure, this immediately implies that when moving $u$ to $v'$, we have $(w,v)$ must be on the left side. For ease of reference later, we state the claim formally as an observation:

\begin{figure}[h]
\centering
\includegraphics[scale=0.4]{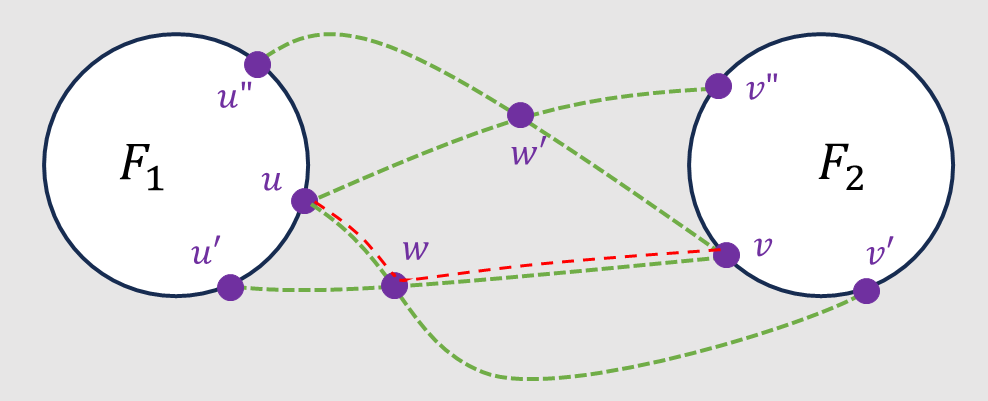}
\caption{An illustration of the shape of the $u$-$v$ canonical path (red).\label{fig:orientation}}
\end{figure}
\begin{observation}
Fix two faces $F_1$ and $F_2$, where $F_1$ has a lower face index than $F_2$. Let $(u,v')$ be a primary path from $u$ to $F_2$ and $(v,u')$ be a primary path from $v$ to $F_1$, that form the $u$-$v$ canonical path with bend $w$. Then when we move along $u\to v'$, the subpath $(w,v)$ exits from its left.
    \label{obs:orientations}
\end{observation}

\paragraph{$1$-bend paths.} 
We say a path $Q$ between terminals $a$ and $b$, where $F_a$ has lower face index than $F_b$, is a \emph{$1$-bend path} with bend $c$, iff (i) $Q[a,c]$ is a subpath of some primary path $P_1$ from $a$ to $F_b$; (ii) $Q[b,c]$ is a subpath of some primary path $P_2$ from $b$ to $F_a$; and (iii) when moving along $P$ from $a$ to the endpoint on $b$, $(c,b)$ exits from its the left (i.e., $Q$ satisfies the property  in  \Cref{obs:orientations}). We say that a 1-bend path $Q$ formed by primary paths $P_1$ and $P_2$ is  \emph{safe}, iff $\textsc{Area}(P_1,P_2)$ contains no terminal faces. All canonical paths are safe 1-bend paths, but the converse is not true.

\paragraph{Bad pairs.} 
Let $P$ be a critical path and $Q$ be a safe 1-bend path. We say $(P,Q)$ is a \emph{bad pair}, iff (i) there exist faces $F_a,F_b$, such that both $P,Q$ connect a terminal on $F_a$ to a terminal on $F_b$; (ii) $P$ and $Q$ cross more than once; and (iii) the endpoints of $P$ and $Q$ are four distinct terminals. Note that such a pair $P$ and $Q$ cross exactly twice, since pairs of critical paths cross at most once and $Q$ is formed by two critical paths. We say that $(P,Q)$ is a \emph{canonical bad pair} iff $(P,Q)$ is a bad pair and $Q$ is a canonical path.

\subsubsection{Description of an iteration}

While a canonical bad pair exists, our algorithm will iteratively decrease the number of bad pairs from $\hat H$ by moving some of the critical paths.
As canonical paths are defined as concatenations of subpaths of critical paths, they will also evolve during this iterative process. We will then show that the final $\hat{H}$ which contains no canonical bad pairs also must satisfy Property \ref{prop: intersect iff}. If in the current graph $\hat H$, there are no canonical bad pairs, then we return $\hat H$ and terminate the algorithm. Otherwise, the first step is to find a bad pair with some additional properties.

Let $(P,Q)$ be a bad pair. Suppose $P$ is a primary path from $a$ to $b$. 
Let $(c,f)$ and $(e,g)$ denote the primary paths forming $Q$, where $c$ lies on $F_a$ and $e$ lies on $F_b$. Let $d$ denote the bend of $Q$, and let $h$ and $i$ denote the two crossings between $P$ and $Q$, where $h$ is closer to $a$ and $i$ is closer to $b$. See Figure \ref{fig:notations} for an illustration.

\begin{figure}[h]
\centering
\includegraphics[scale=0.35]{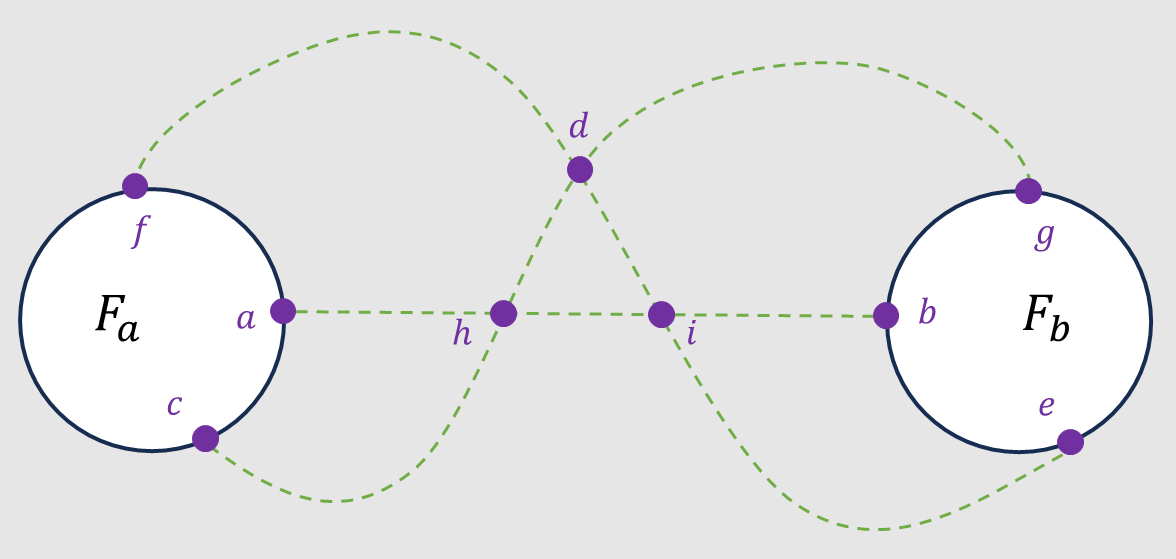}
\caption{Notations for relevant nodes induced by a bad pair.\label{fig:notations}}
\end{figure}


We say that $(P,Q)$ is a \emph{minimal bad pair}, iff
\begin{itemize}
    \item no critical path between $F_a$ and $F_b$ crosses $(h,d)$ and $(d,i)$ and
    \item no primary path from $F_b$ to $F_a$ crosses $(h,i)$ and forms a safe 1-bend path with $(a,b)$.
\end{itemize}

\paragraph{Step 1. Finding a minimal bad pair.} Given a canonical bad pair, we find a minimal bad pair through an iterative process as follows.
We maintain the bad pair $(\hat{P},\hat{Q})$ throughout the process, initialized as the given canonical bad pair $(P,Q)$, and repeat the following until convergence:
\begin{itemize}
    \item \textbf{Case 1. There exists a critical path $R$ between $F_a$ and $F_b$ which crosses $(h,d),(d,i)$.} We update $\hat P$ to be $R$, and then we update $a,b,\ldots,h,i,F_a,F_b$ with respect to $(\hat P,\hat Q)$ accordingly.
    \item \textbf{Case 2. There exists a primary path $R$ from $F_b$ to $F_a$ that crosses $(h,i)$ and forms a safe $1$-bend path with $(a,b)$.} Note that in this case, $R$ must also cross $(h,g)$, because $R$ must leave the $(h,g,b)$ region. We update $\hat Q$ to be the 1-bend path formed by $(c,g)$ and $R$. Then we update $a,b,\ldots,h,i,F_a,F_b$ with respect to $(\hat P,\hat Q)$ accordingly. Note that $\hat{Q}$ satisfies the property in Observation \ref{obs:orientations} since $R$ has the same structure as primary path $(e,f)$ which formed a 1-bend path with $(c,g)$, so $\hat{Q}$ is indeed a 1-bend path. Furthermore, by Observation \ref{obs:wing-no-face} below, $\hat{Q}$ is actually a safe 1-bend path. 
\end{itemize}

Now, we show that Step 1 indeed finds a minimal bad pair. We first need to define some language. For iteration $t$ of the process, let $a_t,b_t,\ldots,i_t$ denote the corresponding points in the graph. Let $\Delta_t=\textsc{Area}(h_t,d_t,i_t)$ and $W_t=\textsc{Area}(c_t,d_t,f_t)\cup\textsc{Area}(g_t,d_t,e_t)$.

\begin{observation}
    For any $t\ge0$, $\Delta_t$ and $W_t$ don't contain any terminal faces. \label{obs:wing-no-face}
\end{observation}
\begin{proof}
We prove the claim by induction on $t$. For the base case: $W_0$ and $\Delta_0$ don't contain any terminal faces since $Q$ is a canonical path. For the induction step: If $\hat{P}$ is updated, then $\Delta_{t+1}$ is a subset of $\Delta_t$ and $W_{t+1}=W_t$. By induction, $W_{t+1}$ and $\Delta_{t+1}$ don't contain any terminal faces. If $\hat{Q}$ is updated, let $R$ be the path which forms a 1-bend path with $(a_t,b_t)$ and causes the update to $\hat{Q}$. By our choice of $R$,  $\textsc{Area}(R,(a_t,b_t))$ contains no terminal faces. In {Figure \ref{fig:wing-induction}}, we see that in both cases, $\Delta_{t+1}\cup W_{t+1}$ is a subset of  $\Delta_{t}\cup W_t\cup\textsc{Area}(R,(a_t,b_t))$, completing the proof since the latter doesn't contain terminal faces. 
\end{proof}

\begin{figure}[h]
\centering
\subfigure
{\scalebox{0.22}{\includegraphics{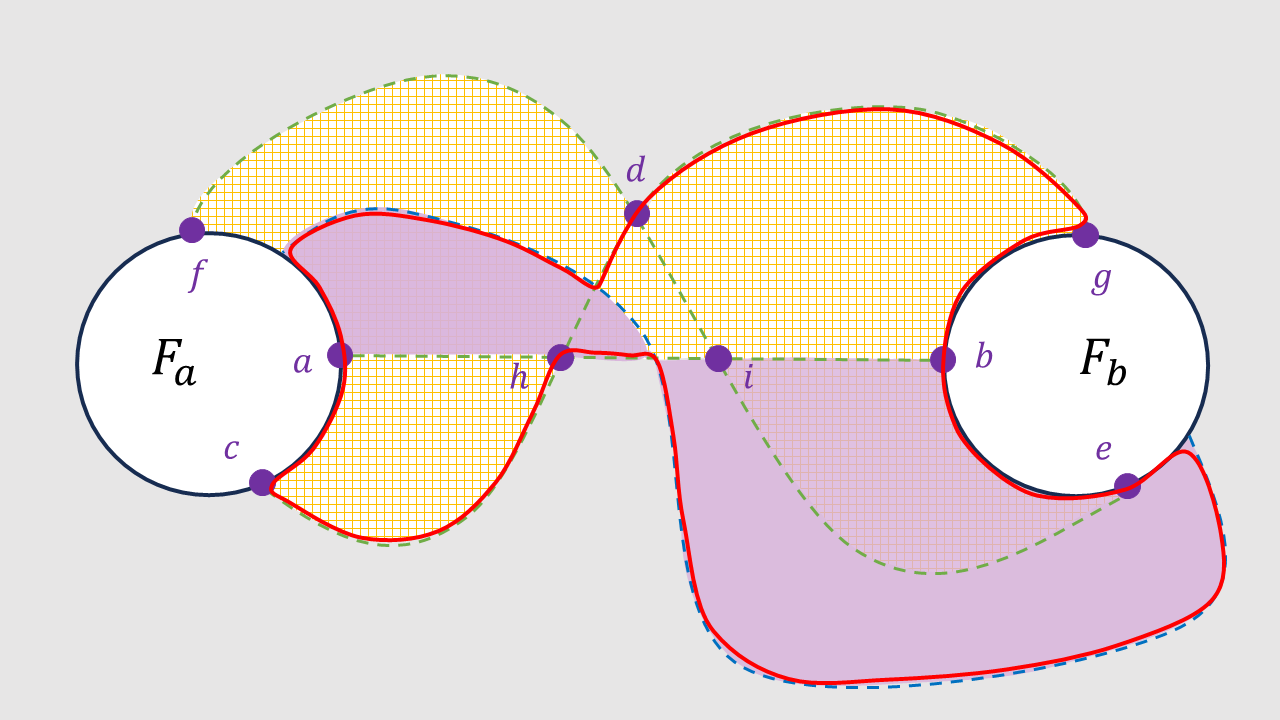}}}
\hspace{1.0cm}
\subfigure
{\scalebox{0.22}{\includegraphics{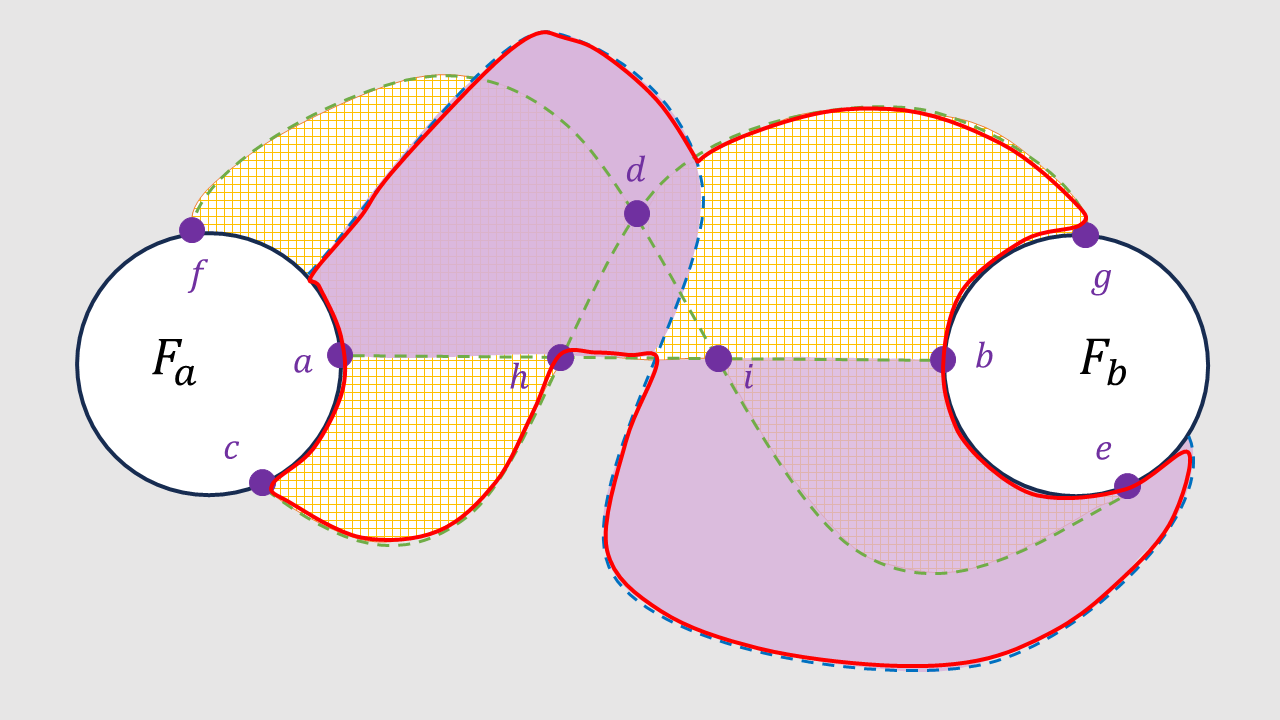}}}
\caption{In the figures, the pink shaded region is $\textsc{Area}(R,(a_t,b_t))$, the yellow striped region is $\Delta_{t}\cup W_t$, and the region outlined in a red curve is $\Delta_{t+1}\cup W_{t+1}$. In both cases, $\Delta_{t+1}\cup W_{t+1}$ is a subset of the union $\Delta_{t}\cup W_t\cup\textsc{Area}(R,(a_t,b_t))$.}\label{fig:wing-induction}
\end{figure} 

\begin{observation}
    $\textsc{Area}((a,b),(a_t,b_t))$ doesn't contain any terminal face.
\end{observation}
\begin{proof}
    We first prove that $\textsc{Area}((a_t,b_t),(a_{t+1},b_{t+1}))$ doesn't contain any terminal faces. There are two cases based on if $\hat P$ or $\hat Q$ is updated. When $\hat{Q}$ is updated, we have that $(a_t,b_t)=(a_{t+1},b_{t+1})$ so there is nothing to prove. When $\hat P$ is updated, we have that $(a_{t+1},b_{t+1})$ crosses through the two sides, $(h_t,d_t)$ and $(d_t,i_t)$, of the triangle. In particular, this means that $(a_{t+1},b_{t+1})$ crosses $(c_t,g_t)$ and $(e_t,f_t)$, so the entirety of $(a_{t+1},b_{t+1})$ lies inside $W_t\cup\Delta_t$. Since $(a_t,b_t)$ also lies in $W_t$, we have that $\textsc{Area}((a_t,b_t),(a_{t+1},b_{t+1}))$ also does. But since $W_t$ and $\Delta_t$ don't contain any terminal faces by Observation \ref{obs:wing-no-face}, this implies the claim.

    Now, we complete the proof by induction. The base case follows by the above discussion, since $(a,b)=(a_0,b_0)$ by definition. Suppose that we have $\textsc{Area}((a,b),(a_t,b_t))$ contains no terminal faces. We proved above that $\textsc{Area}((a_t,b_t),(a_{t+1},b_{t+1}))$ also contains no terminal faces. Recall that the area between two paths can be viewed as the minimal region which is passed over in a continuous deformation between the two paths. Since such a continuous deformation between $(a,b)$ and $(a_t,b_t)$ can be composed together with one between $(a_t,b_t)$ and $(a_{t+1},b_{t+1})$ to form one between $(a,b)$ and $(a_{t+1},b_{t+1})$, we have that $\textsc{Area}((a,b),(a_{t+1},b_{t+1}))$ is a subset of the union of $\textsc{Area}((a,b),(a_{t},b_{t}))$ and $\textsc{Area}((a_t,b_t),(a_{t+1},b_{t+1}))$. Since the latter two areas don't contain any faces, neither does the former area, completing the proof. See Figure \ref{fig:casework} for an illustration of some examples.
    \begin{figure}[h]
    \centering
    \subfigure
    {\scalebox{0.33}{\includegraphics{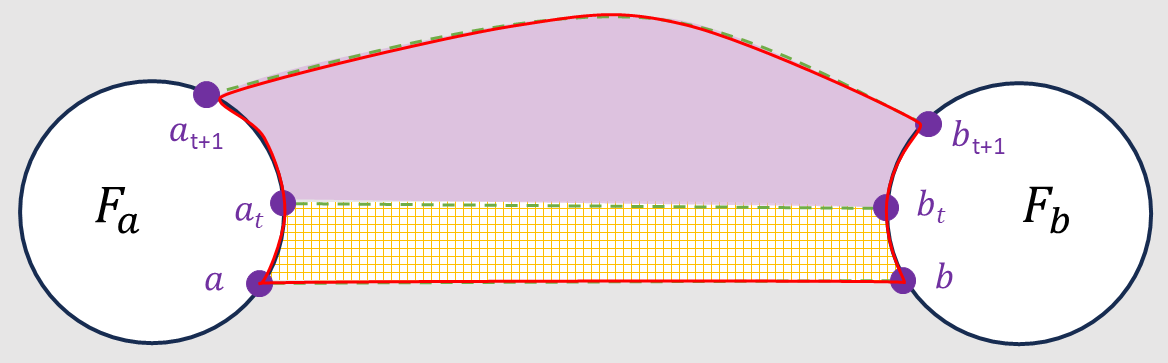}}}
    \hspace{1.0cm}
    \subfigure
    {\scalebox{0.33}{\includegraphics{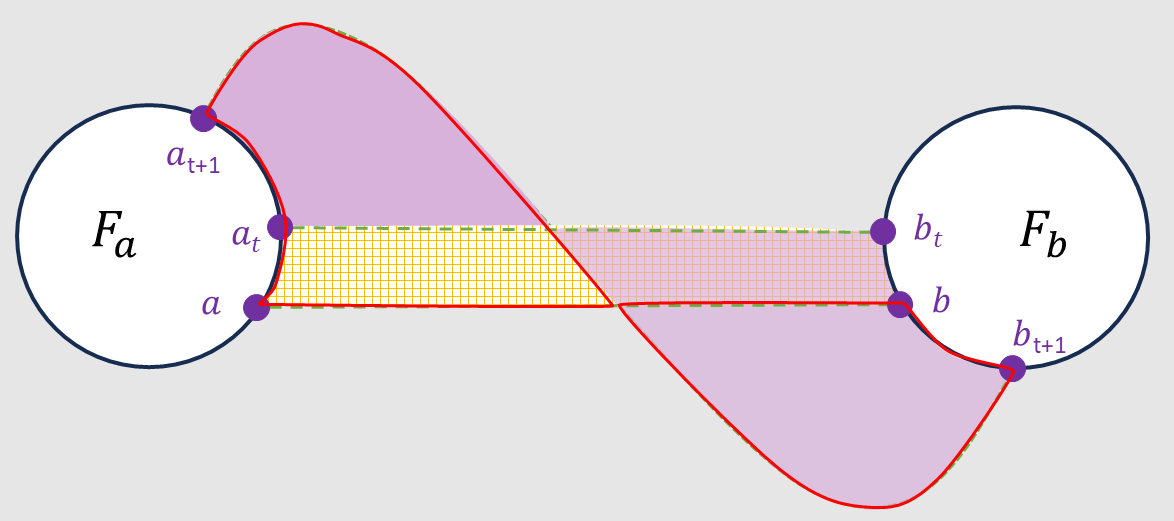}}}
    \hspace{1.0cm}
    \subfigure
    {\scalebox{0.33}{\includegraphics{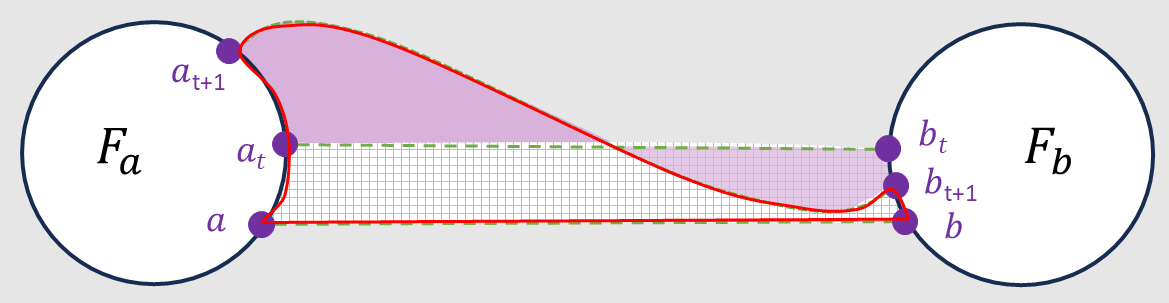}}}
    \caption{An illustration of examples of $\textsc{Area}((a,b),(a_{t},b_{t}))$ in orange, $\textsc{Area}((a_t,b_t),(a_{t+1},b_{t+1}))$ in pink, and $\textsc{Area}((a,b),(a_{t+1},b_{t+1}))$ outlined in red. In each of the three examples, we can see that $\textsc{Area}((a,b),(a_{t+1},b_{t+1}))$ is a subset of $\textsc{Area}((a,b),(a_{t},b_{t}))\cup\textsc{Area}((a_t,b_t),(a_{t+1},b_{t+1}))$. \label{fig:casework}}
    \end{figure}      
\end{proof}

\begin{claim}\label{cl:feng-ding}
    Let $b'$ denote the terminal on $F_b$ that lies clockwise right next to $b$ let $S$ denote the critical path between $a$ and $b'$ (see left one of Figure \ref{fig:feng-ding}). Then $\Delta_t$ never crosses $S$.
\end{claim}
\begin{figure}[h]
\centering
\subfigure
{\scalebox{0.22}{\includegraphics{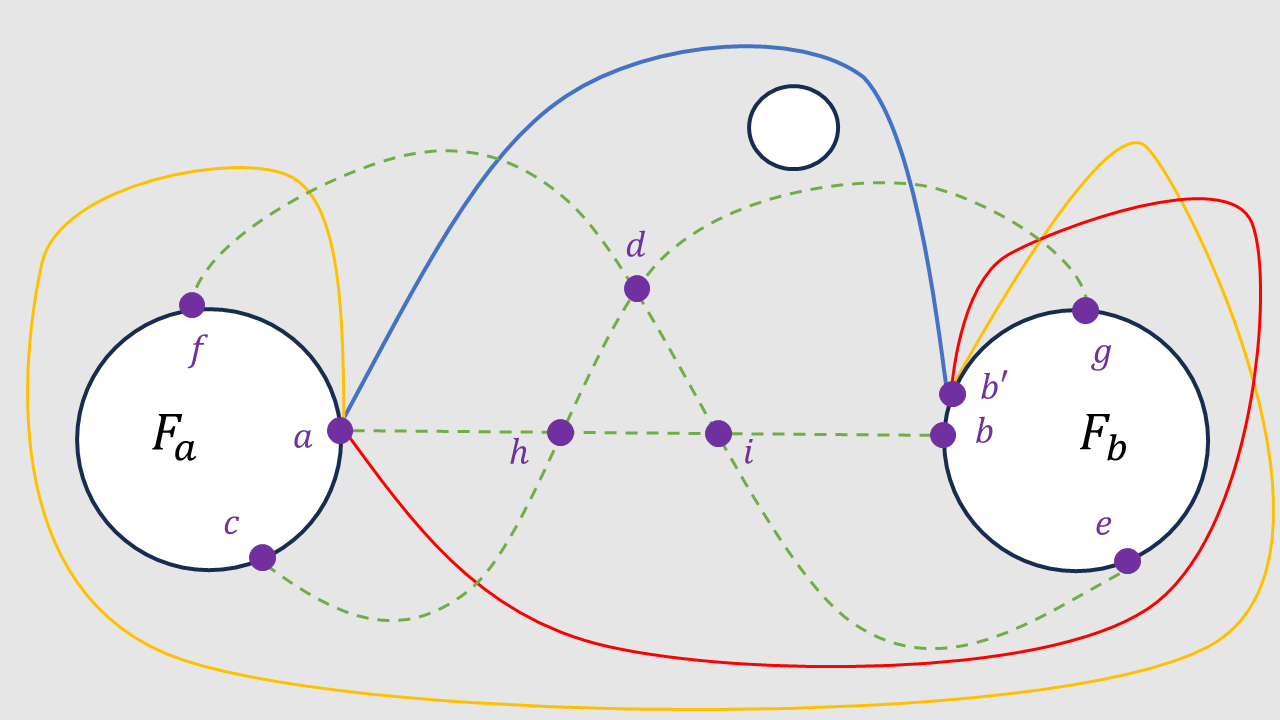}}}
\hspace{1.0cm}
\subfigure
{\scalebox{0.22}{\includegraphics{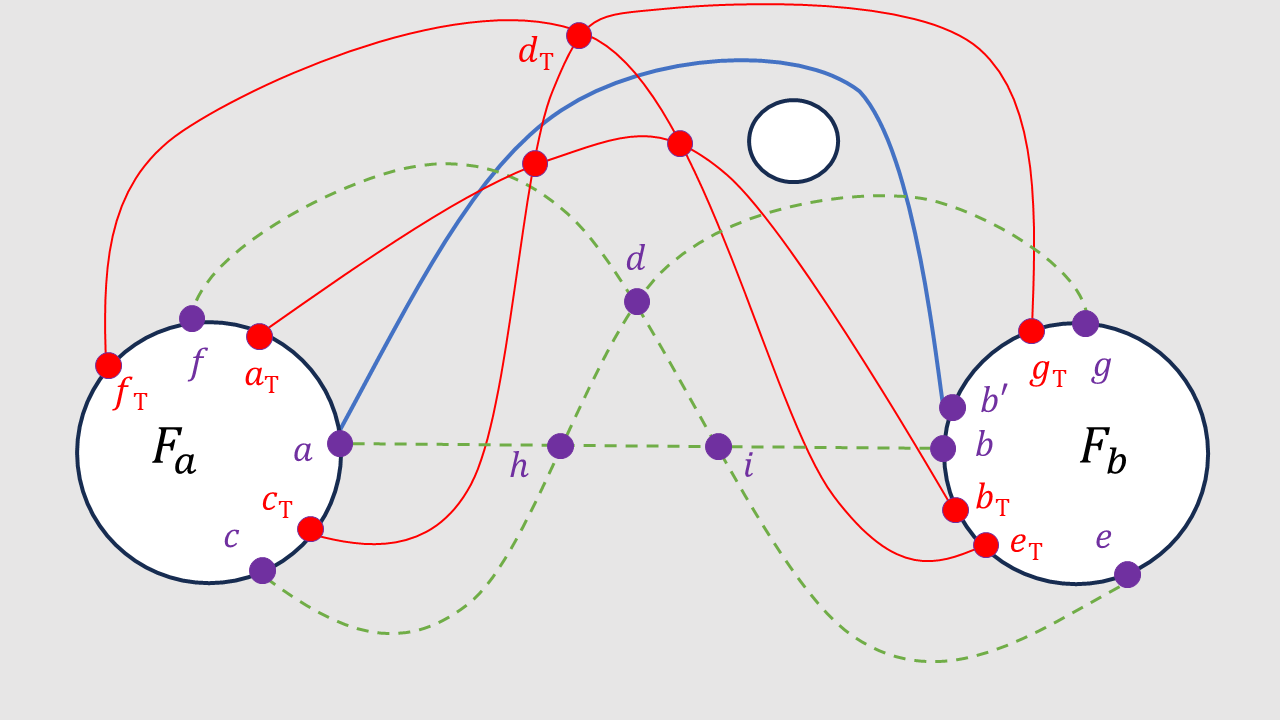}}}
\caption{In the left figure, we assert that the critical path $(a,b')$ must be the blue path. The other possibilities are illustrated by the red and orange drawings of $(a,b')$. The red path is not possible because it crosses $(c,g)$ twice and the yellow path is actually equivalent to the blue path. The right figure is for reference in the proof below.}\label{fig:feng-ding}
\end{figure} 
\begin{proof}
    First, observe that the region between $R=(a,b')$ and $(a,b)$ must contain a face $F$ because they are a pair of not equivalent critical paths. Next, suppose that some triangle crosses $S$. Consider the first such triangle; let's say it occurred at iteration $T$. Let $a_T,b_T,\ldots,i_T$ denote the analogous nodes induced by this triangle, such that $\textsc{Area}(h_T,d_T,i_T)$ is the triangle. Note that for the first time a triangle crosses $S$, we must have that $S$ crosses $(h_T,d_T)$ and $(d_T,i_T)$; otherwise the triangle would have crossed $S$ at iteration $T-1$ as well (see {Figure \ref{fig:triangle_crossing}, and the corresponding caption}).
\begin{figure}[h]
\centering
\subfigure
{\scalebox{0.22}{\includegraphics{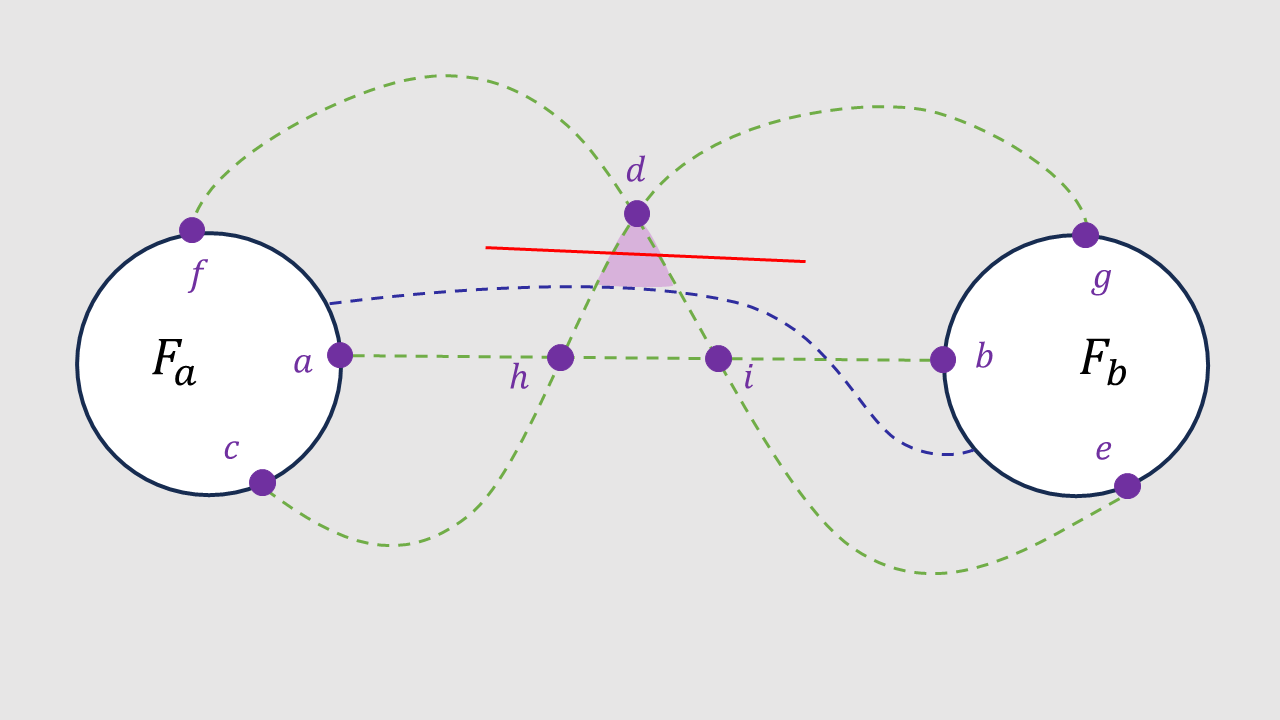}}}
\hspace{1.0cm}
\subfigure
{\scalebox{0.22}{\includegraphics{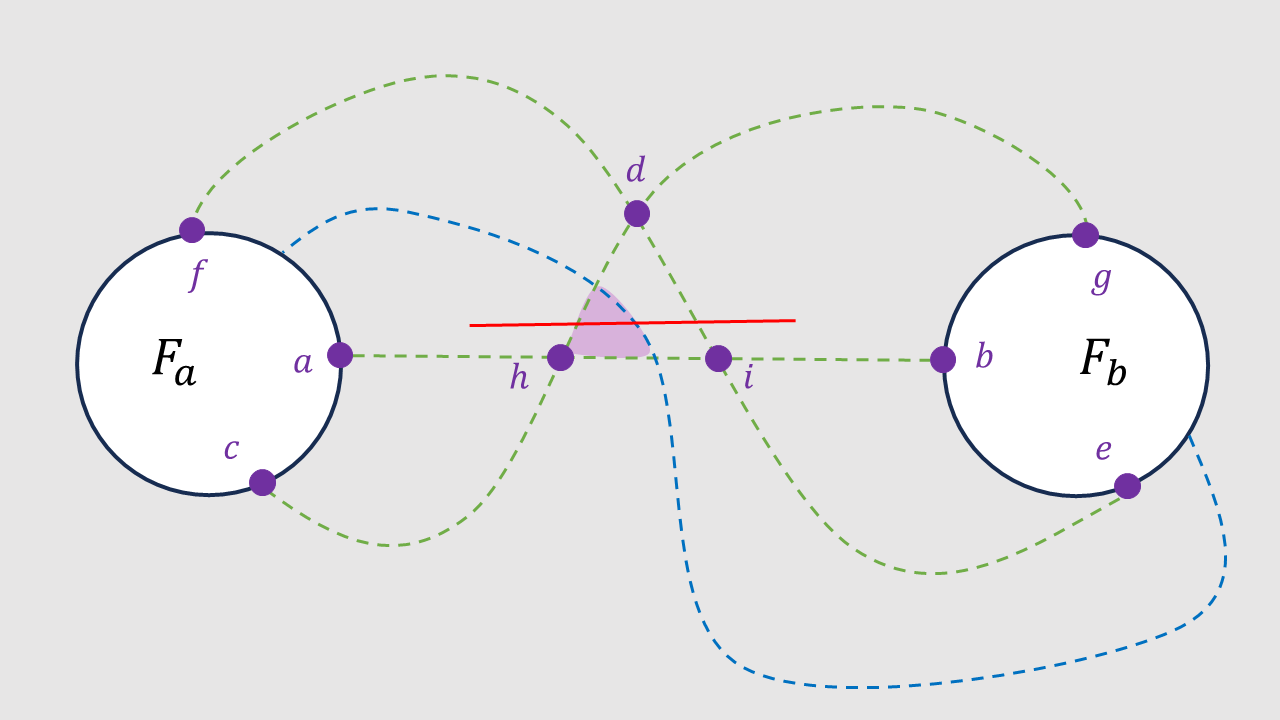}}}
\hspace{1.0cm}
\subfigure
{\scalebox{0.22}{\includegraphics{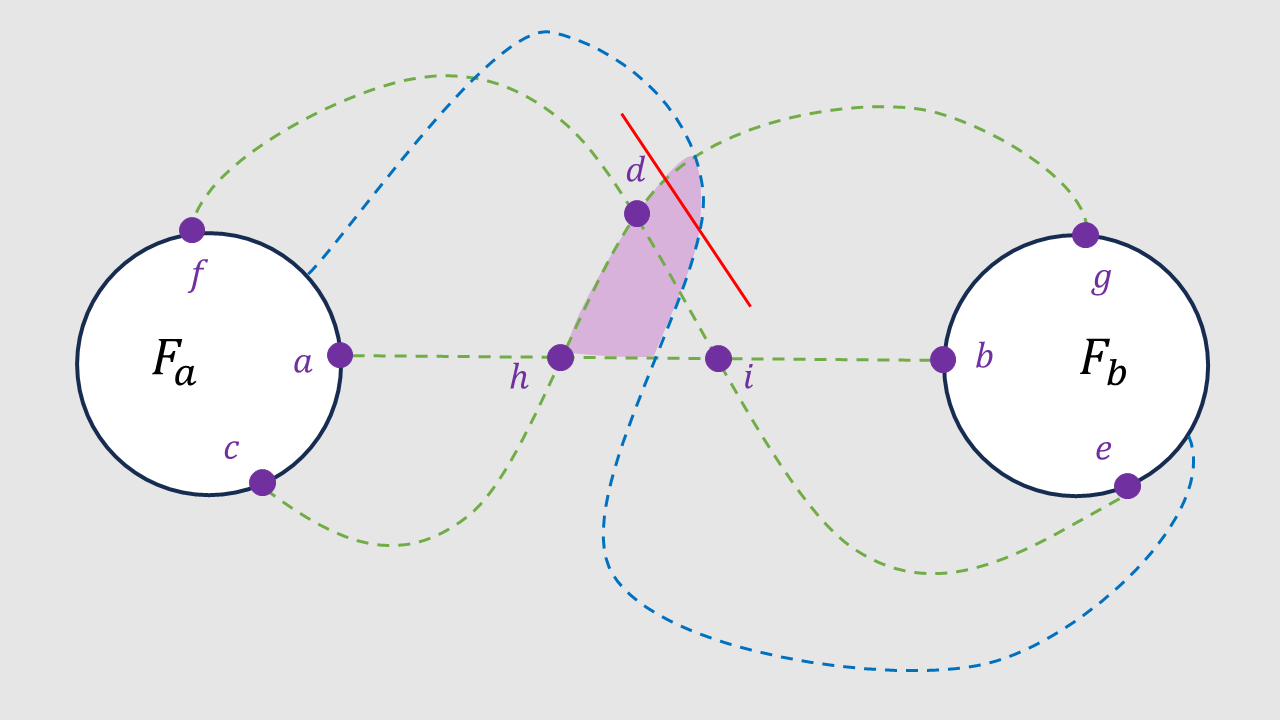}}}
\caption{In each of the figures, $(h,d,i)$ represents $\Delta_{T-1}$, the pink area represents $\Delta_{T}$, and the red curve represents the path $S$. In the top two figures, any path $S$ which crosses $\Delta_T$ twice must also cross $\Delta_{T-1}$ twice. In the bottom figure, a path $S$ which crosses $\Delta_T$ twice but doesn't cross $\Delta_{T-1}$ twice must look like the red path. In particular, note that the red path indeed crosses $(h_T,d_T)$ and $(d_T,i_T)$, as claimed in the proof.}\label{fig:triangle_crossing}
\end{figure} 

    We claim that the area $A$ between $f_T\to d_T\to g_T$ and $(a,b)$ contains the face $F$. The proof of the claim is illustrated in the right side of Figure \ref{fig:feng-ding}, and explained below. Since $(c_T,d_T)$ crosses $S$, we know that $(d_T,g_T)$ cannot cross $S$ again; otherwise, two critical paths would cross twice. Similarly, $(e_t,d_T)$ crosses $S$ so $(d_T,f_T)$ cannot cross $S$ again. Thus, $f_T\to d_T\to g_T$ doesn't cross $S$. Next, observe that $g_T$ is further than $b'$ from $b$ in the clockwise direction on $F_b$. 
    Since we have shown that $f_T\to d_T\to g_T$ doesn't cross $S$, the area between $f_T\to d_T\to g_T$ and $(a,b)$ must contain the area between $S$ and $(a,b)$. Since the latter area contains a face, the former area must as well.
    
    Note that the difference $A-\textsc{Area}((a_T,b_T),(a,b))$ is a subset of $W_T$. But from Observation \ref{obs:wing-no-face}, $W_T$ cannot contain a face, a contradiction. Hence, such a triangle $\textsc{Area}(h_T,d_T,i_T)$ cannot exist.
\end{proof}

\begin{observation}
    Let $U=\textsc{Area}((a,b),(a,b'))$. We have $\Delta_t\subseteq U$ for each iteration of Step 1.
\end{observation}
\begin{proof}
We prove by induction on $t$. The first triangle clearly lies inside $U$. For the induction step, there are two cases based on whether $\hat{P}$ or $\hat{Q}$ are modified. If $\hat{P}$ is modified, then the region enclosed by the triangle is a subset of that of the old triangle, so it must still lie inside $U$ by induction. If $\hat{Q}$ is modified, let $R$ denote the path which forms a safe 1-bend path with $(a,b)$ causing the update to $\hat{Q}$. In the new triangle, the three endpoints are $h_t$, the crossing between $(h_t,i_t)$ and $R$, and the crossing between $(h_t,g_t)$ and $R$. Label the endpoints of the new triangle $h_{t+1}$, $i_{t+1}$, and $d_{t+1}$ so that $h_{t+1}=h_t$,  The former two points lie inside the previous triangle, so they already lie in $U$. Since the triangle never crosses $R$, $(h_{t+1},d_{t+1})$ and $(d_{t+1},i_{t+1})$ can't cross $R$. We claim that $(h_{t+1},d_{t+1})$ and $(d_{t+1},i_{t+1})$ also can't cross $(a,b)$. If they did, Observation \ref{obs:orientations} would mean there are two critical paths which cross twice (see {Figure \ref{fig:stay_in_region}}). Since $R$ and $(a,b)$ are the boundary of the region $U$ and $(h_{t+1},d_{t+1})$ and $(d_{t+1},i_{t+1})$ never cross the boundary, each side of $\Delta_{t+1}$ lies inside $U$. This implies that  $\Delta_{t+1}$ lies inside $U$. 
\end{proof}

\begin{figure}[h]
\centering
\subfigure
{\scalebox{0.22}{\includegraphics{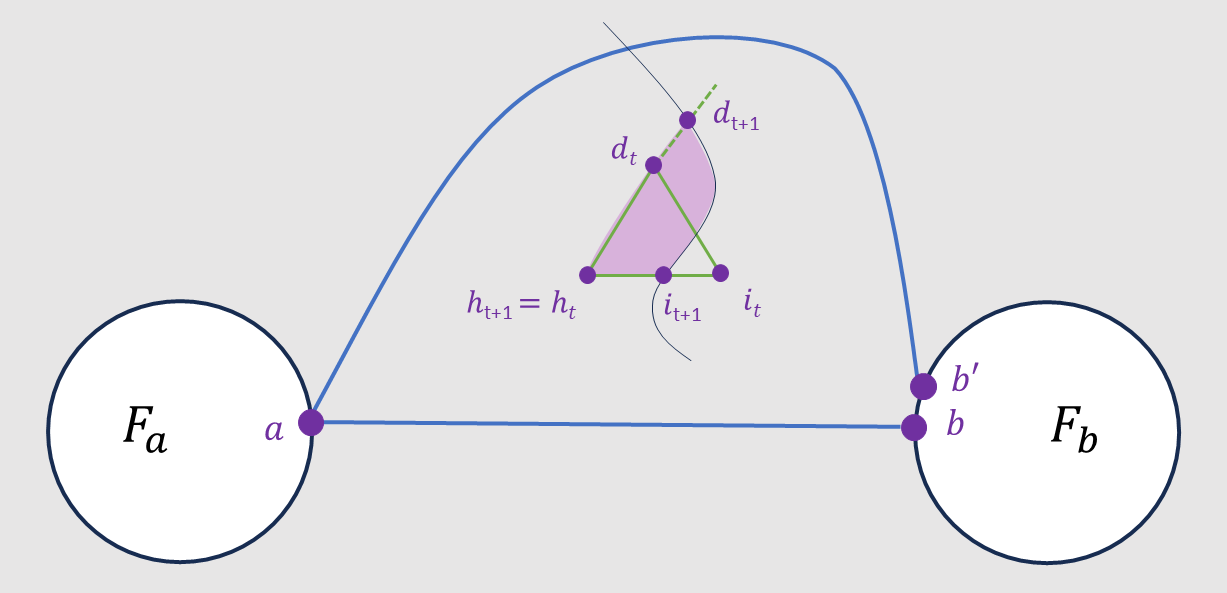}}}
\hspace{1.0cm}
\subfigure
{\scalebox{0.22}{\includegraphics{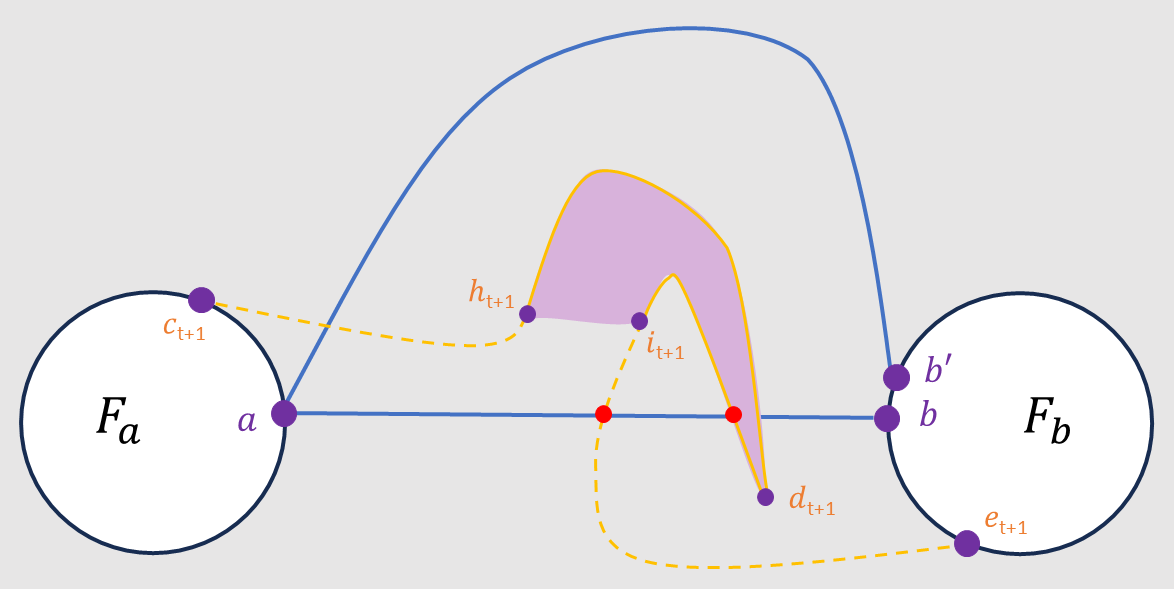}}}
\caption{The right figure illustrates what actually happens: the triangle $(h_{t+1},d_{t+1},i_{t+1})$ stays in the region $U$. The left figure illustrates what cannot happen: the triangle $(h_{t+1},d_{t+1},i_{t+1})$ crossing through $(a,b)$. This is not possible because it forces the subpath $(e_{t+1},d_{t+1})$ to cross $(a,b)$ twice, indicated by the red crossing nodes, a contradiction.\label{fig:stay_in_region}}
\end{figure} 

\begin{claim}\label{cl:step1-convergence}
    Step 1 converges efficiently to a minimal bad pair $(P,Q)$.
\end{claim}
\begin{proof}
    It suffices to prove Step 1 converges, since the $(\hat P,\hat Q)$ pair at convergence must be a minimal bad pair. Suppose the process didn't converge. Since 
    there are polynomially many triangles in $U$, triangles must repeat. Let $\Delta_{t_1}=\Delta_{t_2}$ be a pair of repeated triangles. Consider the union of areas $A_t:=\textsc{Area}((h_t,i_t),(h_{t+1},i_{t+1}))$ for $t_1\le t<t_2$. The border of this area starting at $h_{t_1}$ forms a loop which (1) starts and ends at $h_{t_1}$ and (2) stays in the region $U$. In particular, this implies that there must exist some iteration $t$ where $A_t$ crosses back through $(a_{t_1},b_{t_1})$ (see {Figure \ref{fig:loop}}). Let $t^*$ denote the first iteration between $t_1$ and $t_2$ where $A_{t^*}$ crosses back through $(a_{t_1},i_{t_1})$. By similar arguments in Claim \ref{cl:feng-ding} and {Figure \ref{fig:triangle_crossing}}, both $(h_{t^*},d_{t^*})$ and $(i_{t^*},d_{t^*})$ cross $(a_{t_1},b_{t_1})$. But Observation \ref{obs:orientations} implies such an $A_t$ cannot exist, since there would have to be two critical paths which cross each other twice (see discussion in the caption of Figure \ref{fig:loop}).
\end{proof}

\begin{figure}[h]
\centering
\subfigure
{\scalebox{0.22}{\includegraphics{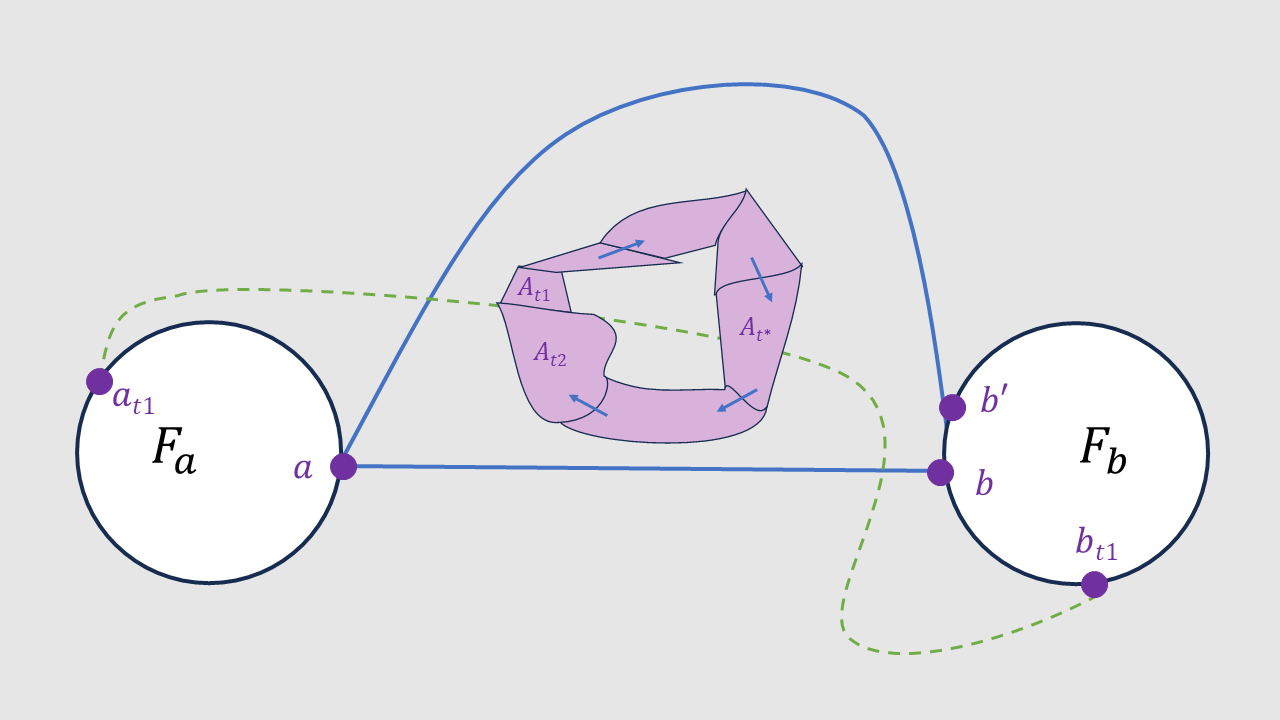}}}
\hspace{1.0cm}
\subfigure
{\scalebox{0.22}{\includegraphics{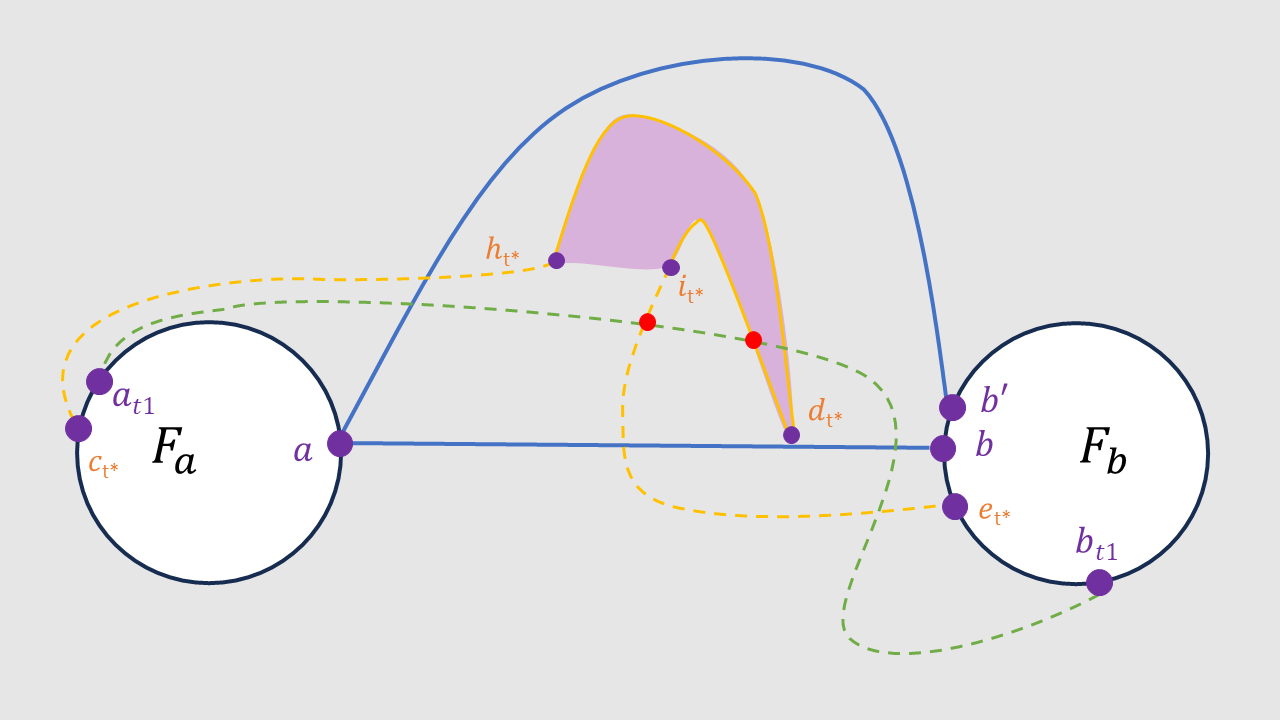}}}
\caption{The left figure illustrates the sequence of areas $A_t$. In particular, there must exist some $A_{t*}$ which crosses back through $(a_{t_1},b_{t_1})$. We claim that such an $A_{t*}$ cannot exist. As illustrated in the right figure, such an $A_{t*}$ would imply that the corresponding $(e_{t*},d_{t*})$ would cross with $(a_{t_1},b_{t_1})$ twice, indicated by the red crossing nodes, a contradiction. \label{fig:loop}}
\end{figure}

\paragraph{Step 2: Rerouting the minimal bad pair.}
First, we re-define some notation. Let $(a,b)$ denote the critical path $P$ and let $(c,g)$ and $(e,f)$ denote the critical paths which form $Q$, with $(c,g)$ and $(e,f)$ crossing at $d$ and $c$ lying on the same face as $a$. Let $h$ and $i$ denote the two crossings between $P$ and $Q$, where $h$ is closer to $a$ and $i$ is closer to $b$. From now on, we will refer to paths $P$ and $Q$ as $(a,b)$ and $c\to d\to e$, so that the notations $P$ and $Q$ can be used for arbitrary pairs of paths. For simplicity, we will denote $\textsc{Area}(h,i,d)$ as \emph{the triangle}. 

In order to fix the bad pair, we reroute $(a,b)$ to follow the path $(a,b)$ until it hits the triangle at $h$. When the path reaches $h$, we reroute it above the top angle of the triangle by first going along (but above) the $(h,d)$ subpath and then going along (but above) the $(d,i)$ subpath. When the rerouting hits $i$, we then follow the $(i,b)$ subpath of the original path $(a,b)$ until we reach $b$. This rerouting is illustrated in {Figure \ref{fig:2a}.} Clearly, the rerouted version of $(a,b)$ no longer crosses the $1$-bend path $c\to d\to e$ so we have fixed the bad pair. 

\begin{figure}[h]
    \centering
    \includegraphics[width=0.7\linewidth]{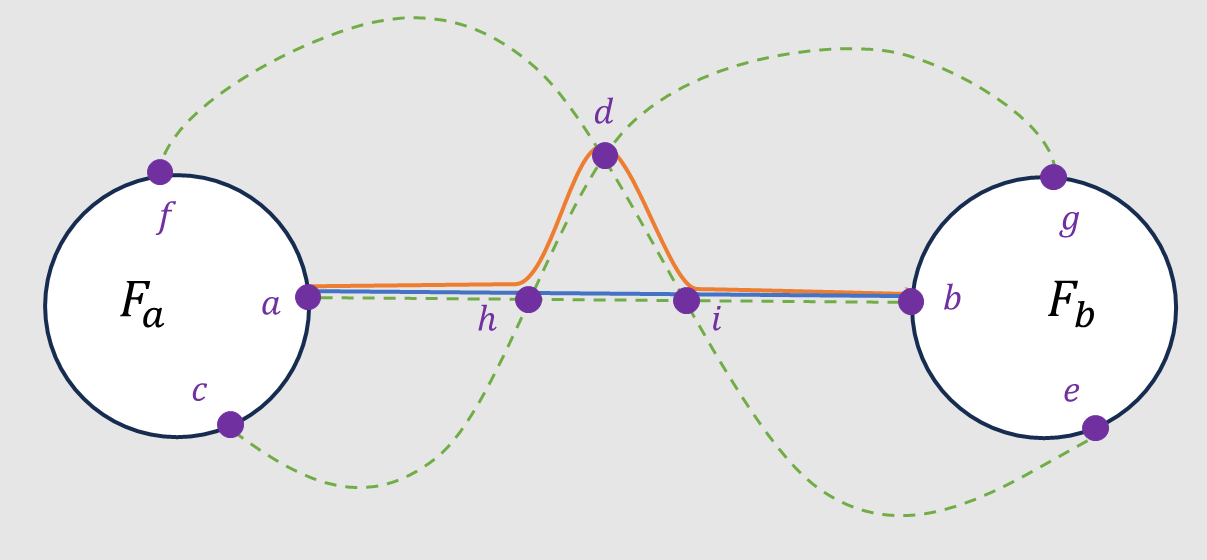}
    \caption{Illustrates the rerouting of $(a,b)$ over the triangle $(h,i,d)$.}
    \label{fig:2a}
\end{figure}

This completes the description of an iteration, and also completes the description of the algorithm. 

\subsubsection{Analysis of the algorithm}

\begin{observation}\label{obs:prop123}
    After one iteration, Properties \ref{prop: f-face instance}, \ref{prop: canonical intersection}, and \ref{prop: region} continue to hold.
\end{observation}
\begin{proof}
Property \ref{prop: f-face instance} continues to hold, as we have not modified the locations of any faces or terminals.
%
    For Property \ref{prop: canonical intersection}, consider any pair of critical paths $P$ and $Q$. If neither $P$ nor $Q$ are the $(a,b)$ path, the claim is trivial since they are not rerouted. If $P=(a,b)$, then we claim $Q$ crosses $P$ after the rerouting if and if $Q$ crossed $P$ before. To see this, observe that it suffices to show that $Q$ crosses $(h,i)$ if and only if $Q$ crosses either $(d,h)$ or $(d,i)$. But this follows by minimality of the triangle since no path can cross both $(d,h)$ and $(d,i)$, and any path must cross the boundary of the triangle an even number of times. 
    
    Next, we show that Property \ref{prop: region} continues to hold. For any two terminals $t_1$ and $t_2$, let $A$ denote the region between the $t_1$-$t_2$ canonical path and the $t_1$-$t_2$ shortest path before the iteration and let $A'$ denote the region after the iteration. Observe that any canonical path is only modified around an area which is a subset of $\textsc{Area}(h,i,d)$ which doesn't contain any terminal faces or terminals (Observation \ref{obs:wing-no-face}). Consequently, the difference between $A$ and $A'$ doesn't enclose any terminal or face. Since $A$ didn't contain any face or any other terminals, this implies inductively that $A'$ also doesn't contain any face or any other terminals. Consequently, Property \ref{prop: region} also continues to hold. 
\end{proof}

Next, we will show that Property \ref{prop: intersect iff} also holds for the final graph. We start by showing that the number of bad pairs decreases in each iteration in following arguments.

\begin{observation}\label{obs:same-faces}
If two safe 1-bend paths cross twice, their endpoints  lie on the same two terminal faces.
\end{observation}
\begin{proof}
    Let $P$ be a safe 1-bend path from $u$ to $v$ with bend $w$ and let $Q$ be a safe 1-bend path from $x$ to $y$ with bend $z$. Let $(u, v')$ and $(v, u')$ denote the critical paths forming $P$ and let $(x, y')$ and $(y, x')$ denote the critical paths forming $Q$. Assume for contradiction that the endpoint terminals of $P$ and $Q$ lie on at least three terminal faces. Then there must exist at least one endpoint of $P$ which lies on a different face than the other endpoints. Without loss of generality, assume $F_v\neq F_x,F_y$.

    Since $P$ and $Q$ cross twice and neither $\textsc{Area}(u,w,u')\cup \textsc{Area}(v,w,v')$ nor $\textsc{Area}(x,z,x')\cup \textsc{Area}(y,z,y')$ contains any terminal faces, $P$ crosses either $(x, y')$ or $(y, x')$ twice as well. Without loss of generality, let us say that $(u,w)$ and $(w,v)$ each crosses $(x,y')$. Since $F_v\neq F_x,F_y$ and $\textsc{Area}(u,w,u')\cup \textsc{Area}(v,w,v')$ doesn't contain any other terminal faces, we know that the terminals $x,y,x',y'$ don't lie in $\textsc{Area}(w,v,v')$. Combined with the fact that $(x,y')$ crosses $(w,v)$, this implies that $(x,y')$ also crosses $(w,v')$. But this means that $(x,y')$ crosses $(u,v')$ twice ($(u,w)$ and $(w,v')$), a contradiction. See Figure \ref{fig:two-face-reduction}.
    \begin{figure}[h]
        \centering
        \includegraphics[width=0.6\linewidth]{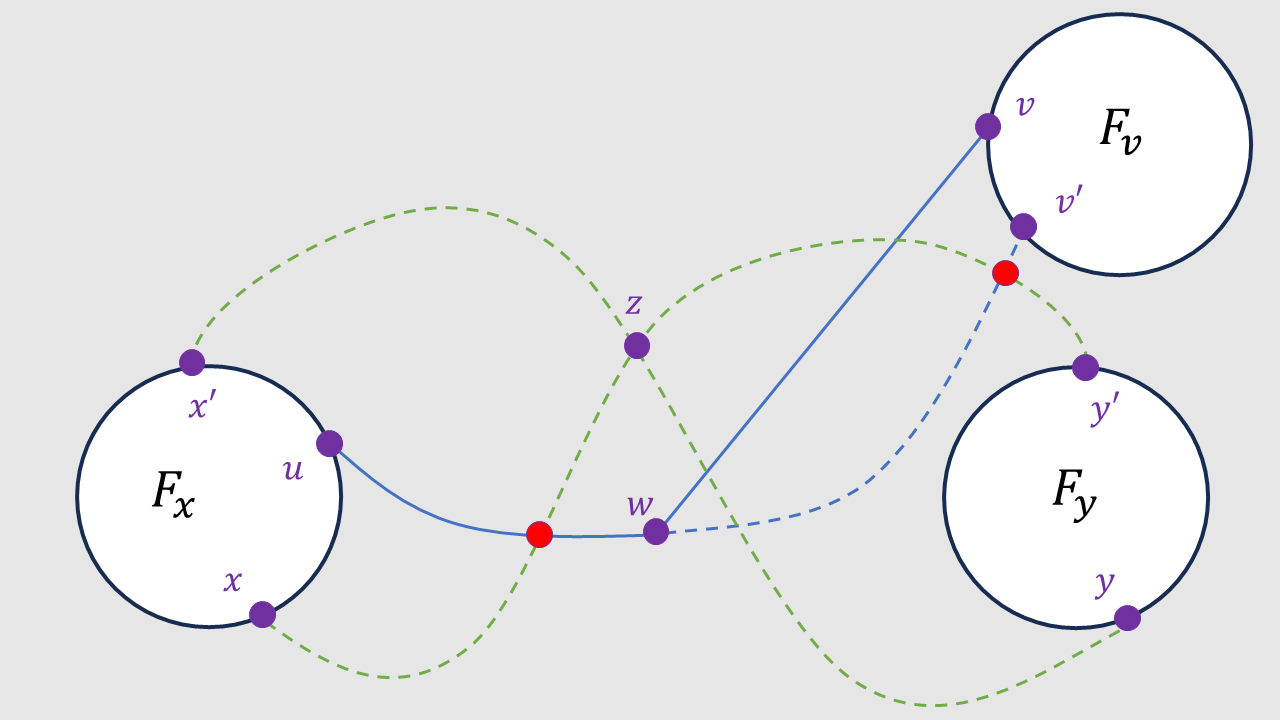}
        \caption{An illustration of the proof of Observation \ref{obs:same-faces}. The two crossings which must occur between $(u,v')$ and $(x,y')$ are identified in red.}
        \label{fig:two-face-reduction}
    \end{figure}
\end{proof}

\begin{claim}\label{cl:canonical-path-crossing-triangle}
Before the path $(a,b)$ gets rerouted, any safe $1$-bend path $u\to w\to v$ crosses $(h,i)$ at least as many times as it crosses $(i,d)$ and $(h,d)$ combined.
\end{claim}
\begin{proof}
    Recall that $a\to b$ is the primary direction of the $(a,b)$ critical path. The safe 1-bend path can cross the triangle either 0, 2, or 4 times; it is clear that 6 times is never possible, or else there would necessarily exists two critical paths which cross twice. If the path crosses the triangle 0 times, the claim is trivial. If the path crosses the triangle 4 times, observe that the path cannot cross $(h,d)$ and $(i,d)$ a total of 3 or 4 times since that would contradict the minimality of the triangle as at least one primary path forming the 1-bend path would need to cross both $(h,d)$ and $(i,d)$. But if the path crosses $(h,d)$ and $(i,d)$ twice or less, it would cross $(h,i)$ twice or more, and the desired claim holds. 
\begin{figure}[h]
\centering
\subfigure
{\scalebox{0.22}{\includegraphics{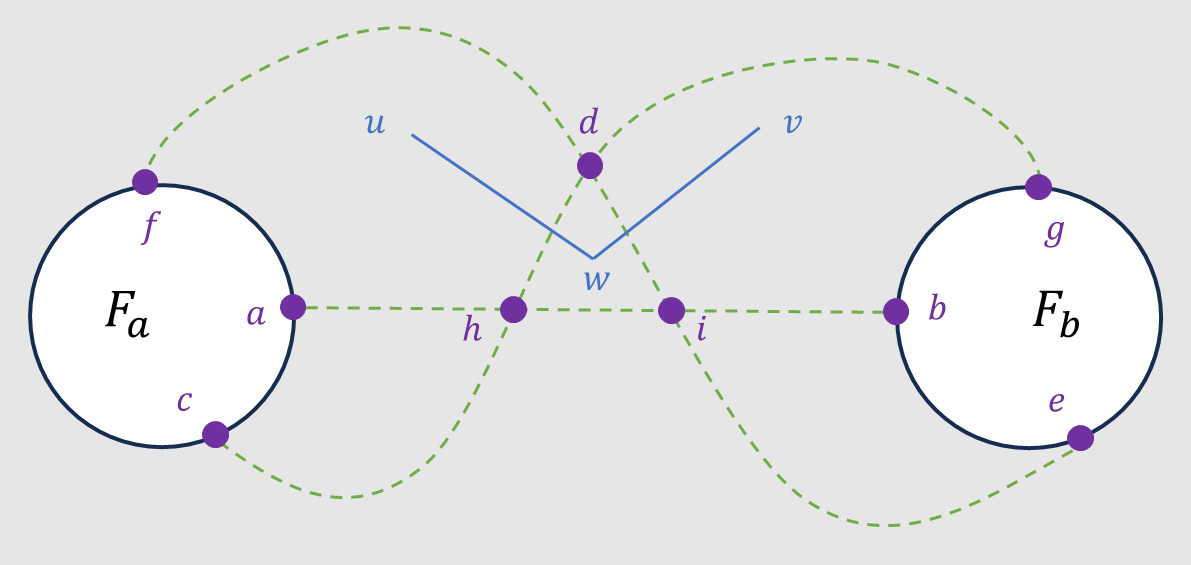}}}
\hspace{1.0cm}
\subfigure
{\scalebox{0.22}{\includegraphics{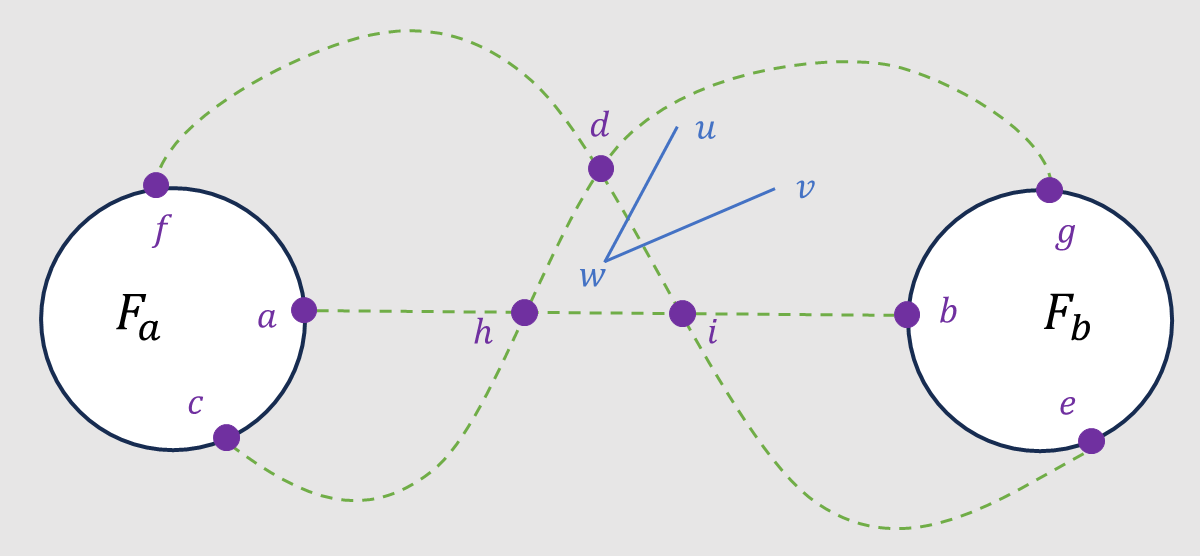}}}
\hspace{1.0cm}
\subfigure
{\scalebox{0.22}{\includegraphics{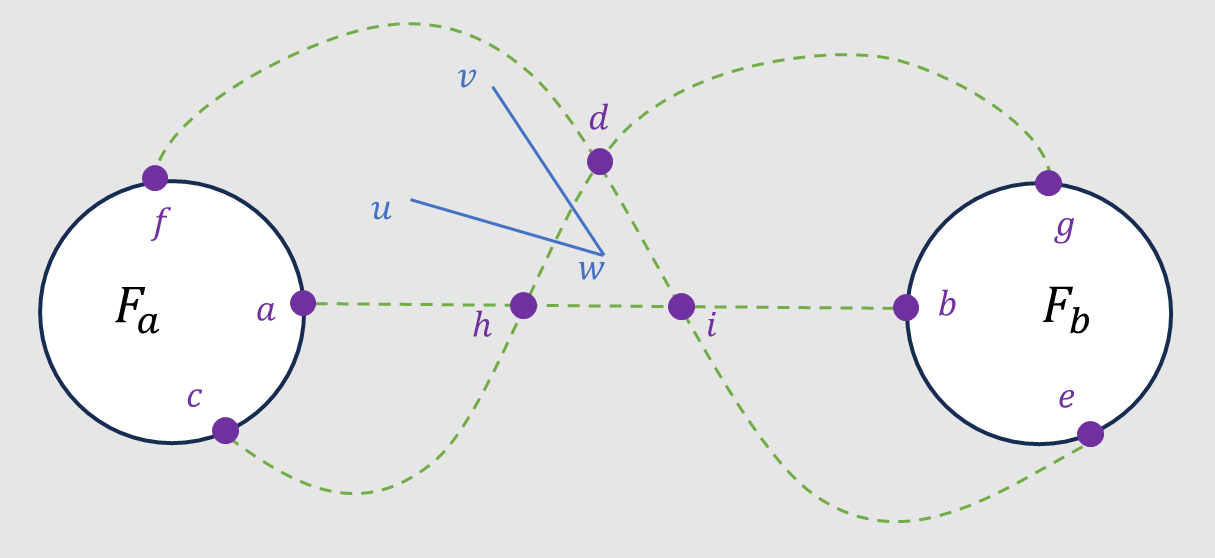}}}
\hspace{1.0cm}
\subfigure
{\scalebox{0.22}{\includegraphics{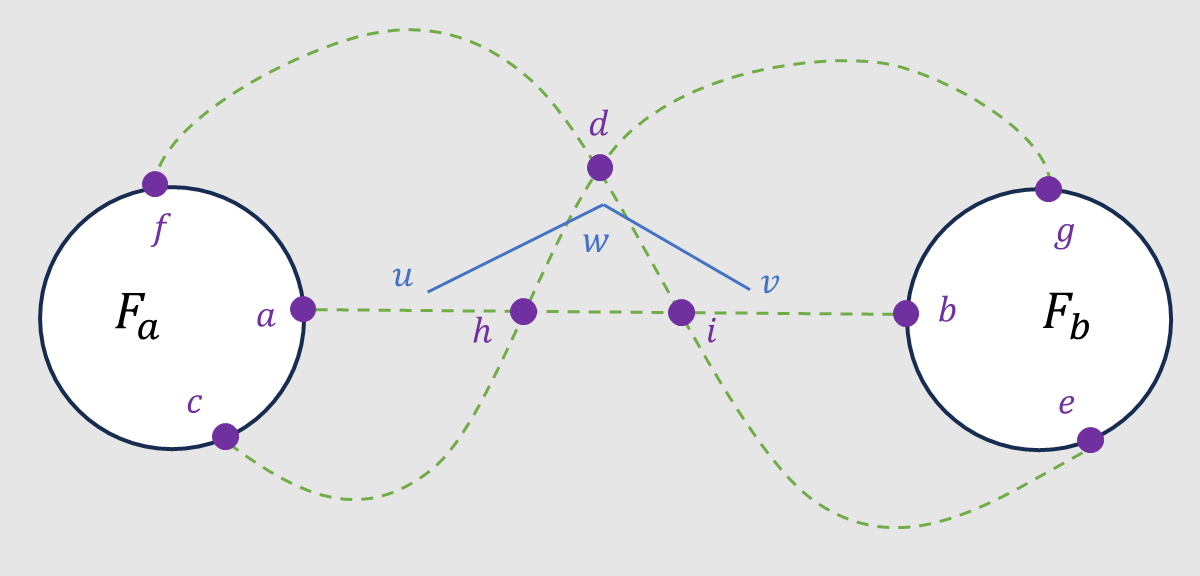}}}
\caption{In all subfigures, $(u,v)$ and $(v,w)$ are represent by blue lines. The former three cases are not possible by a similar argument as in the caption of \ref{fig:triangle_crossing} via Observation \ref{obs:orientations}. The final case is not possible by minimality of the triangle.}\label{fig:4a}
\end{figure}     
    Finally, consider a safe 1-bend path which creates two crossings with the triangle. By Observation \ref{obs:same-faces}, we may assume without loss of generality that $F_a=F_u$ and $F_b=F_v$. We show that any safe 1-bend path crossing either $(h,d)$ twice, $(i,d)$ twice, or $(h,d)$ and $(i,d)$ contradicts the structure of the 1-bend paths in Observation \ref{obs:orientations}, or the minimality of the $(P,Q)$ bad pair. This is illustrated for the four different cases in the {Figure \ref{fig:4a}}. In the first case, the existence of such a 1-bend path contradicts the minimality of the $(P,Q)$ bad pair. In the remaining three cases, the $c\to d\to e$ and $u\to w\to v$ canonical paths are oriented differently while being between the same two faces, implying that at least one of them contradicts the Observation \ref{obs:orientations} by a similar argument as the proof of Claim \ref{cl:step1-convergence}. Thus, the only possibilities which remain for a 1-bend path to cross the triangle twice are: $(h,d)$ and $(h,i)$, $(h,i)$ and $(h,i)$, or $(h,i)$ and $(i,d)$, satisfying the claim.
\end{proof}

Using the above claims and observations, we show that the number of bad pairs decreases.

\begin{figure}[h]
\centering
\subfigure
{\scalebox{0.22}{\includegraphics{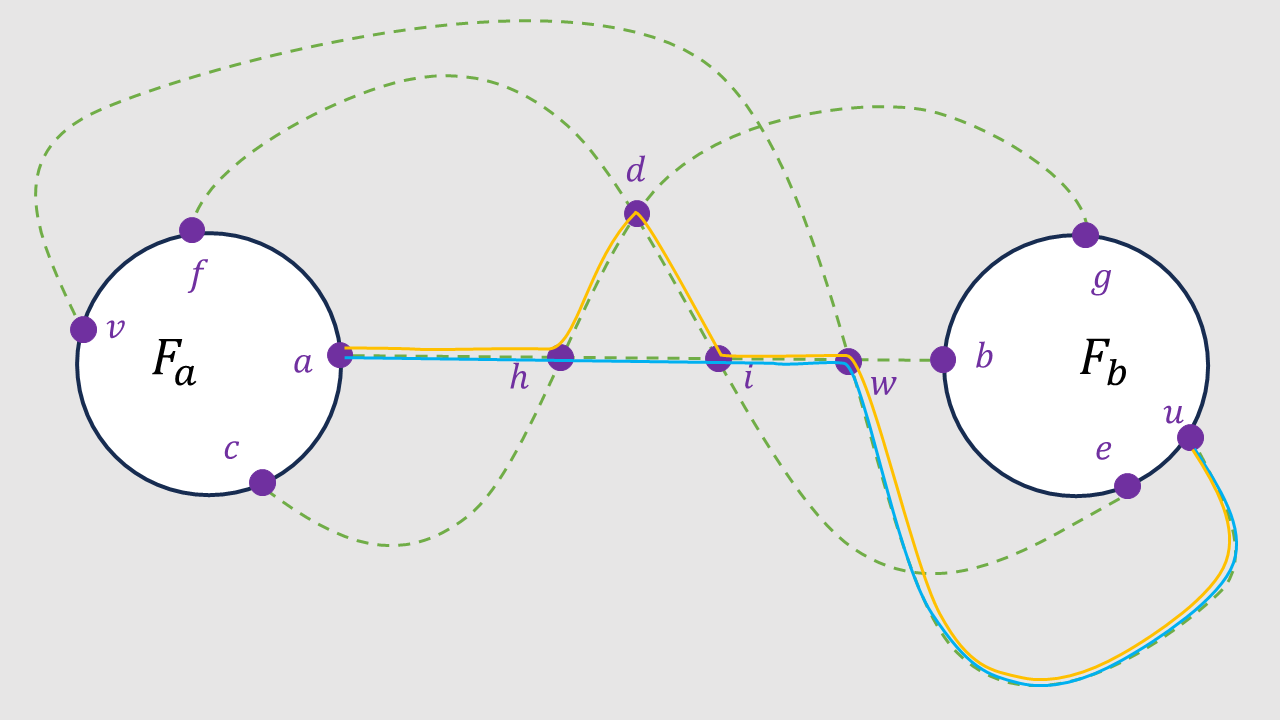}}}
\hspace{1.0cm}
\subfigure
{\scalebox{0.22}{\includegraphics{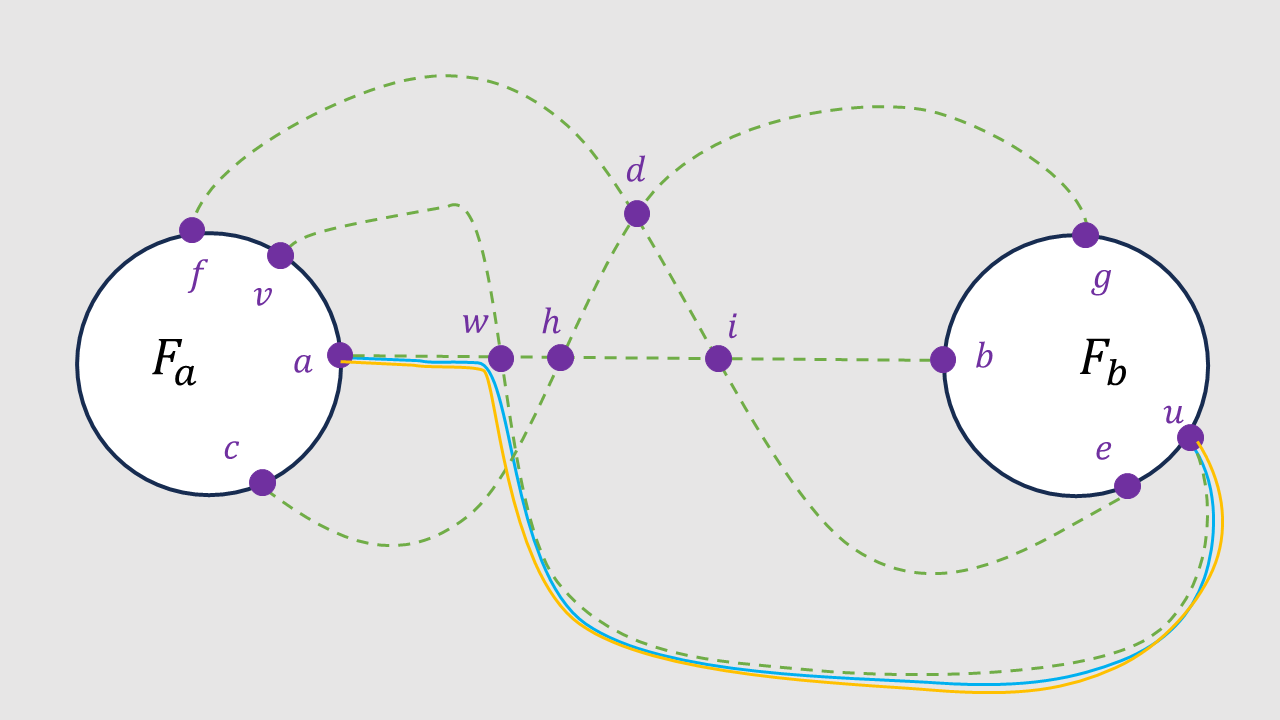}}}
\caption{In both figures, the blue path represents $Q$ before the rerouting and the orange path represents $Q$ after the rerouting.}\label{fig:13}
\end{figure} 

\begin{observation}\label{obs:bad-pairs-decrease}
    The number of bad pairs $(P,Q)$ decreases after each iteration.
\end{observation}
\begin{proof}
    The $(a,b)$ critical path no longer crosses $c\to d\to e$, so we have fixed a bad pair. Consider any other pair of paths $P$ and $Q$, where $P$ is a critical path and $Q$ is safe 1-bend path. If neither $P$ nor $Q$ involves the $(a,b)$ path, then the claim is trivial as the paths $P$ and $Q$ are not changed in the rerouting. 
    If $P$ involves $(a,b)$, then $P=(a,b)$. To see that $P$ crosses $Q$ fewer times after the rerouting than before, it suffices to show that $Q$ crosses $(d,i)$ or $(d,h)$ at most as many times as it crosses $(h,i)$, which follows from Claim \ref{cl:canonical-path-crossing-triangle}, so $P$ and $Q$ cannot form a new bad pair after the rerouting.
    Finally, suppose $Q$ involves $(a,b)$. Let $(u,v)$ be the other primary path that forms $Q$, and let $w$ be the bend of $Q$. This case is illustrated in Figure \ref{fig:13}. Since $(a,b)$ is primary, $(a,b)$ and $c\to d\to e$ are a minimal pair, so no primary path $P'$ from $F_b$ to $F_a$ which crosses the bottom of the triangle can form a safe 1-bend path with $(a,b)$. In particular, this implies that the bend $w$ cannot lie on $(h,i)$, so it must either lie on $(a,h)$ or $(i,b)$. If $w$ is on $(a,h)$, then $Q$ is not changed after the rerouting since the $(a,h)$ subpath is not changed in the rerouting and $(u,w)$ is not rerouted (see right of Figure \ref{fig:13}). If $w$ is on $(i,b)$, then the rerouting of $Q$ only changes the $(h,i)$ subpath to go above the triangle along $(h,d)$ and $(d,i)$ (see left of Figure \ref{fig:13}). In this case, observe that $P$ always crosses $(h,i)$ the same number of times as $(h,d)$ and $(d,i)$ combined due to the minimality of the triangle, so the number of crossings between $P$ and $Q$ cannot increase so $P$ and $Q$ will not form a new bad pair after the rerouting.
\end{proof}

It the end, we prove the following corollary which gives bounds on the size of $H^*$ and $H$.

\begin{corollary}
\label{clm: size}
$|V(H^*)|\le O(f^2k^2)$, and therefore $|V(H)|=O(f^2k^2)$.
\end{corollary}
\begin{proof}
All non-terminal vertices in $H^*$ are crossings between pairs of critical paths, so it suffices to count the number of crossings.
We have shown that the number of critical paths from each terminal $t$ is $O(f)$, so the total number of critical paths is $O(fk)$. Since two critical paths have at most one crossing, the total number of crossings is $O(f^2k^2)$. 

Recall that $H$ is obtained by sticking graphs $H^*,H_1,\ldots,H_f$ together.
From the construction, for $1\le r\le f$, $|V(H_r)|=O(|T_r|^2)$. Therefore,
\[|V(H)|\le |V(H^*)|+\sum_{1\le r\le f}O(|T_r|^2)=O(f^2k^2)+ O\bigg(\sum_{1\le r\le f}|T_r|\bigg)^2=O(f^2k^2)+O(k^2)=O(f^2k^2).\qedhere\]
\end{proof}

\section{The $f$-Face Case: Setting the Edge Weights}
\label{sec: edge weight}

As a last step, we need to set the edge weights in $H^*$ to preserve the distances between all interface terminal pairs. This is done in the following lemma, which, combined with the properties of graphs $H_1,\ldots,H_f$ and our construction of $H^*$, immediately implies that $H$ preserves the distances between all pairs of terminals in $G$.

\begin{lemma}
\label{lem: f face weight}
There is an edge weight function $w: E(H^*)\to \mathbb{R^+}$, such that
\begin{itemize}
    \item for each pair $t,t'$ of terminals lying on different faces, $\dist_{H^*}(t,t')= \dist_G(t,t')$; and
    \item for each pair $t,t'$ of terminals lying on the same face, $\dist_{H^*}(t,t')\ge \dist_G(t,t')$.
\end{itemize}
\end{lemma}

The remainder of this section is dedicated to the proof of \Cref{lem: f face weight}.

\subsection{The framework}
\label{subsec: weight overview}

Instead of giving clean formula for directing computing edge weights as in \cite{ChangO20}, here we adopt a different approach: we characterize the required properties of such an edge weight function by an LP, and then show its existence by proving the feasibility of the LP. 
Such an approach is recently used in \cite{chen2025path}, and we believe it will find more applications in distance-based graph problems.

\subsubsection*{Step 1. Computing edge weights by an LP}

We now describe the LP. The variables are $\set{x_e\mid e\in E(H^*)}$, where each $x_e$ represents the length we give to the edge $e$ in $E(H^*)$. There is no objective function, as we only care about its feasibility.
\begin{eqnarray*}
	\mbox{(LP-$H^*$)}	\quad
&\sum_{e\in E(Q)}x_e\geq \dist_G(t,t') &\forall t,t', \forall \text{ simple path }Q\text{  from }t \text{ to  }t'\\
&\sum_{e\in E(Q)}x_e\leq \dist_G(t,t') &\forall t,t', \forall \text{ canonical path }Q\text{ from }t \text{ to to }t'\\
	&x_e\geq 0&\forall e\in E(H^*)
\end{eqnarray*}

It is easy to verify that \Cref{lem: f face weight} is true iff (LP-$H^*$) is feasible.

\subsubsection*{Step 2. Reducing the feasibility of the LP to the non-existence of certain flows}

In order to show that (LP-$H^*$) is feasible, we use Farkas' Lemma, which takes the dual of (LP-$H^*$) and convert the distance-based problem into a flow-based problem.
For preparation, we introduce some notion on flows.

\newcommand{\pset}{{\mathscr{P}}}
\newcommand{\cost}{\mathsf{cost}}

\paragraph{Terminal flows.} 
Let $L$ be a graph with $T\subseteq V(L)$, and let $\mathcal{P}$ be the collection of paths in $L$ with both endpoints lying in the terminal set $T$. A \emph{terminal flow} $F: \mathcal{P}\to \mathbb{R}^+$ assigns each path $P\in\mathcal{P}$ with a flow value $F(P)$. 
For each edge $e\in E(L)$, we define $F_e=\sum_{P\in \mathcal{P}: e\in E(P)}F(P)$ as the total amount of flow of $F$ sent through edge $e$. We say that a flow $F$ \emph{dominates} another flow $F'$ if for each edge $e\in E(L)$, $F_e\ge F'_e$. 
For each pair $t,t'\in T$, we define $F_{t,t'}=\sum_{P \text{ connects } t\text{ to }t'}F(P)$ as the total amount of flow in $F$ sent between $t,t'$. 
We define the \emph{cost} of $F$ as $\cost(F)=\sum_{t,t'\in T}F_{t,t'}\cdot \dist_G(t,t')$.
Note that, in the remainder of this section we will consider terminal flows in either $L=H^*$ or $L=G$, but we always use the terminal distances $\dist_G(t,t')$ in $G$ to define their cost.

The following claim is a simple corollary of Farkas' lemma. Its proof is provided in \Cref{apd: Proof of clm: feasibility or flow}.

\begin{claim}
\label{clm: feasibility or flow}
Either \textnormal{(LP-$H^*$)} is feasible, or    
there exist terminal flows $F,F'$ in $H^*$, such that:
    \begin{itemize}
        \item $F$ only assigns non-zero values to canonical paths;
        \item $F$ dominates $F'$; and
        \item $\cost(F)<\cost(F')$.
    \end{itemize}\label{cl:farkas}
\end{claim}

\subsubsection*{Step 3. Proving the non-existence of certain flows}

From now on, we focus on proving that the terminal flows $F,F^\prime$ satisfying the conditions in Claim \ref{cl:farkas} cannot exist. Note that this implies the feasibility of (LP-$H^*$) (by \Cref{clm: feasibility or flow}) and then \Cref{lem: f face weight}.

We will assume for contradiction that such terminal flows exist, and eventually deduce a contradiction from them.
Note that, if, instead of having terminal flows $F,F'$ described in \Cref{clm: feasibility or flow}, we have terminal flows $\hat F,\hat F'$ in $G$, such that
\begin{properties}{P}
    \item $\hat F$ only assigns non-zero values to the shortest paths connecting terminals on different faces;\label{prop: shortest}
    \item $\hat F$ dominates $\hat F'$; and \label{prop: dominates}
    \item $\cost(\hat F)<\cost(\hat F')$, \label{prop: cost}
\end{properties}
then we get a contradiction easily. This is because, if we denote, for each edge $e\in E(G)$, by $\ell_e$ its length, then for each path $P$ in $G$ connecting a pair $t,t'$ of terminals, $\sum_{e\in E(P)}\ell_e\ge \dist_G(t,t')$, so
\[
\begin{split}
\cost(\hat F') =\sum_{t,t'\in T}\hat F'_{t,t'}\cdot \dist_G(t,t') & =\sum_{t,t'\in T}  \sum_{P \text{ connects } t\text{ to }t'}\hat F'(P)\cdot\dist_G(t,t')  \\
& \le \sum_{t,t'\in T}  \sum_{P \text{ connects } t\text{ to }t'}\hat F'(P)\cdot\sum_{e\in E(P)}\ell_e  \\
& =\sum_{e\in E(G)} \ell_e \cdot \sum_{P: e\in E(P)}\hat F'(P)  \\
& =\sum_{e\in E(G)} \ell_e \cdot \hat F'_e  \\
& \le \sum_{e\in E(G)} \ell_e \cdot \hat F_e   \quad\quad\quad\quad\quad\quad\quad\quad\quad\text{($\hat F$ dominates $\hat F'$)}
\\
& =\sum_{e\in E(G)} \ell_e \cdot \sum_{P: e\in E(P)}\hat F(P)  \\ 
& = \sum_{t,t'\in T}  \sum_{P \text{ connects } t\text{ to }t'}\hat F(P)\cdot\sum_{e\in E(P)}\ell_e  \\
& =\sum_{t,t'\in T}  \sum_{P \text{ connects } t\text{ to }t'}\hat F(P)\cdot\dist_G(t,t')  \\
& =\sum_{t,t'\in T}\hat F_{t,t'}\cdot \dist_G(t,t')
= \cost(\hat F) < \cost(\hat F'),
\end{split}
\]
where the last but one line used Property \ref{prop: shortest} that flow $\hat F$ only assigns non-zero values to the shortest paths connecting terminals on different faces (and so $\dist_G(t,t')=\sum_{e\in E(P)}\ell_e$ as $P$ is the shortest path connecting $t$ to $t'$).

Therefore, it remains to convert the given terminal flows $F,F'$ in $H^*$ into terminal flows in $G$ with Properties \ref{prop: shortest}-\ref{prop: cost}.
We first construct a flow $\hat F^*$ in $G$ as follows. For each pair $t,t'$ of terminals, we send $F_{t,t'}$ units of flow between $t,t'$ along the shortest path in $G$ connecting $t$ and $t'$. Clearly, this flow $\hat F^*$ satisfies Property \ref{prop: shortest}, and Property \ref{prop: cost} as by definition $\cost(F)=\cost(\hat F^*)$.
We need the following claim, whose proof 
is provided later.

\begin{claim}
\label{clm: routable}
There exists a flow $\hat F'$ in $G$, s.t. $\hat F'_{t,t'}=F'_{t,t'}$ for all $t,t'\in T$, and $\hat F^*$ dominates $\hat F'$.
\end{claim}

Let $\hat F'$ be the flow given by \Cref{clm: routable}.
Since $\hat F'_{t,t'}=F'_{t,t'}$ for every pair $t,t'\in T$, 
\[\cost(\hat F')=\sum_{t,t'}\hat F'_{t,t'}\cdot \dist_G(t,t')=\sum_{t,t'} F'_{t,t'}\cdot \dist_G(t,t')=\cost(F')<\cost(F)=\cost(\hat F^*).\]
Since $\hat F^*$ dominates $\hat F'$, they satisfy Properties \ref{prop: shortest}-\ref{prop: cost}. From the above discussion, the existence of such a pair $(\hat F^*,\hat F')$ of flows in $G$ causes a contradiction, and therefore implies \Cref{lem: f face weight}.


The remainder of this section is dedicated to the proof of \Cref{clm: routable}. To construct the flow $\hat{F}'$ in $G$ satisfying Properties \ref{prop: shortest}-\ref{prop: cost}, we will start with the original flows $F, F'$ in $H^*$ and morph them into flows $\hat{F}, \hat{F}'$ in $G$. Before going into details, we first explain our strategy at a high level. Imagine that we draw $H^*$ and $G$ on the same plane, so that the image of terminal vertices in $T$ are the same in $H^*$ and $G$. Then, we will go over all pairs of terminals $t, t'\in T$ and try to gradually change the drawing of the canonical path $\gamma_{t, t'}$ in $H^*$ to the drawing of their shortest path in $G$, and in the meantime we will update flows $F, F'$ in $H^*$ while maintaining Properties \ref{prop: shortest}-\ref{prop: cost}.
Eventually, $H^*$ will become the subgraph of $G$ induced by all critical paths, and then its flow $F'$ will automatically become a flow in $G$, which is the $\hat F'$ we want.

We now present the details. Our morphing process consists of two phases: graph splitting and flow morphing.

\subsection{Phase 1 in proving \Cref{clm: routable}: graph splitting}

For a pair $\gamma,\gamma'$ of canonical paths, if their endpoints are $4$ distinct terminals, then they are either vertex-disjoint or crossing at a single point. However, if $\gamma,\gamma'$ share an endpoint $t$, then it is possible that they share a subpath from $t$. Recall that our goal is to morph canonical paths $\gamma,\gamma'$ to their corresponding shortest paths $P,P'$ in $G$, and $P,P'$ may well be internally disjoint -- sharing only the endpoint $t$.
The goal of this preparation phase is to split the flows $F,F'$ and the drawing of all canonical paths in $H^*$ such that all canonical paths are edge-disjoint. 

We first prove the following lemma, whose purpose will be clear shortly.

\begin{lemma}\label{interval-packing}
    Consider two multi-sets of intervals of positive integers $\aset = \{[1, a_1], [1, a_2], \ldots, [1, a_k]\}$ and $\bset = \{[b_1, c_1], [b_2, c_2], \ldots, [b_l, c_l]\}$, such that for any integer $x$, the number of intervals in $\aset$ that contains $x$ is at least the number of intervals in $\bset$ that contains $x$. Then, there is a $\bset$-partition $\bset = \bset_1\cup \bset_2\cup \cdots \cup \bset_k$, such that for each $1\leq i\leq k$:
    \begin{itemize}
        \item all intervals in $\bset_i$ are disjoint; and
        \item the union of intervals in $\bset_i$ is contained in $[1, a_i]$.
    \end{itemize}
\end{lemma}
\begin{proof}
    This is proved by induction on the sum $a_1+a_2+\cdots+a_k$. The base case $a_1+a_2+\cdots+a_k = 1$ is trivial. In general, assume $a_1 = \max_{1\leq i\leq k}\{a_i\}$. If $a_1 > \max_{1\leq j\leq l}\{c_j\}$, then replace $a_1\leftarrow a_1-1$ and apply induction. Otherwise, without loss of generality, assume $c_1 = a_1$. Then, add $[b_1, c_1]$ to $\bset_1$ and apply induction on the instance $\aset' = \{[1, b_1-1], [1, a_2], \ldots, [1, a_k]\}, \bset' = \{[b_2, c_2], \ldots, [b_l, c_l]\}$.
\end{proof}

\paragraph{Flow packing.}
As $F$ may only use canonical paths, it admits a natural splitting: for each pair $s,t$ of terminals, it sends $F_{s,t}$ units of flow along the canonical path $\gamma_{s,t}$.
But $F'$ can use any path. To better understand its structure, we decompose it as follows.
For each $F'$-flow path $P$ (say from $s$ to $t$), we start from $s$ and follow the route of $P$. Since $H^*$ consists of critical paths, $P$ must start along some critical path $\rho$. Let $\rho[s,x]$ be the maximal subpath from $s$ shared by $\rho$ and $P$. We say $\rho[s,x]$ is the first segment of $P$, and denote it by $\beta_{s,x}$.
We then continue at $x$, find the next critical path that $P$ follows, collect another segment $\beta_{x,y}$. We repeat the same operation until we reach the endpoint $t$ of $P$, obtaining a collection $\set{\beta_{x,y}\mid x,y\in P}$ of segments.
We do the same for every other $F'$-flow path. 

For each canonical path $\gamma_{s, t}$ with bend $r$, we will compute a set $\Gamma_{s, r} = \{\beta_{x, y}\mid x, y\in \gamma_{s, r}\}$ where each $\beta_{x, y}$ is a segment obtained by decomposing $F'$, such that
\begin{itemize}
    \item each edge in $\gamma_{s,t}$ is contained in no more than $F_{s,t}$ segments in $\beta_{x, y}$; and
    \item the sets $\set{\Gamma_{s,r}\mid s,t\in T, r\in \gamma_{s,t}}$ partition all the segments obtained from decomposing $F'$.
\end{itemize}

\begin{figure}[h]
\centering
\includegraphics[scale=0.5]{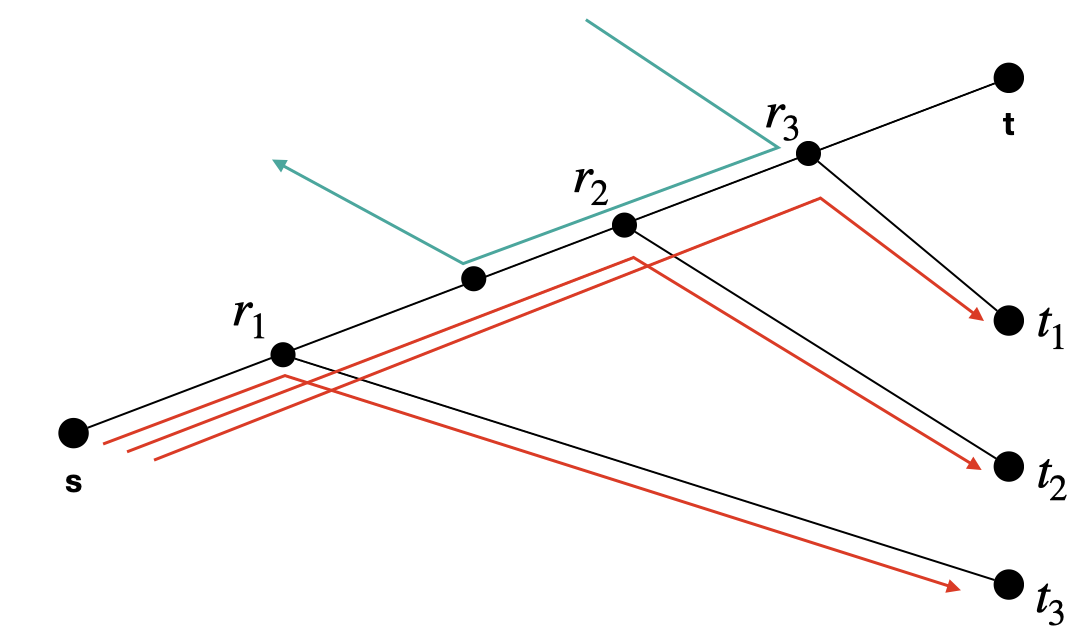}
\caption{\label{fig:flow-pack2} Each canonical path covers a prefix of the critical path $\rho_{s,t}$. Since $F$ dominates $F'$, we can treat each $F$-flow path from $s$ to $t_i$ as an interval $[1, a_i]$, and each $F'$-segment on $\rho_{s, t}$ as an interval $[b_j, c_j]$. Then we can apply \Cref{interval-packing} to assign maximal flow subpaths to canonical paths.}
\end{figure}

Fix any terminal vertex $s$, pick a critical path $\rho_{s, t}$, and consider all canonical path $\gamma_{s, t_1}, \gamma_{s, t_2}, \ldots, \gamma_{s, t_k}$ formed by $\rho_{s, t}$ and some other critical path. So by construction of $H^*$, all vertices $t_1, t_2, \ldots, t_k$ are on the same face as $t$.
Let $r_i$ be the bend of $\gamma_{s, t_i}, \forall 1\leq i\leq k$, and assume $r_1, r_2, \ldots, r_k$ are indexed according to their distances from $s$ in the increasing order. 
Let $\Gamma_s$ be the set of $F'$-segments which are subpaths of $\rho_{s, t}$. By definition of $F$, we send $F_{s, t_i}$ units of flow from $s$ to $r_i$ along $\rho_{s, t}$, and the union of all these flows should dominate the flow of $\Gamma_s$. If we view each flow of $F_{s, t_i}$ from $s$ to $r_i$ as an interval, and each flow in $\Gamma_s$ also as an interval, then by applying \Cref{interval-packing} we can build the sets $\Gamma_{s, r_i}, 1\leq i\leq k$ with the desired properties. See \Cref{fig:flow-pack2} as an example.

\paragraph{Graph splitting.} 
We now describe the iterative process for splitting the canonical paths in $H^*$. Initialize $H\leftarrow H^*$, $\hat{F}\leftarrow F$.
Go over each canonical path $\gamma_{s, t}$ in $H$ with bend $r$ and perform the following operations. See \Cref{fig:flow-split1}.

\begin{figure}[h]
\centering
\includegraphics[scale=0.47]{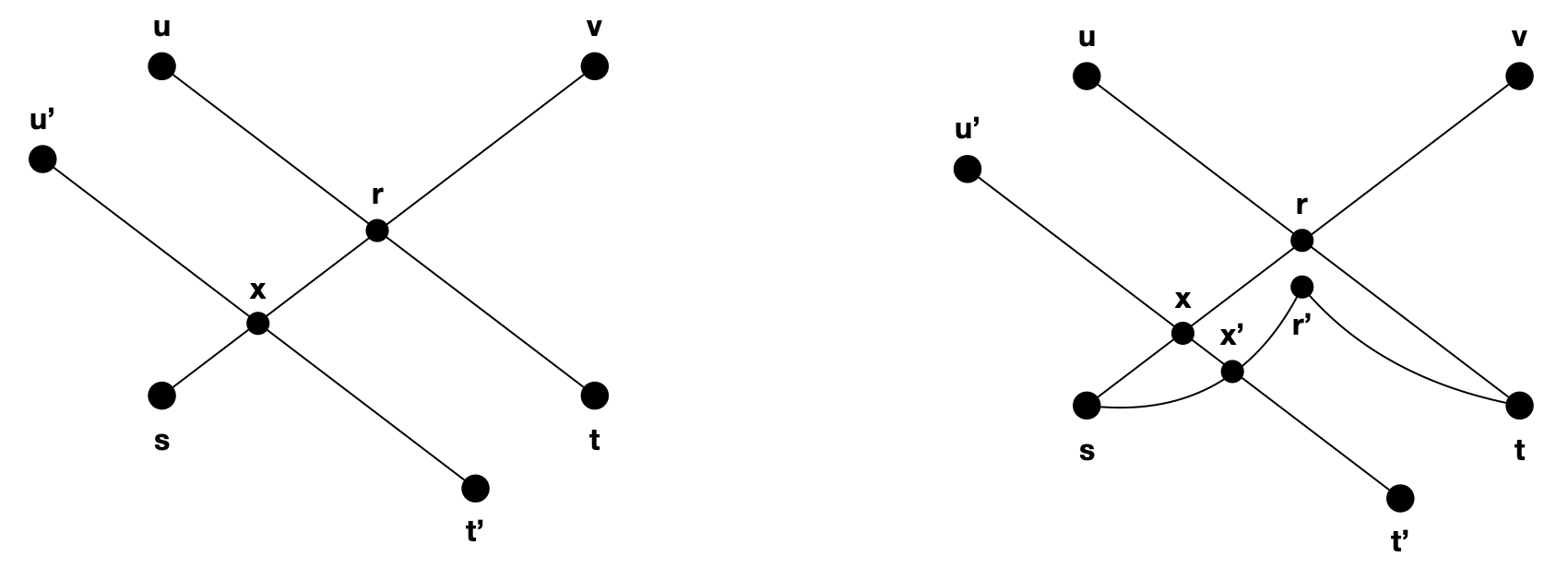}
\caption{An illustration of splitting a canonical path: Before (left) and after (right).\label{fig:flow-split1}}
\end{figure}

\begin{enumerate}[leftmargin=*]
    \item Add a new vertex $r'$ infinitesimally close to $r$ in the region enclosed by $\gamma_{s, t}$ and $\delta_{s, t}$ (the shortest $s$-$t$ path in $G$). For convenience, $r'$ will be called a \emph{copy} of $r$.

    \item Draw two new curves $(s,r')$ and $(t,r')$  infinitesimally close to the original drawing of $(s,r), (t,r)$ respectively. Denote $\gamma'_{s, t} = s\to r'\to t$.

    For any previous canonical path $\gamma_{s, t_1}$ with bend $r_1$ lying between $s, r$ on $sv$, we need to adjust the curve $(s,r')$ such that $(s,r')$ crosses the new curve $(r_1',t_1)$. See \Cref{fig:flow-split3}.
    
    For each path $\rho$ currently in $H$ which crossed $(s,r)$ at a vertex $x$ before this change to $H$, create a new intersection $x'$ between $(s,r')$ and $\rho$; this is done symmetrically for $(t,r')$. Note that $\rho$ might be created in previous iterations and not necessarily in $H^*$.
    
    \item To update the flow $\hat{F}$, originally we sent $\hat{F}_{s, t}$ units of flow from $s$ to $t$ along the path $s\to r \to t$. In the new graph, we reroute this $\hat{F}_{s, t}$ units of flow through the path $\gamma'_{s, t} = s\to r'\to t$.
\end{enumerate}
\begin{figure}[h]
\centering
\subfigure[New crossing between $(s,t')$ and $(r'_1,t'_1)$.]
{\scalebox{0.54}{\includegraphics{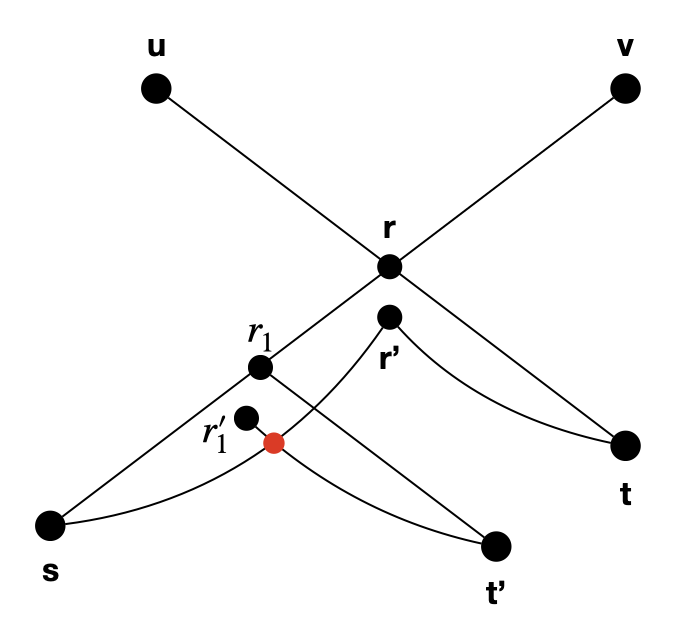}}\label{fig:flow-split3}}
\hspace{0.7cm}
\subfigure[Rerouting $F'$: augment $\alpha$ with path $(w,w')$.]
{\scalebox{0.46}{\includegraphics{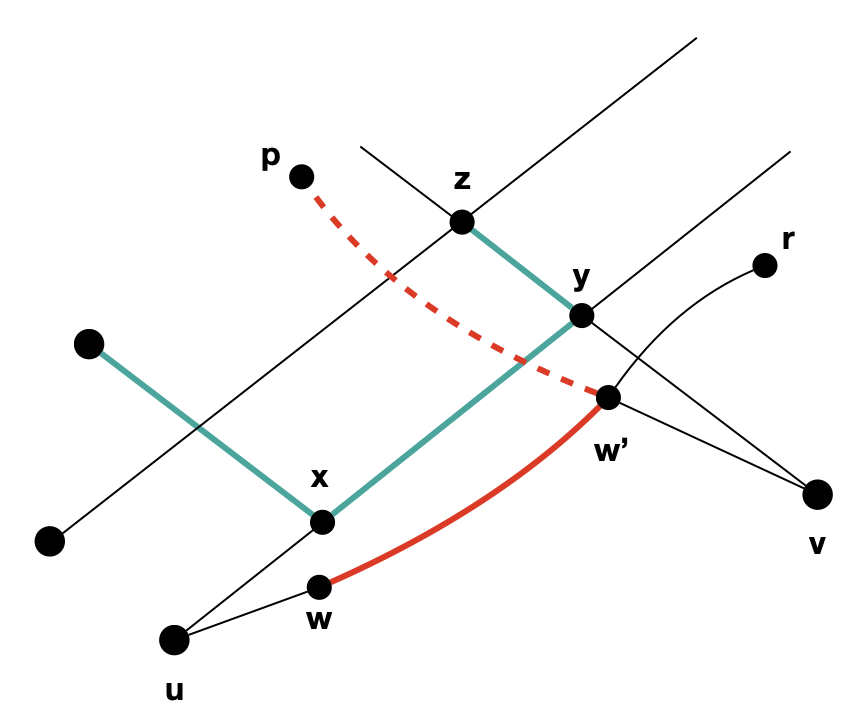}}\label{fig:flow-trace}}
\caption{An illustration of the construction of $H$ and the new $F'$-routing on it.}
\end{figure} 

When all iterations are finished, remove from $H$ all the original critical paths; that is $H\leftarrow H\setminus H^*$. Finally, we need to show that we can reroute $F'$, still satisfying Property \ref{prop: dominates}.

\begin{claim}
The demands $\set{F'_{s,t}}_{s,t\in T}$ can be routed in the current graph $H$ while dominated by $\hat{F}$.
\end{claim}
\begin{proof}
Recall that we have decomposed the flow $F'$ in $H^*$ into segments, now we concatenate them in the updated graph $H$. 
Consider any unit flow in $F'$ from $s$ to $t$.
We will iteratively construct a $s$-$t$ flow path $\alpha$ in $H$ using the flow sets $\Gamma_{s, *}$ constructed before. Throughout, we will maintain
    \begin{itemize}[leftmargin=*]
        \item $u$ as a varying terminal;    
        \item $x, y$ as vertices in $H^*$ which is on a critical path $\rho$ starting from $u$, such that $\beta_{x,y}$ is a segment of the unit $s$-$t$ flow; and        
        \item $r$ as a vertex on $\rho$ with $\beta_{x, y}\in \Gamma_{u, r}$; and $w$ as a vertex on $(s,r')$ in $H$, where $r'$ is a copy of $r$.
    \end{itemize} 

At the beginning, let $u, x\leftarrow s$. In an iteration, assume the flow $s$-$t$ flow travels from the segment $\beta_{x,y}$ to a segment $\beta_{y,z}$.
Assume we have constructed the prefix of $\alpha$ from $s$ to $u$.
    Assume $\beta_{y,z}\in \Gamma_{v, p}$. So $x, y$  lie between $u, r$ and $y, z$ lie between $v, p$. Then, by construction of $H$, $u$ must be connecting to a copy $r'$ of $r$, and $v$ must be connecting to a copy $p'$ of $p$. Therefore, the two paths cross, say at $w'$. We extend $\alpha$ at its current end by the path $(w,w')$.
    If $z = t$, then we can complete $\alpha$ by sending this flow to $t$ through $(w',t)$. Otherwise, we update the variable $u\leftarrow v, x\leftarrow y, y\leftarrow z, r\leftarrow p$, and move on to the next iteration. See \Cref{fig:flow-trace}.
\end{proof}

\subsection{Phase 2 in proving \Cref{clm: routable}: flow morphing}




Recall that in graph $H$ constructed above, all canonical paths are edge disjoint. 
In this subsection, we will gradually morph $H$ to $G$, by morphing each canonical path to their corresponding shortest path, thereby converting flows $\hat F, \hat F'$ in $H$ to flows in $G$, completing the proof of \Cref{clm: routable}.

Throughout the morphing process, $H$ will always be the union of, for each pair $s,t\in T$, a canonical path $\gamma_{s,t}$. We will always let $\hat F$ be the flow that,  for all pairs $s,t\in T$,  sends $F_{s,t}$ units through path $\gamma_{s,t}$. On the other hand, instead of maintaining the other flow $\hat F'$, we will simply ensure that, over the morphing process, the demands $\set{F'_{s,t}}_{s,t\in T}$ is always routable in the graph $H$ where for each pair $s,t\in T$, the edges of $\gamma_{s,t}$ have capacity $F_{s,t}$. Note that eventually this routable property in $G$ immediately implies \Cref{clm: routable}.

\paragraph{Morphing areas (M-Areas).}
For a canonical path $\gamma$ and two points $u,v\in\gamma$, denote by $\gamma_{u,v}$ the subpath of $\gamma$ between $u$ and $v$. For a shortest path $\delta$ in $G$ and two points $u,v\in\delta$, denote by $\delta_{u,v}$ the subpath of $\delta$ between $u,v$.
Let $u$ and $v$ be points lying on both $\gamma$ and $\delta$. We say that \emph{$\gamma,\delta$ form an \textsc{M-Area} between $u,v$}, iff 
\begin{itemize}
    \item $\delta_{u,v}$ and $\gamma_{u,v}$ don't cross internally; and
    \item the area enclosed by $\gamma_{u,v}$ and $\delta_{u,v}$ (i.e., the \textsc{M-Area}) does not contain any terminal faces.
\end{itemize} 
An \textsc{M-Area} is said to be \emph{minimal} if there is no other \textsc{M-Area} which is a strict subset of it. For an \textsc{M-Area}, we will identify it based on the two subpaths $\gamma_{u,v}$ and $\delta_{u,v}$ that enclose it.

\begin{observation}\label{obs:marea-top-bottom}
    Any canonical path crossing a minimal \textsc{M-Area} must cross both $\delta_{u,v}$ and $\gamma_{u,v}$ exactly once. 
\end{observation}
\begin{proof}
    A canonical path can't cross $\gamma_{u,v}$ twice since two canonical paths can't cross twice. This holds at the beginning and continues to hold due to Observation \ref{obs:flow-path-cross-iff-cross} later on. A canonical path $\gamma'$ also can't cross $\delta_{u,v}$ twice or else there would be a smaller \textsc{M-Area} formed by $\gamma'$ and $\delta$.
\end{proof}

We now describe the iterative morphing process, which continues iff there are still \textsc{M-Areas}.
In each iteration, consider a minimal \textsc{M-Area} enclosed by $\gamma_{u,v}$ and $\delta_{u,v}$. Let $a$ be the crossing of two canonical paths $\gamma_b$ and $\gamma_c$ which cross in the \textsc{M-Area}, and let $b$ and $c$ be the crossings between the pairs $\gamma_{u,v},\gamma_b$ and $\gamma_{u,v},\gamma_c$, respectively. We say that the crossing is \emph{closest} (to $\gamma_{u,v}$) if no other canonical paths cross either $(a,b)$ or $(a,c)$. We want to morph $\gamma_{u,v}$ to $\delta_{u,v}$ over the \textsc{M-Area}, and we do this iteratively by routing $\gamma_{u,v}$ around closest crossings.

We now describe how to find a closest crossing within the minimal \textsc{M-Area}. While there exist two canonical paths which cross in in the \textsc{M-Area}, we do the following. Start with two flow paths $\gamma_b$ and $\gamma_c$ which have a crossing $a$ in the \textsc{M-Area} and cross $\gamma_{u,v}$ at $b,c$ respectively. To find a closest crossing, repeat the following until convergence:
\begin{itemize}[leftmargin=*]
    \item If there exists a canonical path which crosses $(a,b)$, update $\gamma_c$ to be canonical path which crosses $(a,b)$ closest to $b$. Update $a$ and $c$ to be the new crossings of $\gamma_c$ with $\gamma_b$ and $\gamma_{u,v}$, respectively.
    \item If there exists a canonical path which crosses $(a,c)$, update $\gamma_b$ to be canonical path which crosses $(a,c)$ closest to $c$. Update $a$ and $b$ to be the new crossings of $\gamma_b$ with $\gamma_c$ and $\gamma_{u,v}$, respectively.
\end{itemize}
At each iteration after the first, we claim the area of the triangle $(a,b,c)$ decreases. Indeed, consider $\gamma_b$ and $\gamma_c$ in any later iteration. When we are about to update $\gamma_c$, we know there are no canonical paths crossing the current $(a,c)$ by construction from the previous iteration. Thus, the new path $\gamma_c$ must cross the current $(a,b)$ and $(b,c)$, forming a triangle which is a proper subset of the previous triangle $(a,b,c)$.
Similarly, when we are about to update $\gamma_b$, there are no canonical paths crossing the current $(a,b)$, so $\gamma_b$ crosses the current $(a,c)$ and $(c,b)$, and the area of the triangle decreases again. Since the area of the triangle strictly decreases each iteration and there is a finite number of triangles, the process converges and we find a closest crossing. See Figure \ref{fig:minimal-crossing} for an illustration.

\begin{figure}[h]
\centering
\subfigure
{\scalebox{0.32}{\includegraphics{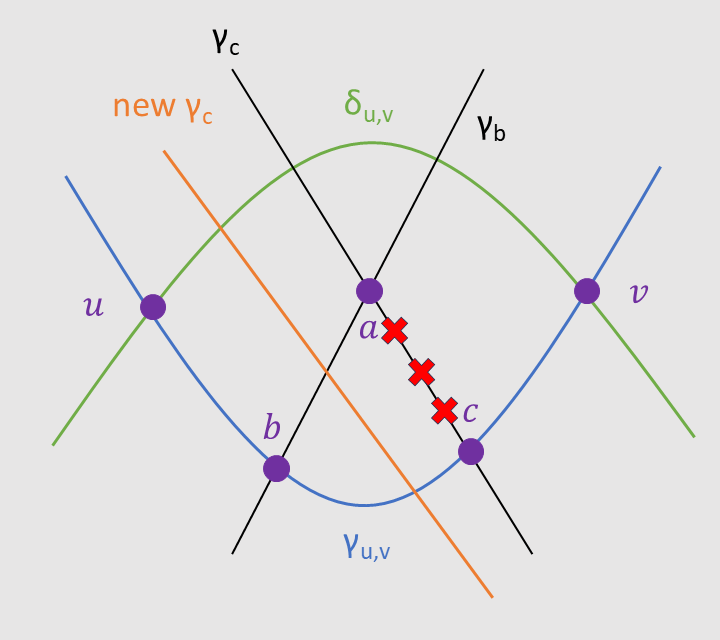}}}
\hspace{1.0cm}
\subfigure
{\scalebox{0.32}{\includegraphics{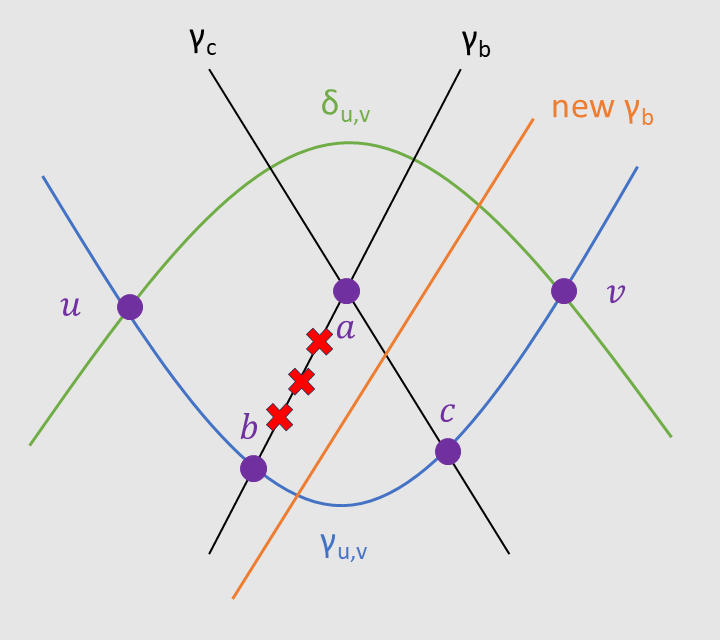}}}
\caption{In the left figure, $(a,b,c)$ is the current triangle with crossing at $a$. The red signs along $(a,c)$ indicate that the new $\gamma_c$ can never cross $(a,c)$. Hence, the new $\gamma_c$ (drawn in orange) must cross $(a,b)$ and $(b,c)$, decreasing the size of the triangle $(a,b,c)$. The right figure is symmetric.}\label{fig:minimal-crossing}
\end{figure} 

Upon convergence, we either have (i) a minimal \textsc{M-Area} such that no two canonical paths cross inside the \textsc{M-Area} or (ii) a closest crossing inside the minimal \textsc{M-Area}. Let $\gamma$ denote the canonical path which $\gamma_{u,v}$ is part of; we now describe the morphing of the canonical path $\gamma$ in the two cases. In case (i), the rerouting of $\gamma$ follows $\gamma$ until reaching $u$. Then it follows along $\delta_{u,v}$ until reaching $\gamma$. Finally, it follows along $\gamma$ until reaching the other terminal. In case (ii), the rerouting of $\gamma$ follows (infinitesimally close to) $\gamma$ until reaching $b$. The rerouting then follows (infinitesimally close to) the subpath $(b,a)$ and then $(a,c)$ until (almost) hitting $\gamma$. Finally, it continues to follow (infinitesimally close to) $\gamma$ until reaching the other terminal. Both of these reroutings are illustrated in {Figure \ref{fig:reroute}}. This completes the description of an iteration.

\begin{figure}[h]
\centering
\subfigure
{\scalebox{0.343}{\includegraphics{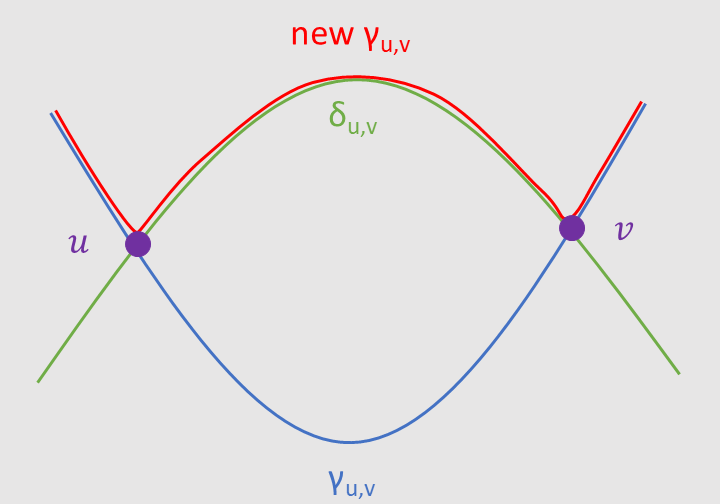}}}
\hspace{1.0cm}
\subfigure
{\scalebox{0.32}{\includegraphics{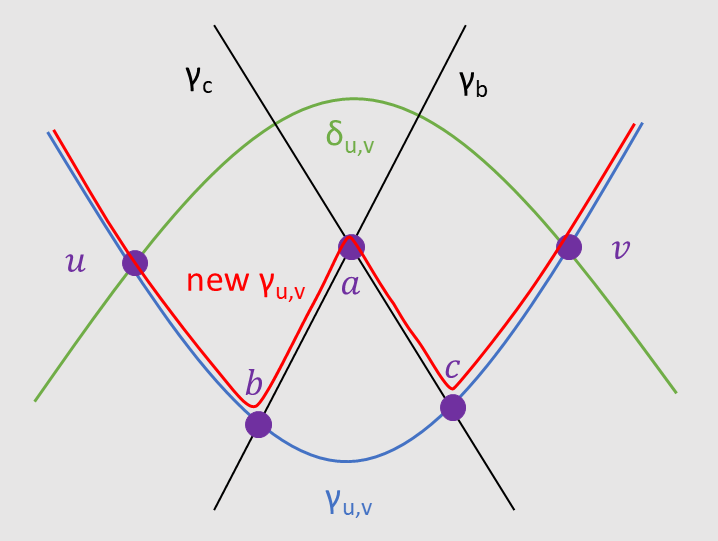}}}
\caption{In both cases, the red path illustrates the morphing of $\gamma_{u,v}$.}\label{fig:reroute}
\end{figure} 

In summary, our morphing process repeats the following while there is a non-empty \textsc{M-Area}:
\begin{itemize}
    \item[1.] Find a minimal \textsc{M-Area}, determined by $\gamma_{u,v}$ and $\delta_{u,v}$.
    \item[2.] Repeat the following while there are two canonical paths crossing in the \textsc{M-Area}.
    \begin{itemize}
        \item Find a closest crossing $a$ between $\gamma_b$ and $\gamma_c$ lying in the \textsc{M-Area}.
        \item Reroute $\gamma_{u,v}$ around the closest crossing $a$ via $(a,b)$ and $(a,c)$ (see Figure \ref{fig:reroute}, right).
        \item Observe that $\gamma_{u,v}$ and $\delta_{u,v}$ still define a minimal \textsc{M-Area}.
    \end{itemize}
    \item[3.] Reroute $\gamma_{u,v}$ to $\delta_{u,v}$ (see Figure \ref{fig:reroute}, left).
\end{itemize}

We first prove some properties of the morphing process. In particular, \Cref{obs:flow-path-cross-iff-cross} implies that Observation \ref{obs:marea-top-bottom} continues to hold throughout the morphing process.
\begin{observation}\label{obs:flow-path-cross-iff-cross}
    At any iteration of the morphing process, two canonical paths cross at most once.
\end{observation}
\begin{proof}
This holds at the beginning by Property \ref{prop: intersect iff}. After an iteration, since only one canonical path $\gamma$ was modified, it suffices to prove that any other flow path $\gamma'$ crosses $\gamma$ the same number of times before and after the iteration. In case (i), Observation \ref{obs:marea-top-bottom} implies that $\gamma'$ crosses $\gamma_{u,v}$ if and only it crosses $\delta_{u,v}$. In case (ii), it suffices to prove that $\gamma'$ crosses $(a,b)$ or $(a,c)$ if and only if it crosses $(b,c)$. This is since no canonical paths cross $(a,b)$ or $(a,c)$ as $a$ is a closest crossing, and also by induction, a canonical path can never cross $(b,c)$ twice before this iteration.
\end{proof}

\begin{observation}\label{obs:flow-shortest-no-face}
    At any iteration of the morphing process, the area between any canonical path and its corresponding shortest path in $G$ contain no terminal faces.
\end{observation}
\begin{proof}
This holds at the start of the morphing process. After each iteration, the change of the area between any canonical path and its corresponding shortest path is a subset of the \textsc{M-Area} in this iteration. Since an \textsc{M-Area} doesn't contain any terminal faces, the area still contains no terminal faces.
\end{proof}

We now show that the process converges.

\begin{claim}
    The iterative morphing process converges, and eventually every canonical path is morphed to the corresponding shortest path in $G$.
\end{claim}
\begin{proof}
    First, we prove that we successfully morph $\gamma_{u,v}$ to $\delta_{u,v}$ for each minimal \textsc{M-Area} we find. Consider the number of canonical path crossings in the \textsc{M-Area} as the potential function. After each iteration of morphing the canonical path $\gamma$ around a closest crossing, this number decreases by 1. Since the total number of such points is finite, the process converges. Next, we prove that the entire process converges. Consider $\Psi_1$ defined as the total number of crossings between canonical paths and shortest paths and $\Psi_2$ defined as the total number of canonical paths which haven't been morphed to their corresponding shortest paths. When we morph $\gamma_{u,v}$ to $\delta_{u,v}$, either $\Psi_1$ decreases by at least 1 when at least one of $u,v$ is a non-terminal or $\Psi_2$ decreases by 1 when both $u,v$ are terminals and we finish morphing the canonical path between $u$ and $v$. Since $\Psi_1$ and $\Psi_2$ are both finite at the start, the process converges.

    Upon convergence, we know that there are no more \textsc{M-Area}s. We want to show that this implies that all canonical paths are morphed to their corresponding shortest paths. Suppose for contradiction this isn't the case. Let $\gamma$ be a canonical path which hasn't been morphed to its corresponding shortest path $\delta$. Let $x$ and $y$ denote two crossings between $\delta$ and $\gamma$ which are consecutive on $\delta$. If they are also consecutive on $\gamma$, then we are done. Otherwise, let $x'$ and $y'$ be two crossings between $\gamma_{x,y}$ and $\delta$ which are consecutive on $\delta$ and $x'\neq x$ or $y'\neq y$. Observe that the number of crossings between $\gamma_{x',y'}$ and $\delta$ is strictly less than the number of intersections between $\gamma_{x,y}$ and $\delta$ because $\gamma_{x',y'}$ is a strict subpath of $\gamma_{x,y}$ and in particular, doesn't include at least one of the intersections $x$ or $y$. Hence, recursively applying this process gives a pair of points $x,y$ such that $\delta_{x,y}$ and $\gamma_{x,y}$ don't intersect. Since the area surrounded by two canonical subpaths is a subset of the area between canonical paths, which from Observation \ref{obs:flow-shortest-no-face} contains no terminal faces, this is an \textsc{M-Area}. 
\end{proof}


The next claim shows the routability of demands $\set{F'_{s,t}}_{s,t\in T}$, and completes the proof of \Cref{clm: routable}.

\begin{claim}
After morphing a canonical path over a crossing, the demands $\set{F'_{s,t}}_{s,t\in T}$ is still routable in the updated graph $H$ where for each pair $s,t$, edges in the canonical $\gamma_{s,t}$ have capacity $F_{s,t}$.
\end{claim}
\begin{proof}
The routability holds at the beginning. It suffices to show that it holds after each iteration.
In case (i), where there are no crossings in the \textsc{M-Area}, the graph defined by the flow $\hat{F}$ is exactly the same as before, so the claim holds by induction. In case (ii), we show that the updated graph and the old graph  are flow-equivalent, via Wye-Delta transformations.
    \begin{observation}[Wye-Delta Transformation]
        The following operations on a capacitated instance $(H,T)$ with edge capacities $c(\cdot)$ preserve flow equivalence (i.e., a demand on terminals $T$ is routable in $H$ iff it is routable in the new graph $H$ after the operation)
        \begin{itemize}
            \item Wye-Delta transformation: let $x$ be a degree-three non-terminal with neighbors $u,v,w$. Remove $x$ along with its incident edges, and add edges $(u,v)$, $(u,w)$, $(v,w)$ with respective capacities $[c(x,u)+c(x,v)-c(x,w)]/2$, $[c(x,u)+c(x,w)-c(x,v)]/2$, and $[c(x,v)+c(x,w)-c(x,u)]/2$.
            \item Delta-Wye transformation: let $u,v,w$ denote endpoints of a triangle. Remove the edges of the triangle and add a new non-terminal $x$ with new edges $(x,u)$, $(x,v)$, and $(x,w)$ and respective capacities $c(u,v)+c(u,w)$, $c(v,u)+c(v,w)$, and $c(w,u)+c(w,v)$.
        \end{itemize}
    \end{observation}
\begin{figure}[h]
    \centering
    \includegraphics[width=0.8\linewidth]{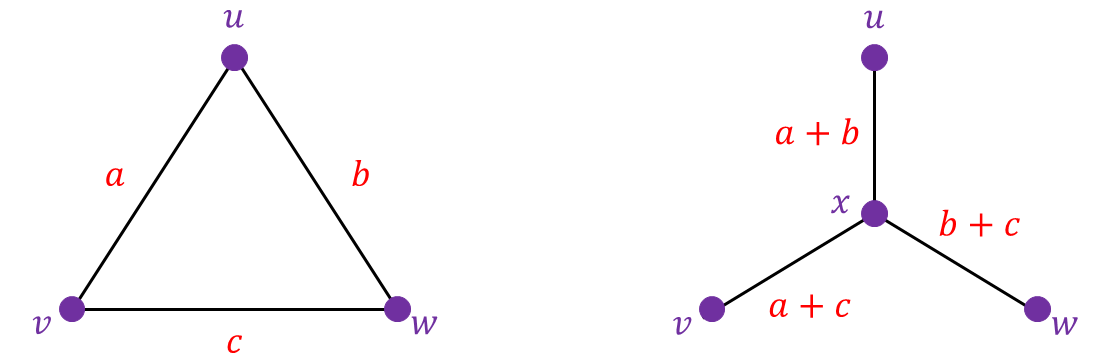}
    \caption{An illustration of the Delta-Wye transformation (applying on the left graph to obtain the right graph), where $a,b,c$ represent the edge capacities of $(u,v)$, $(u,w)$, and $(v,w)$, respectively. Applying Wye-Delta transformation on the right graph gives us the left graph.}
    \label{fig:enter-label}
\end{figure}
    \begin{proof}
        Let $G$ denote the original graph and $H$ denote the graph after the Wye-Delta/Delta-Wye transformation. Let $\Delta$ and $Y$ denote the induced subgraphs of $G$ and $H$ on $\{u,v,w\}$ and $\{u,v,w,x\}$, respectively. First, we reduce showing terminal-flow equivalence on the entire graph to showing terminal-flow equivalence of the $\Delta$ and $Y$ structures. Fix a demand $\mathbf{b}\in\mathbb{R}^T$ on the terminals. We want to show that $\mathbf{b}$ is routable in $G$ if and only if it is routable in $H$, assuming that $\Delta$ and $Y$ are terminal flow equivalent. Indeed, take a flow $f$ in $G$ routing the demand and let $f''$ denote the flow $f$ restricted to the edges $(u,v)$, $(v,w)$, and $(u,w)$. The flow $f-f'$ no longer completely routes $\mathbf{b}$; it routes $\mathbf{b}-\mathbf{b}'$ for some residual demand $\mathbf{b}'\in\mathbb{R}^T$, which is supported on $\{u,v,w\}$. Note that $\mathbf{b}'$ can be routed in $\Delta$, specifically via $f'$. Since $\Delta$ and $Y$ are flow equivalent, there must be a flow $f''$ in $Y$ which routes $\mathbf{b}'$. Thus, the flow $f-f'+f''$ routes $\mathbf{b}-\mathbf{b}'+\mathbf{b'}=\mathbf{b}$ in $H$. The flow $f-f'$ is supported on $H-Y$ and $f''$ is supported on $Y$, so $f-f'+f''$ has congestion $1$, as desired. The reduction in the other direction is similar.
        

        It remains to show that $Y$ and $\Delta$ are terminal-flow equivalent. Since both the $Y$ and $\Delta$ structures are Okamura-Seymour instances (planar graphs with all terminals lying on the outer boundary), the flow-cut gap is $1$~\cite{okamura1981multicommodity}. Therefore, it suffices to prove terminal-cut equivalence of $\Delta$ and $Y$ to conclude the flow equivalence claim. Consider any partition of the terminals $\{u,v,w\}$ in $\Delta$ and $Y$. We want to show that the value of the minimum cut separating the terminals in $\Delta$ and $Y$ are the same. By symmetry, it suffices to show that the minimum cut separating $u$ from $v$ and $w$ is the same in $Y$ and $\Delta$. And this is indeed the case since the minimum cut in $Y$ is $c(x,u)$ and the minimum cut in $\Delta$ is $c(u,v)+c(u,w)=[c(x,u)+c(x,v)-c(x,w)]/2+[c(x,u)+c(x,w)-c(x,v)]/2=c(x,u)$, as desired. For the Delta-Wye transformation, we can similarly argue that it suffices to show that the minimum cut separating $u$ from $v$ and $w$ are the same in $\Delta$ and $Y$. This is also the case since the minimum cut in $\Delta$ is $c(u,v)+c(u,w)$ and the minimum cut in $Y$ is $c(x,u)=c(u,v)+c(u,w)$. 
    \end{proof}

\begin{figure}[h]
\centering
\subfigure
{\scalebox{0.08}{\includegraphics{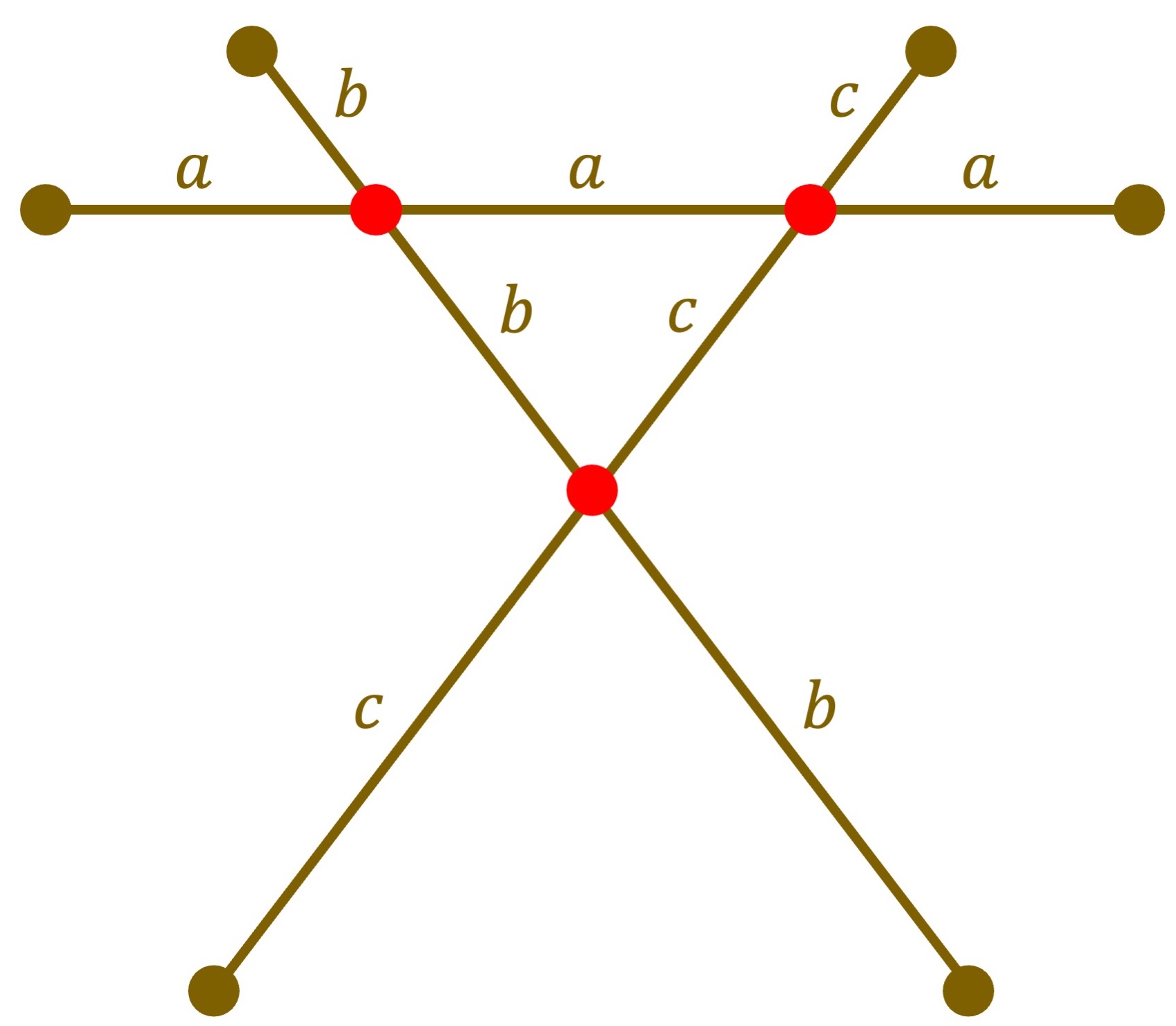}}}
\hspace{0.3cm}
\subfigure
{\scalebox{0.08}{\includegraphics{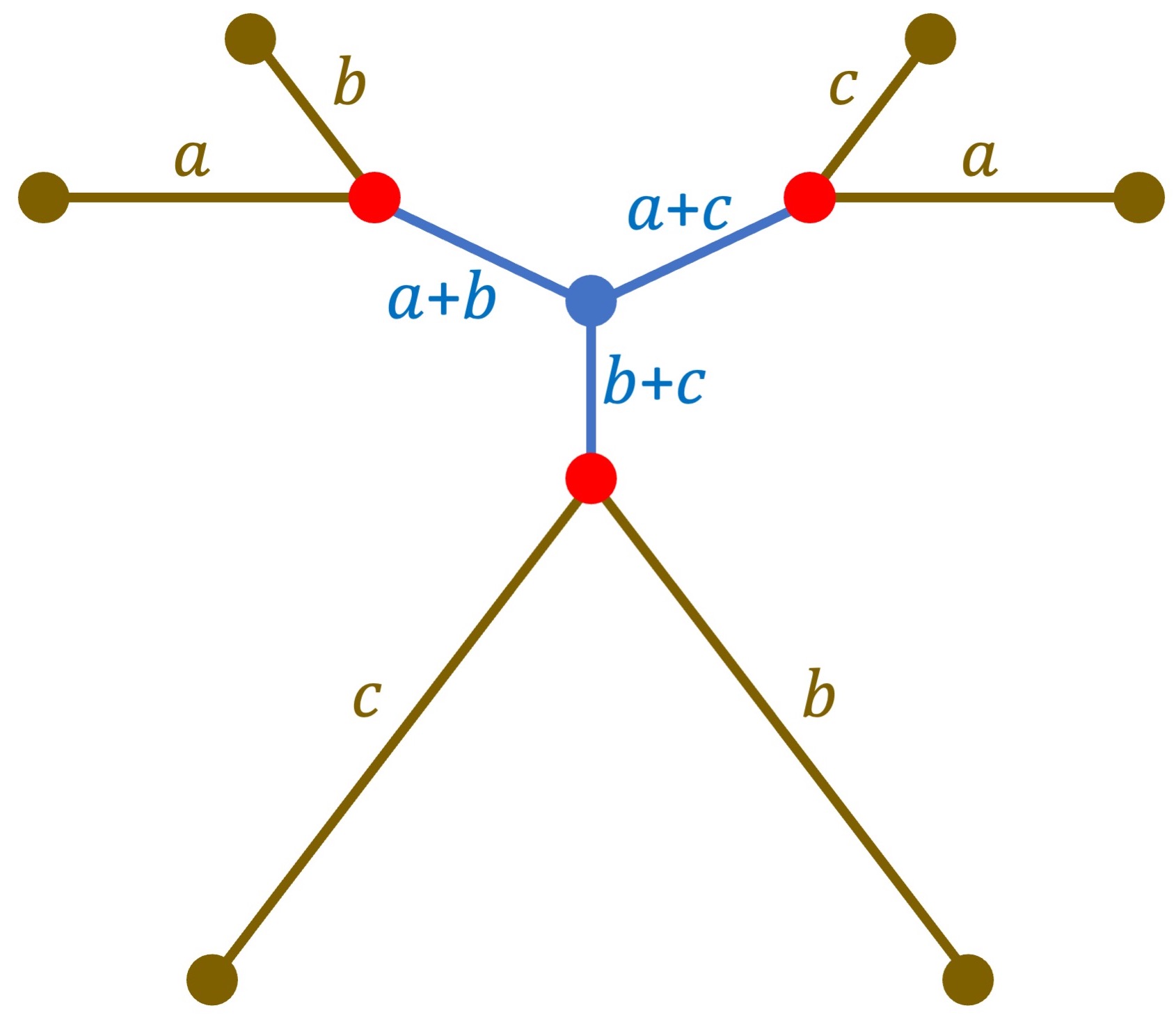}}}
\hspace{0.3cm}
\subfigure
{\scalebox{0.08}{\includegraphics{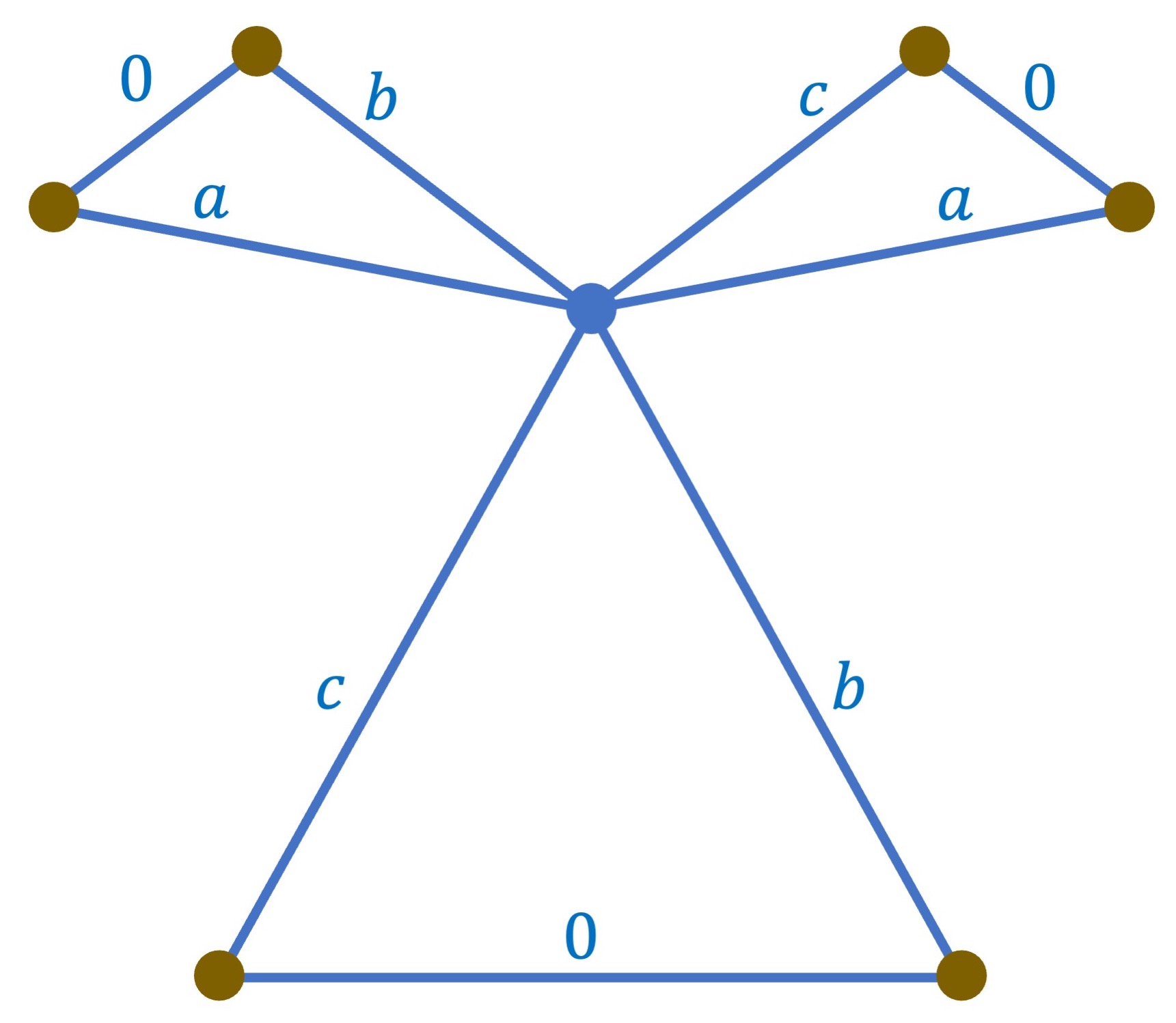}}}
\hspace{0.3cm}
\subfigure
{\scalebox{0.08}{\includegraphics{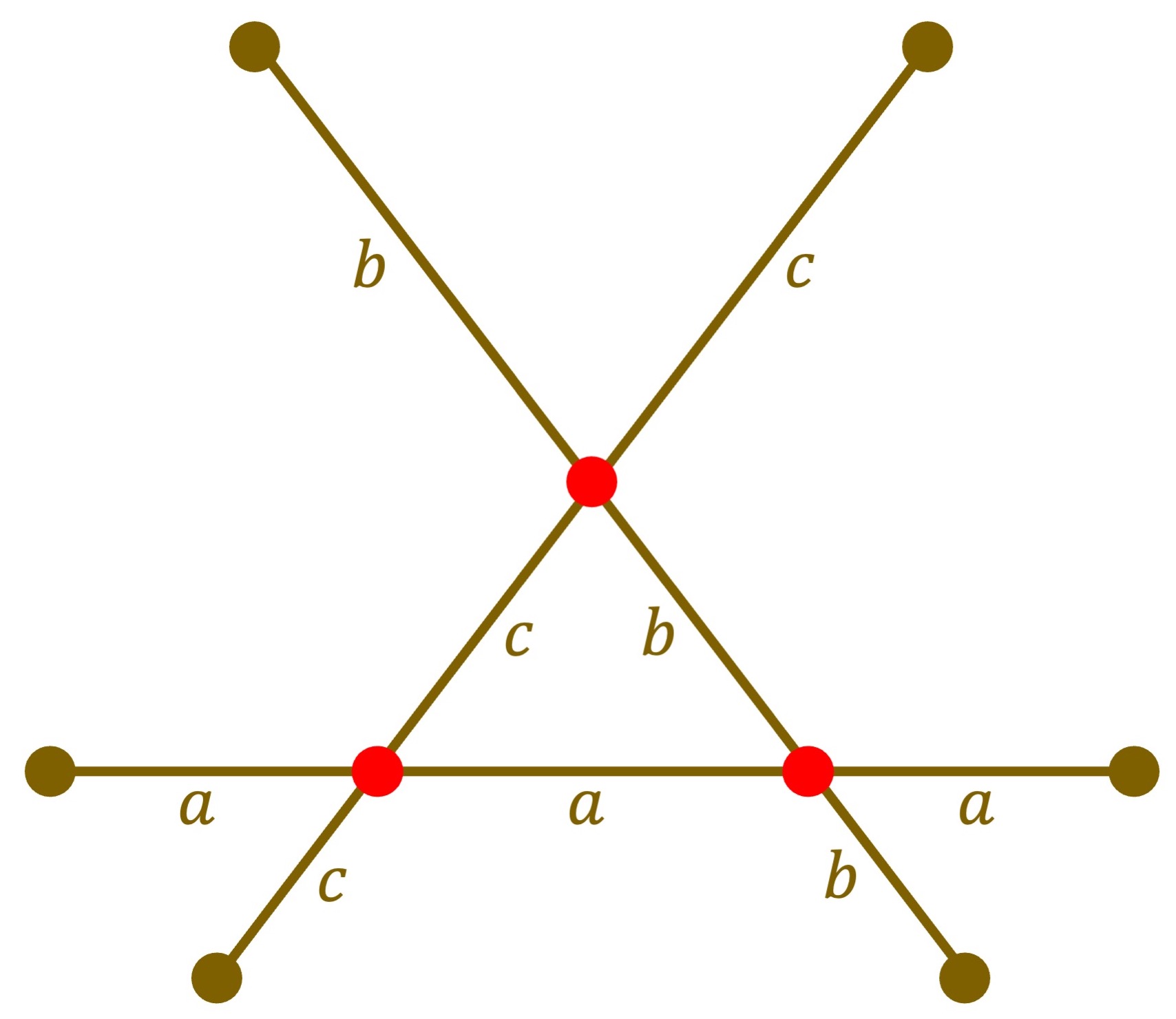}}}
\hspace{0.3cm}
\subfigure
{\scalebox{0.08}{\includegraphics{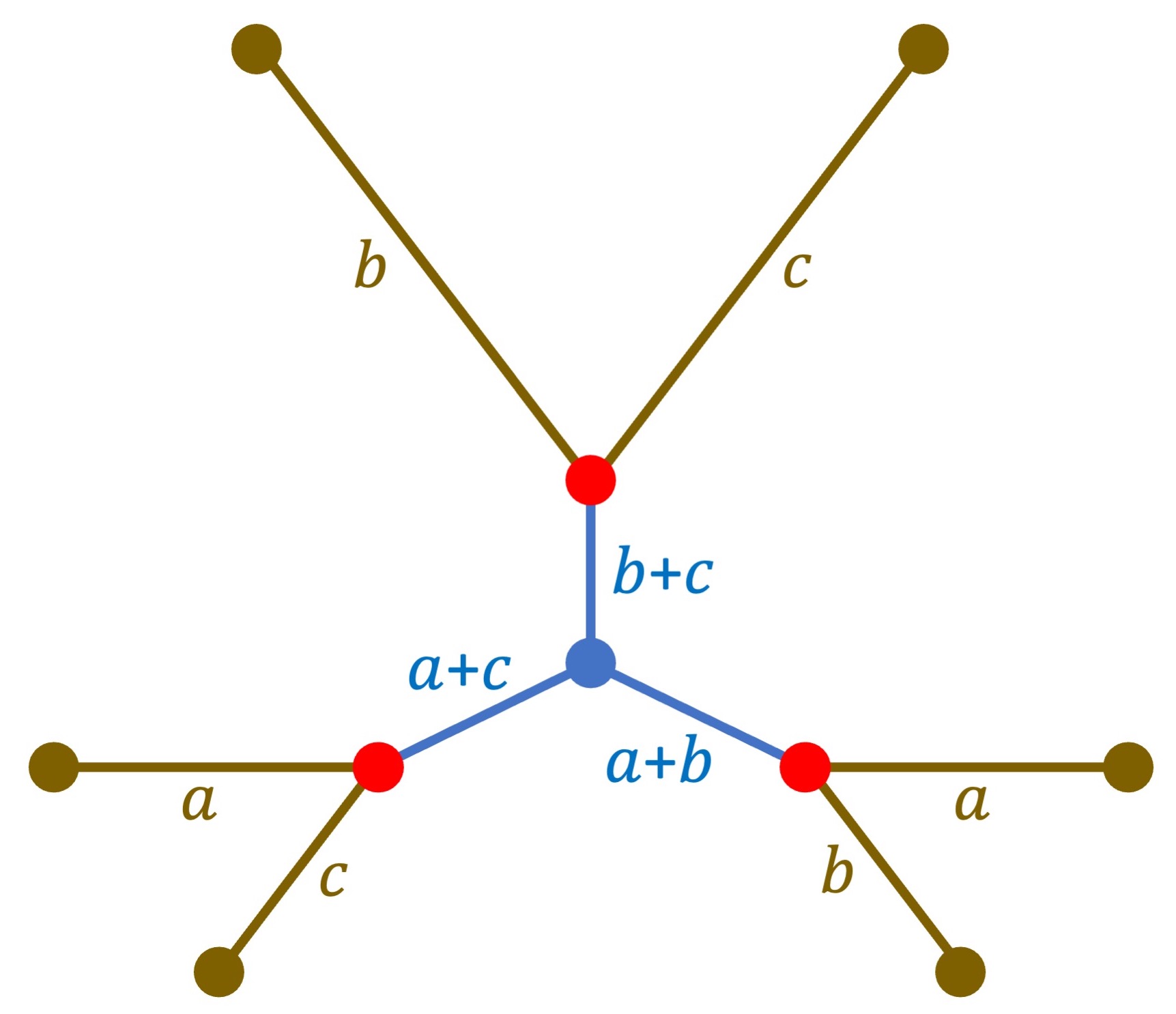}}}
\hspace{0.3cm}
\subfigure
{\scalebox{0.08}{\includegraphics{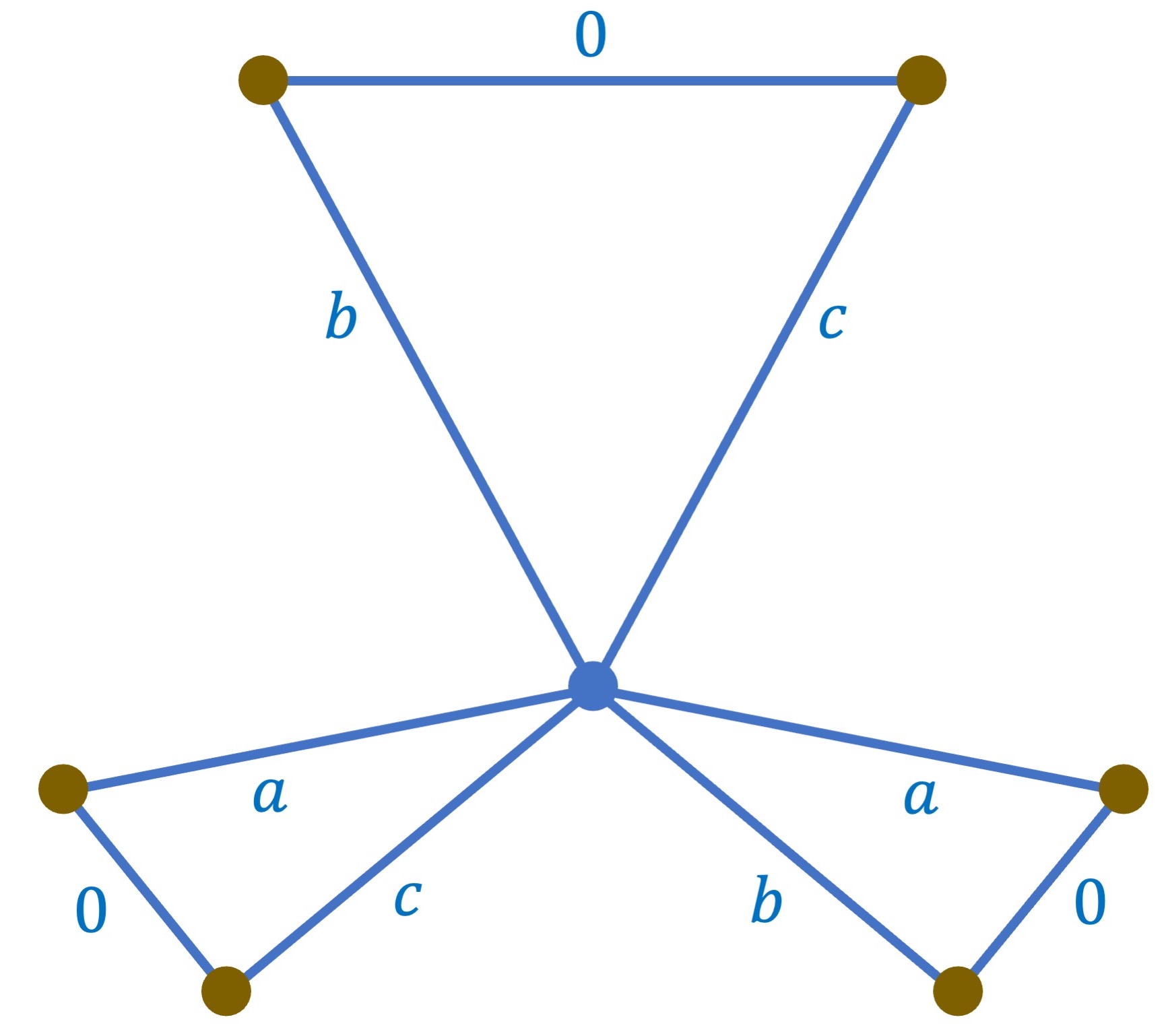}}}
\caption{In the above sequence, we show that locally the old graph $H$ flow-equivalent to the blue structure. In the below sequence, we show that the new graph $H$ after the rerouting is flow-equivalent to the blue structure. But the two blue structures are exactly the same (modulo capacity-$0$ edges), so we can conclude the old and new graphs $H$ are flow-equivalent. \label{fig: YD_2}}
\end{figure}

    The only difference between the old and new graphs $H$ is whether $\gamma$ crosses $\gamma_b$ and $\gamma_c$ before or after $\gamma_b$ and $\gamma_c$ cross each other. The topology of the remainder of the graph is exactly the same, so it suffices to prove flow equivalence locally.
    Let $a,b,c$ denote the capacities of edges in paths $\gamma,\gamma_b,\gamma_c$, respectively. 
    In Figure \ref{fig: YD_2}, we show that the old and new graphs $H$ are locally flow-equivalent via Wye-Delta transformations. 
\end{proof}

%% file: 04_appendix.tex
\section{Missing Proofs in \Cref{sec: warmup}}

\subsection{Proof of \Cref{obs: split}}
\label{apd: Proof of obs: split}

Denote by $a_j$ the segment of the inner face boundary going clockwise from $t'_j$ to $t'_{j+1}$.
Recall that we say that $t_i$ splits at $(t'_j,t'_{j+1})$ if the region enclosed by path $P_{i,j}$, path $P_{i,j+1}$, and segment $a_j$ contains the inner face.

We first prove existence of the splitting pair. If $t_i$ splits at $(t'_j,t'_{j+1})$ for some $j=1,\ldots,\frac{k}{2}-1$, then we are done. Suppose not; then we will show it splits at the pair $(t'_{k/2},t_1)$. To see this, consider the regions $A_j$ inscribed by the compositions of paths $P_{i,j}$, $t'_{j}\to t'_{j+1}$, and $P_{i,j+1}$ for each $j=1,\ldots,\frac{k}{2}-1$. Denote their union by $A=\bigcup_{j=1,\ldots,k/2-1}A_j$, and note that $A$ doesn't contain the inner face. It can be seen by definition that $A$ is the region inscribed by the composition of paths $P_{i,1}$, $t'_1\to t'_2\to\ldots\to t'_{k/2-1}$, and $P_{i,k/2-1}$. As a result, the region inscribed by composition of $P_{i,1}$, $t'_1\to t'_{k/2-1}$, and $P_{i,k/2-1}$ must contain the inner face since the path consisting of the sequence of edges $t'_{j}\to t'_{j+1}$ for $j=1,\ldots,k/2$ inscribes the inner face. Thus, $t_i$ must split at $t_{k/2}$ and we are done.

We now prove uniqueness of the splitting point $(t'_j,t'_{j+1})$. Suppose for contradiction that $t_i$ splits at $(t'_j,t'_{j+1})$ and $(t'_\ell,t'_{\ell+1})$ for $\ell\neq j$. Observe that neither pair of paths $P_{i,j}$ and $P_{i,\ell+1}$ nor $P_{i,j+1}$ and $P_{i,\ell}$ can have a common node other than $t_i$; this would give a contradiction to the assumption that two shortest paths intersect only at a single point. But if neither pair had a common node, then both paths $P_{i,\ell}$ and $P_{i,\ell+1}$ would have to lie in the interior on the region enclosed by $(P_{i,j},P_{i,j+1},a_j)$. This implies that the region enclosed by $(P_{i,\ell},P_{i,\ell+1},a_\ell)$ cannot contain the inner face, a contradiction.

\subsection{Proof of \Cref{obs: split move}}
\label{apd: Proof of obs: split move}

Suppose that $t'_{j_2}$ didn't lie on the (closed) clockwise segment between $t'_{j_1}$ and $t'_{j_3}$; we will show that there must be two shortest paths which intersect twice. Observe that since $j_2$ isn't on the clockwise segment between $t'_{j_1}$ and $t'_{j_3}$, it must be on the counterclockwise segment between $t'_{j_1-1}$ and $t'_{j_3+1}$. Now, consider the critical path $P_{i_2,j_2+1}$. Since $t_{i_2}$ splits at $(t'_{j_2},t'_{j_2+1})$, this forces $P_{i_2,j_2+1}$ to intersect $P_{i_1,j_1}$ twice, giving a contradiction.

\subsection{Formal description of the emulator structure}
\label{apd: 2-face emulator}

Formally, the emulator consists of three parts: 
\begin{itemize}[leftmargin=*]
    \item  The first part is a planar graph $H_1$ that contains terminals $t_1,\ldots,t_{k/2}$ and has them lie on a single face $F_1$ (in the same order as they lie on the face in $G$), preserving the distances between them in $G$. In fact, we will simply use the construction of \cite{ChangO20} and \cite{goranci2020improved}, which gives a construction of $H_1$ on $O(k^2)$ vertices.
    \item The second part is a planar graph $H_2$ that contains terminals $t'_1,\ldots,t'_{k/2}$ and has them lie on a single face $F_2$ (in the same order as they lie on the face in $G$), preserving the distances between them in $G$. Similarly, we will simply use the construction of \cite{ChangO20} and \cite{goranci2020improved}, which gives a construction of $H_2$ on $O(k^2)$ vertices.
    \item the third part is a planar graph $H_3$ that contains all terminals, with terminals $t_1,\ldots,t_{k/2}$ lying on a single face $F_3$ (in the same order), and terminals $t'_1,\ldots,t'_{k/2}$ lying on another face $F_4$ (in the same order). $H_3$ is in charge of preserving the distances between inter-face pairs $t_i,t'_j$ in $G$.
\end{itemize}

\begin{figure}[h]
    \centering
    \subfigure
    {\scalebox{0.1}{\includegraphics{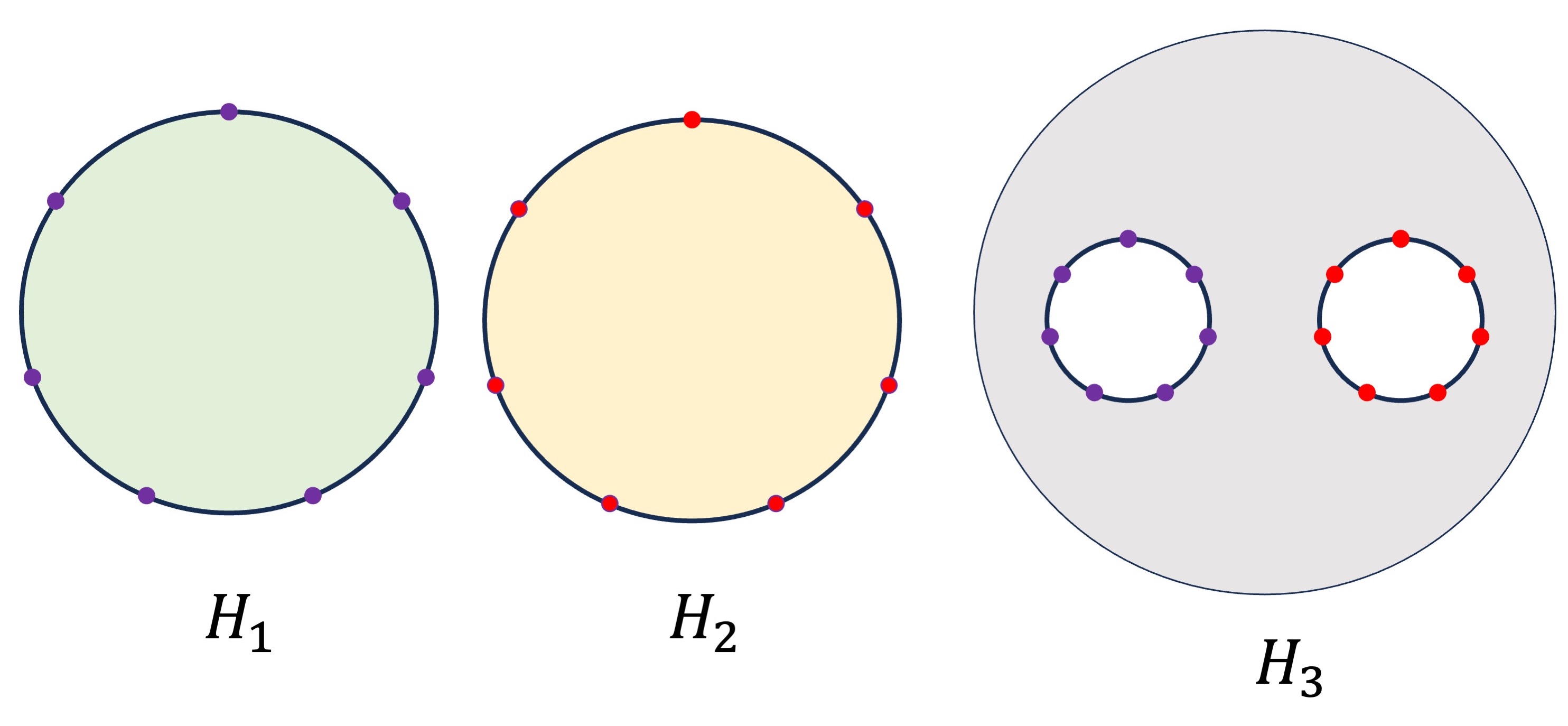}}}
    \hspace{1cm}
    \subfigure
    {\scalebox{0.1}{\includegraphics{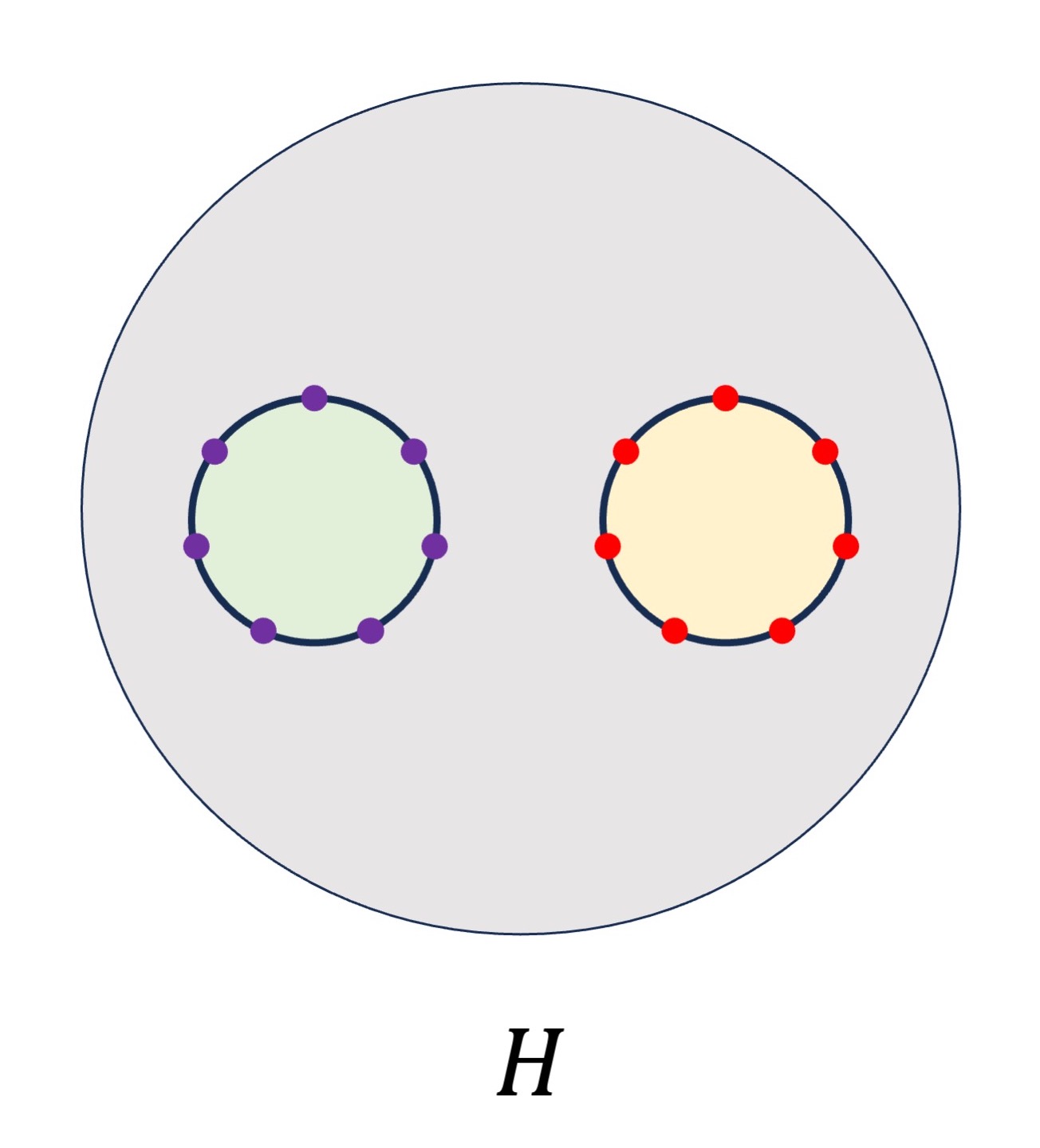}}}
    \caption{An illustration of graphs $H_1,H_2,H_3$ and graph $H$ (right) obtain by concatenating them.\label{fig: concatenation}}
\end{figure}

The final emulator $H$ is obtained by sticking graphs $H_1,H_2,H_3$ together. Specifically, we start from graph $H_3$, and (i) place the drawing of $H_1$ (with $F_1$ being the its outer face) into face $F_3$, and identify the two copies of each terminal $t_i$ in $H_1$ and $H_3$; (ii) place the drawing of $H_2$ (with $F_2$ being the its outer face) into face $F_4$, and identify the two copies of each terminal $t'_j$ in $H_2$ and $H_3$. See \Cref{fig: concatenation} for an illustration.

From the construction, we have that $|V(H)|\le |V(H_1)|+|V(H_2)|+|V(H_3)|\le O(k^2)+|V(H_3)|$. Therefore, if we show $|V(H_3)|=O(k^2)$, then we immediately obtain $|V(H)|=O(k^2)$.

\section{Proof of \Cref{clm: feasibility or flow}}
\label{apd: Proof of clm: feasibility or flow}

We use the following Farkas' lemma \cite{farkas1898fourier}.
\begin{lemma}[Farkas' Lemma]
For any matrix $A\in \mathbb{R}^{m\times n}$ and any vector $b\in \mathbb{R}^m$, exactly one of the two assertions holds:
\begin{itemize}
\item there exists a vector $x\in \mathbb{R}^n$ such that $Ax\le b$ and $x\ge 0$;
\item there exists a vector $y\in \mathbb{R}^m$ such that $A^Ty\ge 0$, $b^Ty<0$, and $y\ge 0$.
\end{itemize}\label{lem:farkas}
\end{lemma}
    
Note that (LP-$H^*$) can be written in the $Ax\le b$ form, and its feasibility is equivalent to the existence of a vector $x\ge 0$ that satisfies $Ax\le b$. Therefore, from Farkas' lemma, either (LP-$H^*$) is feasible, or there exists $y\ge 0$ satisfying $A^Ty\ge 0$ and $b^Ty<0$. 

To define the terminal flows $F$ and $F^\prime$, let us examine the vector $y$. If we view $A^Ty$ as a linear combination of the constraints represented by $A$, then the values in $y$ can be interpreted as the coefficients for the constraints. For each coefficient $y_Q$ of the first type of constraint on simple paths $Q$, define $F^\prime(Q)=y_Q$. For each coefficient $y_P$ of the second type of constraint on canonical paths $P$, define $F(P)=y_P$. Then define $F$ and $F^\prime$ on all other inputs to be zero. We will show that the flows $F,F^\prime$ satisfy the desired conditions.

    Condition (1) is obviously satisfied by how we define $F$. To see that Condition (2) holds, we will try to prove that $F_e-F^\prime_e\ge 0$ for each $e\in E(H)$. We claim that this is equivalent to the fact that the $A^Ty\ge 0$. Let us focus on one edge $e\in E(G)$. The coefficients of $x_e$ form a column of $A$; let us denote this column by $A_e$. Then $A^Ty\ge 0$ implies that $A^T_ey\ge 0$; we claim that this implies $F_e-F^\prime_e\ge 0$. Observe that the coordinate of $A_e$ corresponding to the first type of constraint for simple path $Q$ has value $-1$ if $e\in E(Q)$, and zero otherwise. Similarly, the coordinate of $A_e$ corresponding to the second type of constraint for canonical path $P$ has value $1$ if $e\in E(P)$, and zero otherwise. As a result,
    $$A_e^Ty=\sum_{P:e\in E(P)}y_P-\sum_{Q:e\in E(Q)}y_Q=\sum_{P:e\in E(P)}F(P)-\sum_{Q:e\in E(Q)}F(Q)=F_e-F^\prime_e,$$
    where the summations over $P$ are over canonical paths and summations over $Q$ are over all simple paths.

    Finally, we prove Condition (3) by showing $\cost(F)-\cost(F^\prime)<0$. We claim that the condition is equivalent to the condition that $b^Ty>0$ by showing $\cost(F)-\cost(F^\prime)=b^Ty$. Indeed, the coordinate of $b$ corresponding to the first type of constraint for simple path $Q$ connecting terminals $t,t^\prime$ has value $-\dist_G(t,t^\prime)$. Similarly, the coordinate of $b$ corresponding to the second type of constraint for canonical path $P$ connecting terminals $t,t^\prime$ has value $\dist_G(t,t^\prime)$. As a result, $$b^Ty=\sum_{P}y_P\cdot \dist_G(t,t^\prime)-\sum_{Q}y_Q\cdot\dist_G(t,t^\prime)=\cost(F)-\cost(F^\prime),$$
    where the first summation is over all canonical paths $P$, the second summation is over all simple paths $Q$, and $t,t^\prime$ are the terminals which paths $P$ and $Q$ connect.

%% file: main.bbl
\newcommand{\etalchar}[1]{$^{#1}$}
\begin{thebibliography}{CGMW18}

\bibitem[BG08]{basu2008steiner}
Amitabh Basu and Anupam Gupta.
\newblock {Steiner} point removal in graph metrics.
\newblock {\em Unpublished Manuscript, available from http://www. math. ucdavis. edu/\~{} abasu/papers/SPR. pdf}, 1:25, 2008.

\bibitem[CCL{\etalchar{+}}23]{chang2023covering}
Hsien-Chih Chang, Jonathan Conroy, Hung Le, Lazar Milenkovic, Shay Solomon, and Cuong Than.
\newblock Covering planar metrics (and beyond): {$O(1)$} trees suffice.
\newblock In {\em 2023 IEEE 64th Annual Symposium on Foundations of Computer Science (FOCS)}, pages 2231--2261. IEEE, 2023.

\bibitem[CCL{\etalchar{+}}24]{chang2023shortcut}
Hsien{-}Chih Chang, Jonathan Conroy, Hung Le, Lazar Milenkovic, Shay Solomon, and Cuong Than.
\newblock Shortcut partitions in minor-free graphs: {Steiner} point removal, distance oracles, tree covers, and more.
\newblock In David~P. Woodruff, editor, {\em Proceedings of the 2024 {ACM-SIAM} Symposium on Discrete Algorithms, {SODA} 2024}, pages 5300--5331. {SIAM}, 2024.

\bibitem[CGMW18]{chang2018near}
Hsien{-}Chih Chang, Pawel Gawrychowski, Shay Mozes, and Oren Weimann.
\newblock Near-optimal distance emulator for planar graphs.
\newblock In Yossi Azar, Hannah Bast, and Grzegorz Herman, editors, {\em 26th Annual European Symposium on Algorithms, {ESA} 2018}, volume 112 of {\em LIPIcs}, pages 16:1--16:17. Schloss Dagstuhl - Leibniz-Zentrum f{\"{u}}r Informatik, 2018.

\bibitem[Che18]{cheung2018steiner}
Yun~Kuen Cheung.
\newblock {Steiner} point removal: distant terminals don't (really) bother.
\newblock In {\em Proceedings of the Twenty-Ninth Annual ACM-SIAM Symposium on Discrete Algorithms}, pages 1353--1360. Society for Industrial and Applied Mathematics, 2018.

\bibitem[CKT22]{chang2022almost}
Hsien-Chih Chang, Robert Krauthgamer, and Zihan Tan.
\newblock Almost-linear $\varepsilon$-emulators for planar graphs.
\newblock In {\em Proceedings of the 54th Annual ACM SIGACT Symposium on Theory of Computing}, pages 1311--1324, 2022.

\bibitem[CO20]{ChangO20}
Hsien{-}Chih Chang and Tim Ophelders.
\newblock Planar emulators for {Monge} matrices.
\newblock In J.~Mark Keil and Debajyoti Mondal, editors, {\em Proceedings of the 32nd Canadian Conference on Computational Geometry, {CCCG} 2020}, pages 141--147, 2020.

\bibitem[CT24]{chen2024lower}
Yu~Chen and Zihan Tan.
\newblock An {$O(\sqrt{\log|T|)}$} lower bound for {Steiner} point removal.
\newblock In {\em Proceedings of the 2024 Annual ACM-SIAM Symposium on Discrete Algorithms (SODA)}, pages 694--698. SIAM, 2024.

\bibitem[CT25]{chen2025path}
Yu~Chen and Zihan Tan.
\newblock Path and intersections: Characterization of quasi-metrics in directed {Okamura-Seymour} instances.
\newblock In {\em Proceedings of the 2025 Annual ACM-SIAM Symposium on Discrete Algorithms (SODA)}, pages 2467--2490. SIAM, 2025.

\bibitem[CXKR06]{chan2006tight}
T-H~Hubert Chan, Donglin Xia, Goran Konjevod, and Andrea Richa.
\newblock A tight lower bound for the {Steiner} point removal problem on trees.
\newblock In {\em Approximation, Randomization, and Combinatorial Optimization. Algorithms and Techniques}, pages 70--81. Springer, 2006.

\bibitem[EFL18]{erickson2018holiest}
Jeff Erickson, Kyle Fox, and Luvsandondov Lkhamsuren.
\newblock Holiest minimum-cost paths and flows in surface graphs.
\newblock In {\em Proceedings of the 50th Annual ACM SIGACT Symposium on Theory of Computing}, pages 1319--1332, 2018.

\bibitem[Far98]{farkas1898fourier}
Gyula Farkas.
\newblock A {Fourier}-f{\'e}le mechanikai elv algebrai alapja.
\newblock {\em Math{\'e}matikai {\'e}s Term{\'e}szettudom{\'a}nyi Ertesito}, 16:361--364, 1898.

\bibitem[Fil19]{filtser2018steiner}
Arnold Filtser.
\newblock {Steiner} point removal with distortion {$O(\log k)$} using the relaxed-{Voronoi} algorithm.
\newblock {\em {SIAM} J. Comput.}, 48(2):249--278, 2019.

\bibitem[Fil24]{filtser2020scattering}
Arnold Filtser.
\newblock Scattering and sparse partitions, and their applications.
\newblock {\em {ACM} Trans. Algorithms}, 20(4):30:1--30:42, 2024.

\bibitem[GHP20]{goranci2020improved}
Gramoz Goranci, Monika Henzinger, and Pan Peng.
\newblock Improved guarantees for vertex sparsification in planar graphs.
\newblock {\em SIAM Journal on Discrete Mathematics}, 34(1):130--162, 2020.

\bibitem[Gup01]{gupta2001steiner}
Anupam Gupta.
\newblock {Steiner} points in tree metrics don't (really) help.
\newblock In {\em Proceedings of the twelfth annual ACM-SIAM symposium on Discrete algorithms}, pages 220--227. Society for Industrial and Applied Mathematics, 2001.

\bibitem[HL22]{hershkowitz20211}
D.~Ellis Hershkowitz and Jason Li.
\newblock {$O(1)$} {Steiner} point removal in series-parallel graphs.
\newblock In Shiri Chechik, Gonzalo Navarro, Eva Rotenberg, and Grzegorz Herman, editors, {\em 30th Annual European Symposium on Algorithms, {ESA} 2022}, volume 244 of {\em LIPIcs}, pages 66:1--66:17. Schloss Dagstuhl - Leibniz-Zentrum f{\"{u}}r Informatik, 2022.

\bibitem[KKN15]{kamma2015cutting}
Lior Kamma, Robert Krauthgamer, and Huy~L Nguyen.
\newblock Cutting corners cheaply, or how to remove {Steiner} points.
\newblock {\em SIAM Journal on Computing}, 44(4):975--995, 2015.

\bibitem[KNZ14]{krauthgamer2014preserving}
Robert Krauthgamer, Huy~L Nguyen, and Tamar Zondiner.
\newblock Preserving terminal distances using minors.
\newblock {\em SIAM Journal on Discrete Mathematics}, 28(1):127--141, 2014.

\bibitem[OS81]{okamura1981multicommodity}
Haruko Okamura and Paul~D Seymour.
\newblock Multicommodity flows in planar graphs.
\newblock {\em Journal of Combinatorial Theory, Series B}, 31(1):75--81, 1981.

\end{thebibliography}
